%% file: main.tex
\documentclass[11pt]{article}

\input{preamble/math_commands}

\input{preamble/algorithm}
\input{preamble/preamble}

\input{preamble/acronyms}

\usepackage{setspace}

\usepackage[affil-it]{authblk}

\title{Aggregated Posterior Predictive Checks\\ for Generative Modeling}
\date{\today}          

\author[1]{Shweta Dutta}
\author[1]{Gemma Moran}
\affil[1]{Department of Statistics, Rutgers University}

\begin{document}

\maketitle

\begin{abstract}
\input{abstract}

\end{abstract}

\setstretch{1.2}

\input{1-intro}

\input{2-methods-alt}

\input{3-theory}
\input{4-experiments}

\input{5-discussion}

\bibliographystyle{apalike}
\bibliography{bib/references}

\appendix
\input{appendix}

\end{document}

%% file: preamble/math_commands.tex
\usepackage{amsmath,amsfonts,bm}
\usepackage{dsfont}

\def\1{\bm{1}}

\def\mX{{\bm{X}}}

\DeclareMathAlphabet{\mathsfit}{\encodingdefault}{\sfdefault}{m}{sl}
\SetMathAlphabet{\mathsfit}{bold}{\encodingdefault}{\sfdefault}{bx}{n}

\newcommand{\E}{\mathbb{E}}

\newcommand{\R}{\mathbb{R}}

\newcommand{\Var}{\mathrm{Var}}

\newcommand{\train}{\mathrm{train}}

\newcommand{\new}{\mathrm{new}}
\newcommand{\rep}{\mathrm{rep}}
\newcommand{\agg}{\mathrm{agg}}

\newcommand{\obs}{\mathrm{obs}}
\newcommand{\col}{\mathrm{col}}

%% file: preamble/algorithm.tex
\usepackage[algoruled,ruled,vlined]{algorithm2e}
\usepackage{listings}
\usepackage{fancyvrb}
\fvset{fontsize=\normalsize}

\makeatletter
\renewcommand{\SetKwInOut}[2]{%
  \sbox\algocf@inoutbox{\KwSty{#2}\algocf@typo:}%
  \expandafter\ifx\csname InOutSizeDefined\endcsname\relax
    \newcommand\InOutSizeDefined{}\setlength{\inoutsize}{\wd\algocf@inoutbox}%
    \sbox\algocf@inoutbox{\parbox[t]{\inoutsize}{\KwSty{#2}\algocf@typo:\hfill}~}\setlength{\inoutindent}{\wd\algocf@inoutbox}%
  \else
    \ifdim\wd\algocf@inoutbox>\inoutsize%
    \setlength{\inoutsize}{\wd\algocf@inoutbox}%
    \sbox\algocf@inoutbox{\parbox[t]{\inoutsize}{\KwSty{#2}\algocf@typo:\hfill}~}\setlength{\inoutindent}{\wd\algocf@inoutbox}%
    \fi%
  \fi
  \algocf@newcommand{#1}[1]{%
    \ifthenelse{\boolean{algocf@inoutnumbered}}{\relax}{\everypar={\relax}}%
    {\let\\\algocf@newinout\hangindent=\inoutindent\hangafter=1\parbox[t]{\inoutsize}{\KwSty{#2}\algocf@typo:\hfill}~##1\par}%
    \algocf@linesnumbered
  }}%
\makeatother

\SetKwInOut{KwInput}{input}
\SetKwInOut{KwOutput}{output}

%% file: preamble/preamble.tex
\usepackage[T1]{fontenc}
\usepackage{amsmath}
\usepackage{amsfonts}
\usepackage{amscd}
\usepackage{amssymb}
\usepackage{amsthm}
\usepackage{bm}
\usepackage{thm-restate}
\usepackage{float}
\usepackage{graphicx}
\usepackage{placeins}

\usepackage{enumerate}
\usepackage{enumitem}
\usepackage{tabularx}
\usepackage[usenames,dvipsnames]{xcolor}

\usepackage{url}

\usepackage{array}
\usepackage{booktabs,multirow,multicol}       

\usepackage[]{natbib}

\usepackage[colorlinks,linktoc=all]{hyperref}
\hypersetup{citecolor=blue}
\hypersetup{urlcolor=MidnightBlue}
\hypersetup{linkcolor=Blue}

\usepackage[nameinlink]{cleveref}
\creflabelformat{equation}{#1#2#3}
\crefname{equation}{eq.}{eqs.}
\Crefname{equation}{Eq.}{Eqs.}
\Crefname{section}{\S}{\S}
\Crefname{assumption}{Assumption}{Assumptions}

\usepackage[parfill]{parskip}

\DeclareRobustCommand{\parhead}[1]{\textbf{#1}~}

\usepackage[
  paper  = letterpaper,
  left   = 1.25in,
  right  = 1.25in,
  top    = 1.0in,
  bottom = 1.0in,
  ]{geometry}

\numberwithin{equation}{section}


\usepackage{caption}
\usepackage{subcaption}
\newtheorem{theorem}{Theorem}
\newtheorem{corollary}[theorem]{Corollary}
\newtheorem{definition}{Definition}

\newtheorem{result}{Result}
\newtheorem{assumption}{Assumption}

\newtheorem{proposition}{Proposition}

\newtheorem{remark}{Remark}

\newcommand{\defeq}{\mathrel{\mathop:}=}

\usepackage[most]{tcolorbox}
\newtcolorbox{longexample}[1][]{
  enhanced,
  breakable,
  colback=gray!5,
  colframe=black!60,
  boxrule=0.6pt,
  arc=2pt,
  left=8pt,
  right=8pt,
  top=8pt,
  bottom=8pt,
  before skip=1em,
  after skip=1em,
  fonttitle=\bfseries,
  title=Example,
  #1
}

\usepackage
[acronym,nowarn,section,nogroupskip,nonumberlist]{glossaries}
\setacronymstyle{long-sc-short}
\glsdisablehyper{}
\setglossarystyle{long3col}
\makeglossaries

%% file: preamble/acronyms.tex
\newacronym{HPC}{hpc}{holdout predictive check}
\newacronym{PPC}{ppc}{posterior predictive check}
\newacronym{HAPC}{hapc}{holdout aggregated posterior check}
\newacronym{APPC}{appc}{aggregated posterior predictive check}

%% file: abstract.tex

\glsresetall

Latent variable generative models are commonly fit using simple priors over latent variables, but draws from these priors often fail to produce realistic data. This failure is due to a mismatch between the prior and the \emph{aggregated posterior}, the distribution of latent variables induced by the fitted model and the data.  This mismatch is often viewed as evidence that the prior is misspecified and should be replaced. Alternatively, in modern generative models, a two-stage strategy is increasingly used where first, the model is fit, and second, the aggregated posterior is estimated \citep{van2017neural,rombach2022high}.  Synthetic data are then obtained by sampling from this aggregated posterior instead of the prior. To check such procedures, we introduce the \gls{APPC}. Theoretically, we establish sufficient conditions under which the \gls{APPC} is asymptotically calibrated. For probabilistic principal component analysis,  we show that the \gls{APPC} can remain calibrated under a misspecified latent prior when pervasive factors permit recovery of the signal space.  Experiments with variational autoencoders show that aggregated posterior sampling improves generation for heavy-tailed and clustered data relative to Gaussian prior sampling while performing comparably to models with more flexible latent priors.

%% file: 1-intro.tex
\glsresetall

\section{Introduction}

Latent variable generative models (LVGMs), both classical and deep variants, are widely used for capturing hidden or low-dimensional structure in observed data. LVGMs approximate an unknown data distribution $P_0$, given observed samples $\{\bm{x}_i\}_{i=1}^n$, $\bm{x}_i\in\mathcal{X}$. A general LVGM proceeds by introducing a latent variable $\bm{z}_i$ with prior $p(\bm{z})$, and specifying a conditional likelihood $p_{\theta}(\bm{x}_i|\bm{z}_i)$, yielding the marginal likelihood $p_{\theta}(\bm{x}) = \int p_{\theta}(\bm{x}|\bm{z})\,p(\bm{z})d\bm{z}$. The goal is to estimate $\theta$ such that  $P_{\theta}$ approximates the unknown data distribution $P_0$. Classical models such as probabilistic principal component analysis and linear factor analysis fit into this paradigm, as well as modern deep generative models including variational autoencoders (VAEs) \citep{kingma2013auto, rezende2014stochastic}, normalizing flows \citep{rezende2015variational}, and diffusion models \citep{sohl2015deep,  ho2020denoising}. 

After fitting an LVGM and estimating $\widehat{\theta}$, a conventional strategy for generating synthetic data is to sample from the prior $\bm{z}_i \sim p(\bm{z})$ and then sample $\bm{x} \sim p_{\widehat{\theta}}(\bm{x}|\bm{z}_i)$.  However, in many generative models, this sampling process does not yield realistic data samples \citep{rezende2018taming,li2021low}. One reason for this is the mismatch between the prior and the \emph{aggregated posterior}. The aggregated posterior is the distribution of latent variables induced by a model: $q_{\mathrm{agg},\theta}(\bm{z}) = \int p_{{\theta}}(\bm{z}|\bm{x})\,P_0(d\bm{x}).$
A mismatch between $q_{\mathrm{agg},\theta}$ and the prior $p(\bm{z})$ is often interpreted as evidence that the prior is misspecified and should be replaced \citep{seth2019model}. 

Instead of replacing the prior, an increasingly common alternative is a two-stage sampling strategy. In the first stage,  a latent variable model is fit with a simple prior (e.g. standard Gaussian) in order to learn a useful representation and likelihood parameters. In the second stage, 
the aggregated posterior is estimated and then used for latent sampling. For example, latent diffusion models estimate a  sampling distribution over learned latent variables \citep{rombach2022high}. That is, these approaches replace prior sampling by sampling from an estimated aggregated posterior.  

Aggregated posterior sampling is appealing as it separates two competing goals in LVGMs: representation learning, and sampling synthetic data. For example, diffusion models are state-of-the-art for sampling, but do not learn a unique latent representation $\bm{z}_i$ for each data point \citep{mittal2023diffusion, fuest2026diffusion}. Meanwhile, in VAEs, the aggregated posterior $q_{\mathrm{agg}}$ often has meaningful structure, but it does not equal the prior \citep{kingma2013auto, locatello2020weakly}. Aggregated posterior sampling also provides an alternative to Gaussian
prior sampling, which yields light-tailed outputs under Lipschitz decoders \citep{tam2025statistical}.

Although aggregated posterior sampling strategies are becoming more common, there are limited tools for checking their adequacy. In this paper, we introduce the \gls{APPC}, a predictive check for evaluating the aggregated posterior predictive distribution. The \gls{APPC} extends holdout predictive checks \citep{moran2024a,li2026calibrated} to the setting where synthetic observations are generated by first sampling from an estimated aggregated posterior and then passing these latent variables through a fitted likelihood. Instead of asking whether the aggregated posterior equals the prior, the \gls{APPC} asks whether the resulting two-stage generative procedure produces observations that are consistent with held-out data.

We make three contributions. First, we formulate the \gls{APPC} and distinguish it from alternative predictive checks. Second, we derive sufficient conditions for calibration of \gls{APPC} $p$-values, distinguishing the roles of model misspecification, parameter-estimation, variational, and aggregation error rates. As a case study, we show that the \gls{APPC} for probabilistic principal component analysis (PPCA) can remain calibrated under a misspecified latent prior when pervasive factors allow recovery of the signal space. Third, experiments with VAEs show that aggregated posterior sampling can improve generation of heavy-tailed and clustered data relative to Gaussian prior sampling and can perform comparably to methods with more flexible latent priors.

\textbf{Notation.} Bold lowercase symbols denote individual random vectors, while bold uppercase symbols denote collections of such vectors arranged as matrices. For example, \(\bm z_i\) and \(\bm x_i\) denote individual draws, whereas \(\bm Z=[\bm z_1,\ldots,\bm z_N]^\top\) and \(\bm X=[\bm x_1,\ldots,\bm x_N]^\top\) denote the corresponding datasets.

\subsection{Related work}
\label{sec:Related Work}

\parhead{Two-stage sampling strategies.} 
There have been a number of papers which use a two-stage approach to latent generation. First, a latent variable model is fit with a simple working prior on the latent variables. Second, a different model is fit to the encoded latent representations from the first stage. Examples include two-stage VAEs \citep{dai2019diagnosing}, latent diffusion \citep{rombach2022high}, and VQ-VAE pipelines \citep{van2017neural,ramesh2021zero}.

\textbf{Empirical Bayes.} In empirical Bayes $g$-modeling, the goal is to estimate the prior density \citep{efron2019bayes}  by maximizing the marginal likelihood.
In the context of VAEs, various approaches constrain priors to a class of parameterized densities; the likelihood and prior parameters are jointly estimated. Examples include: a Gaussian prior with learnable mean and variance; mixture priors \citep{falck2021multi}; and the VampPrior \citep{tomczak2018vamp}. 

The setting we consider may be viewed as ``two-stage Empirical Bayes'': in the first stage, the likelihood parameters are learned, and in the second, the latent prior is learned. A related work is \citet{shen2026simulation} who iteratively refine an empirical Bayes prior in the context of simulation-based inference.

\parhead{Bayesian model checking.} Predictive checks assess whether a fitted model can generate data that ``looks like'' the observed data. The \gls{PPC} \citep{Rubin:1984, gelman1996posterior} evaluates models by comparing observed data to samples from the posterior predictive distribution. While \gls{PPC}s have been widely adopted, they suffer from calibration issues due to the reuse of data for both fitting and evaluation \citep{bayarri2000p, robins2000asymptotic}. Alternative approaches include calibrated predictive checks that compare posterior predictive samples to holdout data \citep{moran2024a, li2026calibrated}. 

\cite{seth2019model} check aggregated posterior samples against the prior to detect prior misspecification. The \gls{APPC} focuses on a different question, asking whether the aggregated posterior yields data that is consistent with a holdout dataset.

Deep generative models are typically evaluated through heuristic metrics such as the Fr\'{e}chet Inception Distance \citep[FID,][]{heusel2017gans} and the Inception Score \citep[IS,][]{salimans2016improved}. These metrics can be useful, but do not directly test whether the generative distribution is statistically consistent with observed data \citep{Theis2016a}. The \gls{APPC} provides a calibrated predictive check for the observational distribution induced by aggregated posterior sampling.

%% file: 2-methods-alt.tex

\glsresetall

\section{The aggregated posterior predictive check}\label{sec:methods}

\subsection{Bayesian predictive checks}

Consider an observed dataset $\bm{X}^{\obs}=[\bm{x}^\obs_1, \dots, \bm{x}_N^\obs]^\top$, $\bm{x}_i^{\obs}\in\mathcal{X}$, and a model with global parameters $\theta \in\Theta$ with  prior $p(\theta)$:
\begin{align*}
p(\bm{X}^\obs, \theta) = p(\theta) \prod_{i=1}^N p(\bm{x}_i^\obs\mid\theta).
\end{align*}
Bayesian predictive checks ask whether a model generates data which resembles a real dataset, as summarized by a diagnostic function. Predictive checks have three ingredients:
(i) a reference distribution
$Q_N\in\mathcal{P}(\mathcal{X}^N)$, possibly constructed from data, for generating synthetic data replicates;
(ii) a diagnostic function $d_N:\mathcal{X}^N\to\mathbb{R}$; and
(iii) a real dataset used for checking, $\bm{X}^{\mathrm{check}}\in\mathcal{X}^N$.

The one-sided predictive check $p$-value is
\begin{align}
\mathbb{P}\bigl\{d_N(\bm{X}^{\rep})\geq d_N(\bm{X}^{\mathrm{check}})
\bigm|\mathcal{F}_Q,\,\bm{X}^{\mathrm{check}}\bigr\},
\qquad \bm{X}^{\rep}\sim Q_N(\cdot|\mathcal{F}_Q),
\label{eq:general-predictive-check}
\end{align}
where $\mathcal{F}_Q$ denotes the sigma-field generated by the data used to construct $Q_N$.

The \gls{PPC} \citep{Rubin:1984,gelman1996posterior} takes the reference $Q_N(\cdot|\mathcal{F}_Q)$ to be the posterior predictive $p(\bm{X}^\rep|\bm{X}^{\obs})$ where $\bm{X}^\rep \in \mathcal{X}^N$ denotes replicated data and $\mathcal{F}_Q=\sigma(\bm{X}^\obs)$, with $\sigma(\cdot)$ denoting the sigma-field.  For the \gls{PPC}, the checking dataset is also $\bm{X}^\obs$.  Consequently, the \gls{PPC} uses the data twice, resulting in uncalibrated $p$-values \citep{robins2000asymptotic}.

The \gls{HPC} keeps $Q_N(\cdot|\mathcal{F}_Q)$ as the posterior predictive, but checks against an independent
holdout dataset, $\bm{X}^{\new}$; under standard regularity conditions, the \gls{HPC}
 is calibrated \citep{moran2024a,li2026calibrated}.

\subsubsection{Local latent variable models} \label{sec:local-latent-references}

Now consider a model which additionally has local latent variables $\bm{z}_i \in \mathcal{Z}$: 
\begin{align*}
    \bm{z}_i \sim p(\bm{z}_i), \qquad
\bm{x}_i \mid \bm{z}_i
\sim p_{{\theta}}(\bm{x}_i \mid \bm{z}_i).
\end{align*}
There are distinct ways to generate replicates, corresponding to different reference distributions. For simplicity, we now treat the global parameter $\theta$ as fixed. 

\textbf{Local prior replicates.}
Under the original latent variable model,  a new observation is generated by first sampling from the prior, and then sampling the observation conditional on this latent:
\begin{align}
\bm{z}_i^\rep \sim p(\bm{z}), \qquad \bm{x}_i^\rep \mid \bm{z}_i^\rep
\sim p_{{\theta}}(\bm{x}_i^\rep \mid \bm{z}_i^\rep).
\label{eq:prior-generation}
\end{align}
A holdout predictive check based on \Cref{eq:prior-generation} evaluates whether this prior generative process produces new independent observations consistent with a holdout dataset. 

\textbf{Conditional per-sample replicates.} An alternative is to sample $\bm{z}_i^\rep$ from the latent posterior, conditional on an observed $\bm{x}_i^\obs$:
\begin{align}
\bm{z}_i^\rep \mid \bm{x}_i^\obs
\sim p_{{\theta}}(\bm{z}_i^\rep \mid \bm{x}_i^\obs), \qquad
\bm{x}_i^\rep \mid \bm{z}_i^\rep
\sim p_{{\theta}}(\bm{x}_i^\rep \mid \bm{z}_i^\rep).
\label{eq:conditional-replication}
\end{align}
This procedure can be useful for cases like topic modeling, where we can split \emph{within} a sample (i.e. document) \citep{moran2024a}. In general, however, \Cref{eq:conditional-replication} does not yield calibrated predictive checks, as we demonstrate in \Cref{subsec:ppca}.

\paragraph{Aggregated posterior replicates.}
Many modern pipelines deploy a different generative procedure. After fitting a latent variable model, they estimate a distribution over inferred latent representations, $\widehat{q}_{\agg}$, and use this learned latent distribution for generation (e.g. \citet{dai2019diagnosing}). This procedure is
\begin{align}
\bm{z}_i^\rep \sim \widehat{q}_{\agg}(\bm{z}^\rep), \qquad
\bm{x}_i^\rep \mid \bm{z}_i^\rep
\sim p_{{\theta}}(\bm{x}_i^\rep \mid \bm{z}_i^\rep).
\label{eq:aggregated-generation}
\end{align}
Our proposed \gls{APPC} is a predictive check where the reference distribution is \Cref{eq:aggregated-generation}, checked against an independent holdout dataset.

\subsection{Aggregated posterior predictive checks}\label{subsec:appc-methods}

We now describe the aggregated posterior predictive check in more detail. We denote the marginal distribution of $\bm{x}_i$ as $P_{\theta}$; this is not necessarily equal to the true data generating distribution $P_0$.

\subsubsection{The population reference distribution}
\label{subsec:oracle-APPC}

We first consider the population reference distribution of the \gls{APPC}. This reference distribution uses oracular knowledge of both $P_0$ and the posterior of the latent variables.  We denote $\theta^\star \in \Theta$ as an optimal parameter (defined precisely in \Cref{sec:theory}). 

The population aggregated posterior induced by  $\theta^\star$, the posterior, and $P_{0}$ is
\begin{align}
q_{\theta^\star,\agg}(\bm{z})
= \int
p_{\theta^\star}(\bm{z} \mid \bm{x})
\, P_{0}(d\bm{x}).
\label{eq:oracle-aggregated-posterior}
\end{align}
The corresponding aggregated posterior predictive is
\begin{align}
p_{\theta^\star,\agg}(\bm{x}^\rep)
= \int
p_{\theta^\star}(\bm{x}^\rep \mid \bm{z})
\, q_{\theta^\star,\agg}(\bm{z})\,d\bm{z}.
\label{eq:aggregated-posterior-predictive}
\end{align}
The oracle \gls{APPC} is the predictive check which uses as a reference distribution this population aggregated posterior predictive (\Cref{eq:aggregated-posterior-predictive}). The checking dataset is an independent holdout dataset $\bm{X}^\new$.

\vspace{1em}

\begin{definition}[Oracle aggregated posterior predictive check]
\label{def:oracle-APPC}
Let $d_{N}:\mathcal{X}^{N}\to\mathbb{R}$ be a diagnostic statistic. The oracle \gls{APPC} $p$-value is
\begin{align}
p_{\mathrm{APPC}}^{\mathrm{oracle}}
\defeq
\mathbb{P}
\left\{
d_N(\bm{X}^\rep)
\geq
d_N(\bm{X}^\new) \;\middle|\;
\bm{X}^\new
\right\},
\label{eq:oracle-APPC-pvalue}
\end{align}
where $\bm{X}^\rep=\{\bm{x}_i^\rep\}_{i=1}^{N_\new}$ is generated independently according to
\begin{align}
\bm{z}_i^\rep
\sim
q_{\theta^\star,\agg}(\bm{z}^\rep_i), \qquad
\bm{x}_i^\rep \mid \bm{z}_i^\rep
\sim
p_{\theta^\star}(\bm{x}_i^\rep \mid \bm{z}_i^\rep),
\qquad i=1,\dots,N_\new.
\label{eq:oracle-APPC-generator}
\end{align}
\end{definition}
For a well-specified model $P_0 = P_{\theta^\star}$, the population aggregated posterior predictive equals the true data distribution because $q_{\theta^\star,\agg}(\bm{z})=p(\bm{z})$. In contrast, under model misspecification, $q_{\theta^\star,\agg}(\bm{z})$ need not equal $p(\bm{z})$.  However, $P_{\theta^\star,\agg}$ may still produce samples which resemble samples from $P_{0}$; this is what the \gls{APPC} assesses.

\subsection{APPC with an estimated aggregated posterior}\label{sec:APPC-estimated}

In practice, $q_{\theta,\agg}$ is unknown because it depends on the unknown data distribution $P_0$. We now consider a practical version of the \gls{APPC} where we approximate $P_0$.

Specifically, we partition the available data into training, aggregation, and holdout sets indexed by $I_\train$, $I_\agg$, and $I_\new$, which form a partition of $\{1,\dots, N\}$:
\begin{align*}
\bm{X} = (\bm{X}^\train, \bm{X}^\agg, \bm{X}^\new),
\qquad
\bm{X}^s = \{\bm{x}_i : i \in I_s\}, \qquad s \in \{\train,\agg,\new\}.
\end{align*}
Here, $\bm{X}^\train$ is used to estimate the first-stage model parameters $\widehat{\theta}$, and variational parameters $\widehat{\phi}$ where applicable; $\bm{X}^\agg$ is used for an empirical approximation of $P_0$; and $\bm{X}^\new$ is used for checking in the \gls{APPC}.

\textbf{Empirical aggregated posterior.} A sample from the empirical aggregated posterior is obtained by sampling an index $l$ uniformly from $\{1,\dots, N_{\agg}\}$, then sampling $\bm{z}\sim p_{\widehat{\theta}}(\bm{z}\mid\bm{x}_l^\agg)$. That is, the empirical aggregated posterior is:
\begin{align}
\widehat{q}_{\widehat{\theta},\mathrm{agg-emp}}(\bm{z}) \defeq \frac{1}{N_{\agg}}\sum_{i=1}^{N_{\agg}} p_{\widehat\theta}(\bm{z}\mid \bm{x}_i^\agg).
\label{eq:encoder-empirical-aggregation-notvi}
\end{align}
For models such as VAEs, the exact posterior is generally unavailable. In this case, a fitted encoder $q_{\widehat{\phi}}(\bm{z}\mid\bm{x})$ replaces $p_{\widehat\theta}(\bm{z}\mid \bm{x})$ in \Cref{eq:encoder-empirical-aggregation-notvi}. 

An alternative is to directly fit a density to samples $\bm{z}_i\sim q_{\widehat\phi}(\bm{z}_i|\bm{x}_i^\agg)$, $i=1,\dots, N_\agg$. Examples of density estimators include a Gaussian mixture model, normalizing flow, or diffusion model \citep{rombach2022high}. 

\vspace{1em}

\begin{definition}[Aggregated posterior predictive check]
\label{def:APPC} 
    Consider datasets $\bm{X}^\train$, $\bm{X}^\agg$, $\bm{X}^\new$  drawn independently from $P_0$. For a diagnostic  $d_N:\mathcal{X}^N\to\mathbb{R}$, the \textit{aggregated posterior predictive check} $p$-value is:
    \begin{align}
        p_{\mathrm{APPC}}
:= \mathbb{P} \left\{d_N(\bm{X}^{\rep})\ \ge d_N(\bm{X}^{\new})\;\middle|\;  \bm{X}^{\train}, \bm{X}^{\agg}, \bm{X}^{\new}\right\},
    \label{eq:APPC-pval}
    \end{align}
        where 
$\bm{z}_i^{\rep}\sim \widehat{q}_{\widehat{\phi},\agg}(\bm{z}_i^\rep)$, and $\bm{x}_i^\rep \sim p_{\widehat{\theta}}(\bm{x}_i^\rep\mid\bm{z}_i^\rep)$ are draws from the aggregated posterior predictive distribution. Here, $\widehat{\theta}$ is estimated using $\bm{X}^\train$.
\end{definition}

To calculate the \gls{APPC} $p$-value, we use Monte Carlo sampling: see \Cref{alg:APPC-1}.

\vspace{1em}

\begin{remark}
In \Cref{eq:APPC-pval}, we treat the estimated model parameters $\widehat{\theta}$ as fixed. The \gls{APPC} also supports sampling from the posterior distribution of $\theta$ (see \Cref{app:bayesian-appc}).
\end{remark}

\begin{algorithm} 
\caption{Aggregated posterior predictive check}
\label{alg:APPC-1}
\KwIn{data $\bm{X} = (\bm{X}^\train, \bm{X}^\agg, \bm{X}^\new)$, model family $p_{\theta}(\bm{x}, \bm{z})$, (if applicable: variational family with parameters $\phi$), diagnostic $d$, $\#$ of replicates $R$}
\KwOut{\gls{APPC} $p$-value}
\vspace{0.5em}
\begin{enumerate}[leftmargin=*]
    \item Fit model parameters $\theta$ (and variational parameters $\phi$ if applicable) using $\bm{X}^\train$\;
    \item Estimate the aggregated posterior $\widehat{q}_{\widehat{\phi},\agg}(\bm{z})$ using $\bm{X}^\agg$;
\end{enumerate}
\vspace{0.5em}
\For{$r = 1$ ,\dots, $R$}{
\For{$i=1,\dots, N_\new$}{
    draw samples from the (estimated) aggregated posterior $\bm{z}_i^{(r)} \sim \widehat{q}_{\widehat{\phi},\agg}(\bm{z})$\;
    generate aggregated posterior predictive samples  $\bm{x}^{(r),\rep}_i \sim p_{\widehat{\theta}}(\bm{x}\mid\bm{z}_i^{(r)})$\;}
}
\vspace{0.5em}
Compute the empirical \gls{APPC} $p$-value:
\begin{align*}
\widehat{p}_{\mathrm{APPC}} = \frac{1}{R}\sum_{r=1}^R \mathbf{1}\left [ d(\bm{X}^{(r),\rep}) \geq d(\bm{X}^\new) \right];
\end{align*}

\Return $\widehat{p}_{\mathrm{APPC}}$

\end{algorithm}

\subsection{Example: Probabilistic PCA}
\label{subsec:ppca}

We now illustrate the \gls{APPC} by considering a classic latent variable model: probabilistic principal component analysis \citep[PPCA,][]{tipping1999probabilistic}.  PPCA assumes a latent Gaussian prior, but is commonly used to summarize high-dimensional data with underlying structure (e.g. clusters). That is, in PPCA, we often expect the prior to be incorrect; can the aggregated posterior still generate realistic samples?  

We consider two cases: (i) the well-specified case where the latent prior is correct, and (ii) the misspecified prior case where the latent prior is incorrect. We show:
\begin{itemize}
    \item In the well-specified case, generating latent variables from both the prior and the empirical aggregated posterior yield data consistent with holdout data;
    \item In the misspecified case, samples from the prior do not yield data consistent with holdout data, while samples from the empirical aggregated posterior do.
\end{itemize}
In section \Cref{subsec:ppca-theory}, we provide theoretical corroboration of these empirical results.

\textbf{Model.}
The observed data is $\bm{X} = \{\bm{x}_i\}_{i=1}^N$ where $\bm{x}_i \in \mathbb{R}^{D}$ with associated  latent  $\bm{z}_i \in \mathbb{R}^K,\, K < D$.  The PPCA model is:
\begin{align}
\bm{z}_i \sim \mathcal{N}\left(0,\, I_K\right), \qquad \bm{x}_i \mid \bm{z}_i \sim \mathcal{N}\left(\bm{W} \bm{z}_i,\, \sigma^2 I_D \right),\qquad \bm{W} \in \mathbb{R}^{D \times K}.
\label{eq:ppca-ws-model}
\end{align}

\begin{figure}[t]
\centering
\begin{subfigure}{0.49\textwidth}
    \centering
    \includegraphics[width=\linewidth]{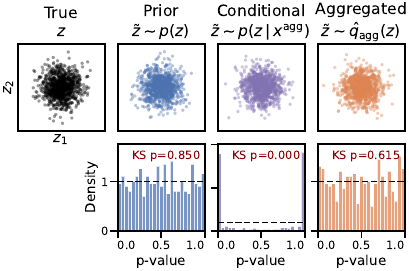}
     \caption{Well-specified case.}
    \label{fig:ppca_summary_ws}
\end{subfigure}
\begin{subfigure}{0.49\textwidth}
    \centering
    \includegraphics[width=\linewidth]{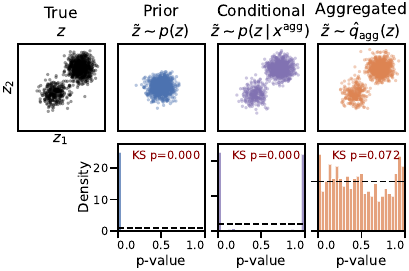}
     \caption{Misspecified case.}
    \label{fig:ppca_summary_ms}
\end{subfigure}
\caption{\textbf{PPCA example.} The empirical aggregated posterior predictive generates realistic samples in both well-specified and misspecified cases. The \gls{APPC} $p$-values are approximately uniform in both cases, unlike prior or conditional-replicate sampling. Top row: example draws of latent variables $\bm{z}$ (up to an invertible linear transformation). Bottom row: $p$-values from predictive checks over 500 trials. Kolmogorov-Smirnov tests of uniformity are in red.}
\end{figure}

{\bf Comparison of predictive distributions.} We split our data into $\bm{X}^\train$, $\bm{X}^\agg$, and $\bm{X}^\new$.  We obtain the MLE of $\theta = (\bm{W}, \sigma^2)$ using $\bm{X}^\train$.

We contrast the different distributions for $\bm{X}^\rep$ from \Cref{sec:local-latent-references}: prior sampling, conditional sampling, and aggregated posterior sampling (see \Cref{app:ppca-predictive-check-pval} for details).

\textbf{Predictive check.} For each distribution for $\bm{X}^\rep$, we conduct a predictive check:
\begin{align}
\mathbb{P}\left\{d(\bm{X}^\rep) \geq d(\bm{X}^\new) \;\middle|\; \bm{X}^\train, \bm{X}^\agg, \bm{X}^\new\right\},
\label{eq:predictive-check-pval}
\end{align}
where we take the diagnostic to be the grand mean $d(\bm{X}) = D^{-1}\bm{1}_D^\top\overline{\bm{x}}$. For each predictive distribution, their corresponding $p$-value formulas are available in \Cref{app:ppca-predictive-check-pval}.

\textbf{Well-specified prior.}  We first consider the well-specified case where the true data generating prior is equal to the PPCA prior. The observed data is generated as in \Cref{eq:ppca-ws-model} with $D=10$, $K=2$, $\bm{W}^\top\bm{W}=9I_K$ where $\bm{W}$ has orthogonal columns, and $\sigma^2=0.1$.  For each reference distribution, we compute the predictive check $p$-values for 500 simulated datasets, each with $N_\train=N_\agg=5000$, $N_\new=1000$. 

For both the latent prior predictive check and the empirical \gls{APPC}, their $p$-value distributions are approximately uniform. Meanwhile, the conditional-replicate predictive check $p$-values are not uniform (\Cref{fig:ppca_summary_ws}). This is because the conditional-replicate predictive does not account for between-sample variance (see \Cref{app:ppca-predictive-check-pval-calibration}).

\textbf{Misspecified prior.} Now consider a misspecified case where the true latent distribution is not the PPCA prior; instead, we take it to be a Gaussian mixture:
\begin{align*}
    \bm{z}_i &\sim \sum_{k=1}^{M}\, \pi_k \, \mathcal{N}\left(\bm{m}_k,\, \bm{S}_k\right), \qquad \sum_{k=1}^{M}\, \pi_k = 1.
\end{align*}
We set $m_1=(3,3)$, $m_2 = (-1,-1)$ and $(\pi_1,\pi_2)=(0.7, 0.3)$. The mixture variances are identity matrices. The likelihood model and parameters remain the same.

We fit the original PPCA model (\Cref{eq:ppca-ws-model}) to this data. The empirical \gls{APPC} $p$-values are approximately uniform (\Cref{fig:ppca_summary_ms}). The aggregated posterior also yields samples that are visually consistent with samples from the true data distribution. In contrast, the other distributions almost always fail their predictive check (\Cref{fig:ppca_summary_ms}).

%% file: 3-theory.tex

\glsresetall

\section{Theory} 
\label{sec:theory}

In this section, we first provide a general characterization of the \gls{APPC} $p$-value. We then provide sufficient conditions for which the \gls{APPC} is calibrated by controlling errors from model misspecification, parameter estimation and empirical approximation of the true data distribution. Finally, we specialize these results to PPCA, proving calibration can hold even when the latent prior is misspecified.

\subsection{Setup and assumptions}

\textbf{Data splitting.} Throughout, $m=1,2,\ldots$ indexes a sequence of experiments. At stage $m$, the mutually independent samples $\bm X_m^s=\{\bm x_{m,i}^s\}_{i=1}^{N_{s,m}}$, $s\in\{\train,\agg,\new\}$, are drawn from $P_{0,m}$, with
$N_{\train,m},N_{\agg,m},N_{\new,m}\to\infty$. The split sizes may grow at different rates, and sample sizes change with $m$. In the PPCA example, the data dimension $D_m$ also grows with $m$. When there is no ambiguity, we suppress $m$ for readability.

\textbf{Model.} We consider the latent variable model:
\begin{align*}
\bm z_i \sim p_Z(\bm z_i), \quad \bm x_i \mid \bm{z}_i \sim p_\theta(\bm x_i \mid \bm z_i).
\end{align*}
The marginal distribution under this model is $P_\theta$; there may be model misspecification in that $P_0\notin\{P_{\theta}: \theta \in \Theta\}$.

\textbf{Reference distributions.} Consider a replicate generated by first sampling from the posterior given an observation $\bm{x}_i$, and then sampling from the likelihood:
\begin{align*}
\bm{z}_i^\rep |\bm{x}_i \sim p_{\theta}(\bm{z}_i^\rep|\bm{x}_i), \quad \bm{x}_i^\rep|\bm{z}_i^\rep \sim p_{\theta}(\bm{x}_i^\rep|\bm{z}_i^\rep).
\end{align*}
We write $K_{\theta,\mathrm{post}}(\bm{x},\cdot)$ for this distribution of $\bm{x}^\rep$ conditional on $\bm{x}$, where ``post'' indicates the posterior $p_{\theta}(\bm{z}|\bm{x})$ is available. That is, for every measurable $A\subseteq \mathcal{X}$, 
\begin{align*}
K_{\theta,\mathrm{post}}(\bm{x}, A) = P_{\theta}(\bm{x}^\rep \in A|\bm{x}) = \int_{\mathcal{Z}} P_{\theta}(\bm{x}^\rep\in A|\bm{z})\,p_{\theta}(d\bm{z}|\bm{x}).
\end{align*}
The oracle aggregated posterior predictive distribution is obtained by
first drawing an observation from the true data distribution, sampling from the posterior, and then sampling from the likelihood. That is, the oracle aggregated posterior predictive is
\begin{align}
\bigl(P_0K_{\theta^\star,\mathrm{post}}\bigr)(A)
=
\int_{\mathcal X}
K_{\theta^\star,\mathrm{post}}(\bm x,A)\,
P_0(d\bm x).
\label{eq:pop-optimal-ap-kernel}
\end{align}
When \(p_\theta(\bm z\mid\bm x)\) is intractable, we allow for an approximate posterior
\(q_\phi(\bm z\mid\bm x)\), where \(\phi\in\mathcal E\). Let
\(\eta=(\theta,\phi)\in\Theta\times\mathcal E\). The corresponding probability kernel is
\begin{align}
K_\eta(\bm x,A)
=
\int_{\mathcal Z}
P_\theta(\bm x^\rep\in A\mid\bm z)\,
q_\phi(d\bm z\mid\bm x).
\label{eq:variational-kernel}
\end{align}
In practice, we use
$K_{\widehat\eta_m}$, where
$\widehat\eta_m=(\widehat\theta_m,\widehat\phi_m)$ is a function of $\bm X_m^\train$.  We approximate
\(P_0\) by the empirical distribution of $\bm{X}_m^\agg$,
\begin{align}
\widehat P_{\agg,m}
=
\frac{1}{N_{\agg,m}}
\sum_{i=1}^{N_{\agg,m}}\delta_{\bm x_{m,i}^\agg},\label{eq:definition-empirical-distribution}
\end{align}
where $\delta_{\bm x}$ denotes the Dirac measure at $\bm x$.
The resulting empirical aggregated posterior predictive distribution is
\(\widehat P_{\agg,m} K_{\widehat\eta_m}\). 

We use $Q_m$ to denote a generic reference distribution (e.g. $Q_m$ can be $P_0K_{\theta^\star,\mathrm{post}}$ or $\widehat{P}_{\agg,m} K_{\widehat{\eta}_m}$); this notation is helpful to state our assumptions more generally. When $Q_m$ depends on data, we use $\mathcal{F}_m$ to denote the sigma-field generated by the data used to construct $Q_m$ (e.g. for $Q_m=\widehat{P}_{\agg,m}K_{\widehat{\eta}_m}$, we have $\mathcal{F}_m=\sigma(\{\bm{X}_m^\train, \bm{X}_m^\agg\})$).

\textbf{Optimal parameter set.} To allow for nonidentified models, such as factor analysis, we define optimizer sets for both the exact posterior and variational cases:
\begin{align*}
\mathcal{S}_{\mathrm{post}}^\star \defeq \operatorname*{arg\,max}_{\theta\in\Theta}
\E_{P_0}
\left[
\ell_{\theta}(\bm x)
\right], \qquad
\mathcal{S}_{\mathrm{VI}}^\star \defeq \operatorname*{arg\,max}_{\eta\in \Theta \times \mathcal{E}}
\E_{P_0}
\left[
\ell_{\mathrm{VI},\eta}(\bm x)
\right],
\end{align*}
where  $\ell_{\theta}$ may be the log-likelihood  and $\ell_{\mathrm{VI},\eta}$ may be the evidence lower-bound (ELBO). 

\vspace{1em}

\begin{assumption}[Predictive equivalence]
\label{ass:optimizer-predictive-equivalence}
All elements of $\mathcal{S}^\star_{\mathrm{post}}$ induce the same $P_0K_{\theta,\mathrm{post}}$, and all elements of $\mathcal{S}^\star_{\mathrm{VI}}$ induce the same $P_0K_{\eta}$.
\end{assumption}

In factor analysis, for example, rotational symmetries produce multiple exact posterior optimizers that induce the same aggregated posterior predictive distribution. 

\textbf{Diagnostics.} We consider additive diagnostic functions $d_N:\mathcal{X}^N\to \R$ where
\begin{align*}
d_N(\bm{X}) = \frac{1}{N}\sum_{i=1}^N h(\bm{x}_i).
\end{align*}

\textbf{Central Limit Theorems (CLTs).} We assume a CLT for the diagnostic, and a conditional CLT for a general reference distribution $Q_m$. Let
\begin{align*}
\nu_0 &= \E_{P_0}[h(\bm{x})],\qquad \tau_0^2=\mathrm{Var}_{P_0}\{h(\bm{x})\} \\
\nu_m &= \E_{Q_m}[h(\bm{x})], \qquad \tau_m^2 = \mathrm{Var}_{Q_m}\{h(\bm{x})\}.
\end{align*}

\begin{assumption}[Holdout diagnostic CLT]\label{ass:diag-clt}
The holdout diagnostic has positive, finite variance and satisfies
\begin{align*}
Z_m \defeq \frac{\sqrt{N_{\new,m}}\{d_{N_{\new,m}}(\bm X_m^\new) - \nu_0\}}{\tau_0} \stackrel{d}{\longrightarrow} \mathcal{N}(0,1).
\end{align*}
\end{assumption}

\begin{assumption}[Reference distribution CLT]
\label{ass:conditional-diagnostic-clt}
Conditional on $\mathcal{F}_m$, the replicate observations $\{\bm{x}_{m,i}^\rep\}_{i=1}^{N_{\new,m}}$ are independent draws from $Q_m$ by construction. We assume
\begin{align*}
\sup_{t\in\mathbb R}
\left|
\Pr
\left(
\frac{
\sqrt{N_{\new,m}}
\{d_{N_{\new,m}}(\bm X_m^\rep)-\nu_m\}
}
{\tau_m}
\leq t
\mid
\mathcal F_m
\right)
-\Phi(t)
\right|
\stackrel{P}{\longrightarrow}0,
\end{align*}
where $\Phi$ denotes the Gaussian CDF. 
\end{assumption}

\subsection{A generic calibration theorem}

For a generic reference distribution $Q_m$, the one-sided predictive $p$-value is
\begin{align*}
p_m(Q_m)
=
\mathbb P\left\{
 d_{N_{\new,m}}(\mX_m^\rep)
 \ge
 d_{N_{\new,m}}(\mX_m^\new)
 \;\middle|\;
 \mathcal F_m,\mX_m^\new
\right\}.
\end{align*}

\begin{theorem}[Generic APPC $p$-value]
\label{prop:generic-appc-calibration}
Under \Cref{ass:diag-clt,ass:conditional-diagnostic-clt},
\begin{align*}
p_m(Q_m) = 1 - \Phi\left( s_m^{-1} Z_m + b_m\right) +o_P(1), 
\end{align*}
where $b_m= \sqrt{N_{\new,m}}(\nu_0 - \nu_m)/\tau_m$ and $s_m = \tau_m/\tau_0$
are centering and scaling terms.
Consequently, as $b_m\stackrel{P}{\longrightarrow} b$ and $s_m\stackrel{P}{\longrightarrow} s>0$,
\begin{align*}
p_m(Q_m) \stackrel{d}{\longrightarrow} 1 - \Phi\left(s^{-1}Z + b\right), \quad Z\sim \mathcal{N}(0,1).
\end{align*}
\end{theorem}

\begin{corollary}[Calibration]\label{cor:generic}
Under \Cref{ass:diag-clt,ass:conditional-diagnostic-clt}, if $b_m\stackrel{P}{\to} 0$ and $s_m\stackrel{P}{\to}1$, then $p_m(Q_m)\stackrel{d}{\longrightarrow}U[0,1].$
\end{corollary}

Thus, a convenient sufficient condition for calibration is that the scaled reference distribution matches the diagnostic mean to $o(N_{\new,m}^{-1/2})$ and the diagnostic standard deviation to relative $o(1)$ error. This condition does not require the latent prior to be well-specified.

\vspace{1em}

\begin{remark}
\label{bayesian-remark}
A fully Bayesian approach where $\theta$ is random results in additional terms corresponding to posterior variability. See \Cref{prop:bayesian-appc-calibration} in \Cref{app:bayesian-appc} for details. 
\end{remark}

\subsection{Successive approximations to the oracle reference}

\subsubsection{Oracle APPC}\label{sec:pop-optimal-theory}
We first consider the oracle \gls{APPC}. In this case, the replicates $\bm{X}_m^\rep$ have the reference distribution $Q_m=P_0 K_{\theta^\star,\mathrm{post}}$ where $\theta^\star \in \mathcal{S}^\star_{\mathrm{post}}$.

For the oracle \gls{APPC}, the centering and scaling terms in \Cref{prop:generic-appc-calibration} are
\begin{align*}
b(P_0 K_{\theta^\star,\mathrm{post}})=\frac{\sqrt{N_{\new,m}} ( \nu_0-\nu^\star_{\mathrm{post}})}{\tau^\star_{\mathrm{post}}},
\qquad 
s(P_0 K_{\theta^\star,\mathrm{post}})= \frac{\tau^\star_{\mathrm{post}}}{\tau_0},
\end{align*}
where $\nu_{\mathrm{post}}^\star = \E_{\bm{x}\sim P_0 K_{\theta^\star,\mathrm{post}}}[h(\bm{x})]$, $(\tau_{\mathrm{post}}^\star)^2 = \mathrm{Var}_{\bm{x}\sim P_0 K_{\theta^\star,\mathrm{post}}}\{h(\bm{x})\}$. 

The centering term captures model misspecification under this diagnostic. 

\vspace{1em}

\begin{corollary}
\label{cor:oracle-appc-calibration}
Suppose \Cref{ass:diag-clt,ass:conditional-diagnostic-clt} hold with $Q_m=P_0K_{\theta^\star,\mathrm{post}}$. 
If $b(P_0 K_{\theta^\star,\mathrm{post}})\stackrel{P}{\to} 0$ and $s(P_0 K_{\theta^\star,\mathrm{post}}) \stackrel{P}{\to} 1$, then the oracle \gls{APPC} is calibrated.
\end{corollary}

\vspace{1em}

\begin{corollary}[Well-specified model]\label{cor:oracle-ws-calibration}
Suppose \Cref{ass:diag-clt,ass:conditional-diagnostic-clt} hold with $Q_m=P_0K_{\theta^\star,\mathrm{post}}$.    If $P_{\theta^\star} = P_0$ then $P_0K_{\theta^\star,\mathrm{post}}=P_0$ and the oracle \gls{APPC} is calibrated.
\end{corollary}

Calibration can also hold under model misspecification. For example, consider a noiseless latent variable model $\bm{x} = f_{\theta}(\bm{z})$. Even when the assumed latent prior differs from the true latent distribution, if the encoder is the exact inverse, $e_{\phi}=f_{\theta}^{-1}$, then $f_{\theta}\left\{e_{\phi}(\bm{x})\right\} = \bm{x}$ almost surely and $P_0K_{\eta}=P_0$.  As another example, in \Cref{subsec:ppca-theory}, we prove that PPCA with a misspecified latent prior can have a calibrated \gls{APPC}.

In both examples, a sufficient condition is recovery of the latent signal: in the autoencoder case via the true encoder and in the PPCA case via pervasiveness \citep{bai2003inferential}. This gives insight into when the aggregated posterior predictive is a good choice: when the signal-to-noise ratio is high enough to recover the latent signal.

\subsubsection{APPC with estimated parameters}
\label{subsec:oracle-appc-learned}

We now consider an estimator $\widehat\theta_m$ based on $\mX_m^\train$. For the moment we assume that both $P_0$ and $p_{\widehat{\theta}_m}(\bm{z}_i|\bm{x}_i)$ are available. The reference distribution is $Q_m=P_0 K_{\widehat\theta_m,\mathrm{post}}$.

The centering and scaling terms from \Cref{prop:generic-appc-calibration} are:
\begin{align*}
b(P_0 K_{\widehat\theta_m,\mathrm{post}}) = \frac{\sqrt{N_{\new,m}} (\nu_0 - \nu_{\widehat{\theta}_m})}{\tau_{\widehat{\theta}_m}}, \qquad
s(P_0 K_{\widehat{\theta}_m,\mathrm{post}}) = \frac{\tau_{\widehat{\theta}_m}}{\tau_0},
\end{align*}
where $\nu_{\widehat{\theta}_m} = \E_{\bm{x}\sim P_0 K_{\widehat{\theta}_m,\mathrm{post}}}[h(\bm{x})]$, $\tau_{\widehat{\theta}_m}^2 =  \mathrm{Var}_{\bm{x}\sim P_0 K_{\widehat{\theta}_m,\mathrm{post}}}\{h(\bm{x})\}$.

We can decompose the centering term into different sources of error:
\begin{align}
(\nu_0- \nu_{\widehat{\theta}_m}) 
&= 
\underbrace{ (\nu_0 - \nu_{\mathrm{post}}^\star )}_{\text{misspecification error}}
+
\underbrace{ (  \nu_{\mathrm{post}}^\star
- \nu_{\widehat{\theta}_m})}_{\text{estimation error}}
\label{eq:error-decomp-est-post}
\end{align}
The first term is the same misspecification error from \Cref{cor:oracle-appc-calibration}. The second term measures estimation error relative to the population optimizer set.

To characterize the estimation error when the optimal parameter is not unique, we define a distance to the optimizer set: for $\theta \in \Theta$, 
\begin{align*}
\operatorname{dist}(\theta,\mathcal {S}^\star_{\mathrm{post}}) = \inf_{\bar\theta\in\mathcal {S}^\star_{\mathrm{post}}}
\lVert\theta-\bar\theta\rVert_2.    
\end{align*}
We also write for $j\in\{1,2\}$, $
m_{j,\theta} = \mathbb{E}_{\bm{x}\sim P_0 K_{\theta,\mathrm{post}}}[h^j(\bm{x})]$ and $m_{j,\mathrm{post}}^\star = \mathbb{E}_{\bm{x}\sim P_0 K_{\theta^\star,\mathrm{post}}}[h^j(\bm{x})]$.

\vspace{1em}

\begin{proposition}
\label{prop:learned-appc-calibration}
Suppose \Cref{ass:diag-clt,ass:conditional-diagnostic-clt}  hold for $Q_m=P_0 K_{\widehat\theta_m,\mathrm{post}}$. Assume further that the conditions of \Cref{cor:oracle-appc-calibration} hold. Suppose that, for some
sequence $r_{\train,m}\to0$, $\operatorname{dist}
(\widehat\theta_m,\mathcal S_{\mathrm{post}}^\star)
=
O_P(r_{\train,m})$.
Assume that there exist constants $C<\infty$, $\epsilon>0$, and
$\alpha>0$ such that, for $j\in\{1,2\}$ and every $\theta$ satisfying
$\operatorname{dist}(\theta,\mathcal S_{\mathrm{post}}^\star)<\epsilon$,
\begin{align}
\left|
   m_{j,\theta}-m_{j,\mathrm{post}}^\star
\right|
\le
C\operatorname{dist}
(\theta,\mathcal S_{\mathrm{post}}^\star)^\alpha,
\label{eq:moment-holder-to-learned-set}
\end{align}
where $\lvert\nu^\star_{\mathrm{post}}\rvert=O(1)$ and
$\tau_{\mathrm{post}}^\star$ is  bounded away from zero. If
\begin{align}
\sqrt{N_{\new,m}}\,r_{\train,m}^\alpha
\longrightarrow0,
\label{eq:training-reference-rate-condition}
\end{align}
then the \gls{APPC} with estimated parameters is calibrated.
\end{proposition}

\vspace{1em}

\begin{remark}
If
$\operatorname{dist}(\widehat\theta_m,\mathcal S_{\mathrm{post}}^\star)
=O_P(N_{\train,m}^{-1/2})$
and the moment functionals are Lipschitz to the optimizer set
($\alpha=1$), then \eqref{eq:training-reference-rate-condition} reduces to $N_{\new,m}/N_{\train,m}\longrightarrow0$.
\end{remark}

\vspace{1em}

\begin{remark}
When $\theta$ is random, $N_\train$ typically can be of the same order as $N_\new$. This is because the Bayesian treatment accounts for uncertainty in $\theta$ (see \Cref{app:bayesian-appc}).
\end{remark}

\subsubsection{APPC with variational parameters}
\label{subsec:appc-vi}

Suppose now the latent posterior is approximated with a variational density, $q_{\phi}(\bm{z}_i|\bm{x}_i)$. We still assume known $P_0$; this is relaxed in the next section. Here, $\widehat{\eta}_m=(\widehat{\theta}_m,\widehat{\phi}_m)$ is a function of the training data $\bm X_m^\train$. The reference distribution is $Q_m = P_0 K_{\widehat\eta_m}$ and we denote the optimal variational parameters as $\eta^\star \in \mathcal{S}^\star_{\mathrm{VI}}$.

We can decompose the centering term into different sources of error:
\begin{align}
(\nu_0-\nu_{\widehat{\eta}_m}) 
&= \underbrace{ (\nu_{0} - \nu_{\mathrm{post}}^\star )}_{\text{misspecification error}}
+\underbrace{ (  \nu_{\mathrm{post}}^\star- \nu_{\mathrm{VI}}^\star)}_{\text{variational error}}
+
\underbrace{ (  \nu_{\mathrm{VI}}^\star-\nu_{\widehat{\eta}_m})}_{\text{estimation error}}
,\label{eq:error-decomp-variational}
\end{align}
where $\nu_{\mathrm{VI}}^\star = \E_{\bm{x}\sim P_0K_{\eta^\star}}[h(\bm{x})] $ and $\nu_{\widehat{\eta}_m} = \E_{\bm{x}\sim P_0 K_{\widehat{\eta}_m}}[h(\bm{x})]$.

The variational error combines the posterior-approximation error with the population bias in estimating $\theta$ via optimizing  the ELBO instead of the marginal likelihood. 

The \gls{APPC} with variational parameters is calibrated under similar conditions as the previous section, provided the variational error is controlled (\Cref{prop:vi-appc-calibration} in \Cref{app:appc-vi}). 

\subsubsection{Empirical APPC}
\label{subsec:empirical-appc-general}

We now consider the empirical \gls{APPC}, where $P_0$ is not available and is instead approximated by the empirical distribution of $\bm X_m^\agg$.  The reference distribution is
$Q_m=\widehat{P}_{\agg,m} K_{\widehat{\eta}_m}$. The mean and variance under  $\widehat{P}_{\agg,m} K_{\widehat{\eta}_m}$ are:
\begin{align*}
\nu_{\agg,\widehat{\eta}_m}
=
\mathbb E_{\bm{x}\sim \widehat{P}_{\agg,m} K_{\widehat{\eta}_m}}\left[h(\bm{x})\right], \quad\tau_{\agg,\widehat{\eta}_m}^2
= \Var_{\bm{x}\sim \widehat{P}_{\agg,m} K_{\widehat{\eta}_m}}\{h(\bm{x})\}.
\end{align*}
The centering and scaling terms from \Cref{prop:generic-appc-calibration} are:
\begin{align*}
b(\widehat{P}_{\agg,m} K_{\widehat{\eta}_m}) &= \frac{\sqrt{N_{\new,m}} ( \nu_0 - \nu_{\agg,\widehat{\eta}_m} )}{\tau_{\agg,\widehat{\eta}_m}}, \qquad
s(\widehat{P}_{\agg,m} K_{\widehat{\eta}_m}) = \frac{\tau_{\agg,\widehat{\eta}_m}}{\tau_0}.
\end{align*}
The centering term decomposes as
\begin{align}
(\nu_0 - \nu_{\agg,\widehat\eta_m})
&=
\underbrace{
(\nu_0-\nu_{\mathrm{post}}^\star)
}_{\text{misspecification error}} + 
\underbrace{
(\nu_{\mathrm{post}}^\star-\nu_{\mathrm{VI}}^\star)
}_{\text{variational error}}
+
\underbrace{
(\nu_{\mathrm{VI}}^\star - \nu_{\widehat\eta_m})
}_{\text{estimation error}}
+
\underbrace{
(\nu_{\widehat\eta_m}-\nu_{\agg,\widehat\eta_m})
}_{\text{aggregation error}}.
\label{eq:empirical-centering-decomposition}
\end{align}
Thus, we need to control the aggregation error due to approximating $P_0$ with $\widehat{P}_\agg$, in addition to the training, variational and model misspecification errors.

\vspace{1em}

\begin{proposition}\label{prop:emp-appc-calibration}
Assume the conditions of \Cref{prop:vi-appc-calibration} hold. Assume further that
\Cref{ass:conditional-diagnostic-clt} holds with
$Q_m
=
\widehat P_{\agg,m}K_{\widehat\eta_m}$ and $\mathcal F_m
=
\sigma\left(
\bm X_m^{\train},
\bm X_m^{\agg}
\right)$.
Moreover, assume that $N_{\new,m}/N_{\agg,m}\to 0$ and 
\begin{align}
\E_{\bm{x}^\rep\sim P_0 K_{\widehat{\eta}_m}}[h^4(\bm{x}^\rep)]&=O_P(1) \label{eq:h4-emp-agg}.
\end{align}
Then, the empirical \gls{APPC} is calibrated. 
\end{proposition}

Intuitively, the empirical \gls{APPC} is calibrated when  $\widehat{P}_{\agg,m}$ converges to $P_0$ at a rate much faster than the rate at which the holdout diagnostic $d_{N_{\new,m}}(\bm X_m^\new)$ converges to its mean. This is because any discrepancy between $\widehat{P}_{\agg,m}$ and $P_0$ can be detected with a large enough sample size. 

\subsection{Case Study: PPCA}\label{subsec:ppca-theory}

We now specialize the preceding results to the PPCA model:
\begin{align*}
\bm z_i \sim \mathcal{N}\left(0,\,I_K\right),
\qquad
\bm x_i\mid\bm z_i \sim \mathcal{N}\left(\bm W\bm z_i,\,\sigma^2I_D\right),\qquad \bm{W}\in \R^{D\times K}.
\end{align*}
\textbf{Notation.} For a fixed parameter $\theta=(\bm W,\sigma^2)$, define
\begin{align}
\bm B
\defeq
\bm W \left[\bm W^\top\bm W+\sigma^2I_K\right]^{-1}\bm W^\top,
\qquad
\bm R
\defeq
\sigma^2\left(I_D+\bm B\right).
\label{eq:ppca-B-R}
\end{align}
Let $\theta^\star\in\mathcal S_{\mathrm{post}}^\star$, and let $\bm B^\star$ and $\bm R^\star$ denote the corresponding
quantities in \Cref{eq:ppca-B-R}. The matrices $\bm B^\star$ and $\bm R^\star$ are invariant to rotations of the loadings matrix. Let
$\widehat\theta_m=(\widehat{\bm W}_m,\widehat\sigma_m^2)$ be the maximum likelihood estimator based on
$\bm X_m^\train$, and define $\widehat{\bm{B}}_m$ and $\widehat{\bm{R}}_m$ analogously to \Cref{eq:ppca-B-R}. We fit zero-mean PPCA without centering. The aggregation  sample mean and empirical covariance are  $\overline{\bm x}_m^\agg$ and $\widehat{\bm\Sigma}_{\agg,m}$.
Unless stated otherwise, all distributions, parameters, moments, and matrices may depend on $m$; we suppress this subscript when it does not aid readability.

\textbf{Data generating process.} We consider a misspecified prior case. The true data generating process is
\begin{align}
\bm{z}_i
\sim
G,
\qquad
\bm{x}_i\mid\bm{z}_i
\sim
\mathcal N(\bm{W}_0\bm{z}_i,\,\sigma_0^2I_D),
\label{eq:ppca-ms-data-theory}
\end{align}
where $\bm{W}_0$ and $\sigma_0^2$ are the true parameters.
The prior $G$ need not be standard Gaussian; we denote $\bm m_Z
\defeq
\E_G[\bm z]$ and $\bm\Lambda
\defeq \Var_G(\bm z)$. At the data level, the population mean and variance are $\bm\mu_0\defeq\E_{P_0}[\bm x]$ and $\bm\Sigma_0\defeq\Var_{P_0}(\bm x)$. 

\textbf{Diagnostic.} For a diagnostic function, we consider the weighted mean:
\begin{align*}
d_{N_{\new,m}}(\bm X_m)
=
\bm a_m^\top\overline{\bm x}_m, \quad a_{m,j}\geq 0,\quad \sum_{j=1}^{D_m} a_{m,j} = 1,\quad D_m\lVert \bm{a}_m\rVert_2^2 = O(1).
\end{align*}
Given this diagnostic, we have successive reference distributions in \Cref{tab:ppca-reference-distributions}.

\begin{table}
\centering
\small
\renewcommand{\arraystretch}{1.45}
\begin{tabular}{lll}
\toprule
\textbf{Reference distribution $Q_m$}
&
\textbf{Reference mean $\bm\mu_m^\rep$}
&
\textbf{Reference covariance $\bm V_m^\rep$}
\\
\midrule

$P_0 K_{\theta^\star,\mathrm{post}}$
&
$\bm B^\star\bm\mu_0$
&
$\bm V_m^{\mathrm{oracle}}
 =
 \bm R^\star+\bm B^\star\bm\Sigma_0\bm B^\star$
\\[1mm]

$P_0 K_{\widehat\theta_m,\mathrm{post}}$
&
$\widehat{\bm B}_m\bm\mu_0$
&
$\bm V_m^{\mathrm{learn}}
 =
 \widehat{\bm R}_m
 +
 \widehat{\bm B}_m\bm\Sigma_0\widehat{\bm B}_m$
\\[1mm]
$\widehat P_{\agg,m}K_{\widehat\theta_m,\mathrm{post}}$
&
$\widehat{\bm B}_m\overline{\bm x}_m^\agg$
&
$\bm V_m^{\mathrm{emp}}
 =
 \widehat{\bm R}_m
 +
 \widehat{\bm B}_m
 \widehat{\bm\Sigma}_{\agg,m}
 \widehat{\bm B}_m$
\\
\bottomrule
\end{tabular}
\caption{PPCA reference distributions.}\label{tab:ppca-reference-distributions}
\end{table}

\textbf{Asymptotic $p$-value.} Under a conditional central limit theorem, all three one-sided predictive $p$-values have the common form
\begin{align}
p_m(Q_m)
&=
1-\Phi\left(
\frac{
\sqrt{N_{\new,m}}\,
\bm a^\top
\left(
\overline{\bm x}_m^\new-\bm\mu_m^\rep
\right)
}{
\sqrt{\bm a^\top\bm V_m^\rep\bm a}
}
\right)
+o_P(1).
\label{eq:ppca-common-appc-pvalue}
\end{align}
The PPCA centering and scale quantities for
\Cref{prop:generic-appc-calibration} are
\begin{align}
b_m =\frac{\sqrt{N_{\new,m}}\,
\bm a^\top\left(\bm\mu_0-\bm\mu_m^\rep
\right)}{\sqrt{\bm a^\top\bm V_m^\rep\bm a}},
\qquad 
s_m = \left\{\frac{\bm a^\top\bm V_m^\rep\bm a
}{\bm a^\top\bm\Sigma_0\bm a}\right\}^{1/2}.
\label{eq:ppca-common-b-s}
\end{align}
The three centering discrepancies are nested:
\begin{align}
\bm\mu_0-\bm B^\star\bm\mu_0
&=
(I_D-\bm B^\star)\bm\mu_0,
\label{eq:ppca-fixed-oracle-centering}
\\
\bm\mu_0-\widehat{\bm B}_m\bm\mu_0
&=
\underbrace{(I_D-\bm B^\star)\bm\mu_0}_{\text{misspecification error}}
+
\underbrace{(\bm B^\star-\widehat{\bm B}_m)\bm\mu_0}_{\text{estimation error}},
\label{eq:ppca-learned-centering}
\\
\bm\mu_0-\widehat{\bm B}_m\overline{\bm x}_m^\agg
&=
\underbrace{(I_D-\bm B^\star)\bm\mu_0}_{\text{misspecification error}}
+
\underbrace{(\bm B^\star-\widehat{\bm B}_m)\bm\mu_0}_{\text{estimation error}}
+
\underbrace{
\widehat{\bm B}_m
\left(
\bm\mu_0-\overline{\bm x}_m^\agg
\right)
}_{\text{aggregation error}}.
\label{eq:ppca-empirical-centering}
\end{align}

\begin{assumption}[Regularity conditions]
\label{ass:ppca-regularity}
We impose the following regularity conditions on the PPCA data generating process.
\begin{itemize}
    \item The error variance satisfies $0<c_\sigma\le\sigma_0^2\le C_\sigma<\infty.$
    \item The latent factor mean satisfies $\lVert \bm{m}_{Z}\rVert = O(1)$. 
    \item The latent factor variance satisfies $0 < c_{\bm \Lambda} \leq \lambda_{\mathrm{min}}(\bm \Lambda_m) \leq \lambda_{\mathrm{max}}(\bm \Lambda_m) \leq C_{\bm \Lambda} < \infty.$
\item The latent factor distribution has bounded fourth moments $\sup_m \E_{G_m}\lVert\bm z\rVert_2^4 < \infty.$
\end{itemize}
\end{assumption}

In this misspecified case, we assume the latent factors are pervasive \citep{bai2003inferential,fan2013large}. Pervasiveness enables recovery of the signal subspace as $N_{\train,m}$ and $D_m$ grow. 

\begin{assumption}[Pervasiveness]
\label{ass:ppca-pervasive}
The number of factors $K$ is fixed. For some positive-definite
$K\times K$ matrix $\bm Q_W$, we have
$\frac{1}{D_m}\bm W_0^\top\bm W_0
\longrightarrow
\bm Q_W.$
\end{assumption}

We now specialize \Cref{prop:generic-appc-calibration} for PPCA.

\begin{proposition}
\label{prop:ppca-three-reference-calibration} Consider the model in \Cref{eq:ppca-ms-data-theory}. Assume that
$\bm a^\top\bm\Sigma_0\bm a$ and
$\bm a^\top\bm V_m^{\mathrm{oracle}}\bm a$ are bounded away from zero, and that the
relevant conditional central limit theorem holds for each reference
distribution in \Cref{tab:ppca-reference-distributions}.

\begin{enumerate}
\item[(i)] Suppose \Cref{ass:ppca-regularity,ass:ppca-pervasive} hold. If $D_m\to \infty$ and $\sqrt{N_{\new,m}}/D_m\to 0$, then
\begin{align}
\frac{
\sqrt{N_{\new,m}}\,
\bm a^\top(I_{D_m}-\bm B^\star)\bm\mu_0
}{
\sqrt{\bm a^\top\bm V_m^{\mathrm{oracle}}\bm a}
}
&\stackrel{P}{\longrightarrow} 0,
\label{eq:ppca-oracle-centering-condition}
\\
\frac{
\bm a^\top\bm V_m^{\mathrm{oracle}}\bm a
}{
\bm a^\top\bm\Sigma_0\bm a
}
&\stackrel{P}{\longrightarrow} 1,
\label{eq:ppca-oracle-scale-condition}
\end{align}
and the oracle \gls{APPC} is calibrated.

\item[(ii)] Suppose the conditions in part (i) hold. If $N_{\train, m} \to \infty$ and $N_{\new,m}/N_{\train,m} \to 0$, 
then the \gls{APPC} with estimated parameters and known $P_0$ is calibrated. 

\item[(iii)] Suppose the conditions in part (ii) hold. If  $N_{\new,m}/N_{\agg,m}\to0$ and
\begin{align}
\E_{\bm{x}\sim P_0 K_{\widehat\theta_m,\mathrm{post}}}
\left[\left(\bm a^\top\bm x\right)^4\right]
&=
O_P(1),
\label{eq:ppca-oracle-fourth-moment}
\end{align}
then the empirical \gls{APPC} is calibrated.
\end{enumerate}
\end{proposition}

%% file: 4-experiments.tex
\glsresetall

\section{Experiments}
\label{sec:Experiments}

We now study the \gls{APPC} for deep latent variable models, predominantly VAEs. Our study considers the questions: (1) Can the aggregated posterior predictive (APP) generate realistic samples even when the prior is incorrect? (2) How does the APP with a simple prior compare with existing VAE variants which have more complex priors? 
We find that:
\begin{enumerate}
    \item When the prior is incorrect, the APP can still generate samples for which the diagnostic does not detect a discrepancy.
    \item Unlike the PPCA example, the \gls{APPC} for VAEs are often not uniformly distributed. This may be due to optimization, variational or aggregation errors. 
    \item The APP performs similarly to existing VAE variants with more complex priors.
\end{enumerate}

\subsection{Datasets}
We consider both synthetic and real-world datasets covering heavy-tailed, multimodal, structured, and high-dimensional settings. For details on the exact settings, see \Cref{app:syn-datasets}.

\textbf{Synthetic datasets:}
\begin{enumerate}
    \item Bivariate-$t_{2}$: bivariate $t$-distributed data with $2$ degrees of freedom. 
    \item Clustered-$t_{2}$: $5$-component mixture of bivariate $t$ with $2$ degrees of freedom. 
\end{enumerate}

\textbf{Real datasets:}
\begin{enumerate}[start=3]
    \item Finance:  standardized daily bivariate returns of S$\&$P $500$ and Dow Jones indices from Yahoo Finance market data, following the setup of \citet{tam2025statistical}.
    \item Images: MNIST, a dataset of grayscale handwritten digit images \citep{lecun2010mnist}.
\end{enumerate}

\subsection{Models}

The baseline model is a standard VAE \citep{kingma2013auto} with a Gaussian latent prior. The decoder uses a Gaussian likelihood with learnable variance for continuous data and a Bernoulli likelihood for image data.  

Across all the experiments, for the fitted standard VAE, we consider three latent sampling strategies to generate replicates: (i) prior sampling, where latent variables are drawn from the Gaussian prior (G-PS); (ii) empirical aggregated posterior sampling, where latents are obtained by encoding the aggregation dataset, $\mX^\agg$ and resampling from the aggregated posterior (\Cref{eq:encoder-empirical-aggregation-notvi}) (G-EAP); and (iii) fitted aggregated posterior sampling, where a diffusion model is trained on encoded latent variables to produce an approximation to the aggregated posterior (G-DAP). (The ``G'' indicates that the first-stage VAE had a Gaussian prior). In each case, the latent draws are used to generate replicates via the fitted likelihood.

We compare these strategies against VAE variants that increase latent expressiveness and target heavy-tailed or multimodal structures, including VAEs with mixture-based priors \citep{tomczak2018vamp} and VAEs with heavy tailed priors \citep{kim2024t}. 

For the real datasets, we employ additional strategies for the VAE including: adjusting the KL penalty to avoid posterior collapse, and using tempering when drawing from the latent posterior (see \Cref{app:real-data-adjustments} for more details). 

\parhead{Experimental Framework.} 
We split the data into three: $\{\bm{X}^\train, \bm{X}^\agg, \bm{X}^\new\}$. For each model, $\bm{X}^\train$ is used to obtain parameter estimates. For strategies which estimate the aggregated posterior, the empirical distribution of $\bm{X}^\agg$ is used to approximate $P_0$. Finally, $\bm{X}^\new$ is the holdout data for the predictive check.

For synthetic data, the  $p$-value distribution is approximated by repeating the experiments $T$ times. For real data, we instead use $T$ bootstrapped trials. In each trial, we resample the observed data with replacement to form the three datasets.

\parhead{Evaluation.} We evaluate sampling strategies based on the distribution of $p$-values of their corresponding predictive checks. Ideally, a predictive check should have uniform $p$-values under the null hypothesis that the sampling strategy produces realistic data (with respect to the diagnostic function). We use Kolmogorov-Smirnov to test uniformity of $p$-values.

\subsection{Diagnostic Functions}

For each dataset, we chose a diagnostic sensitive to the characteristics of that data. 

\textbf{$t$-distributed data.} We use the mean log-norm diagnostic, 
\begin{align*}
d(\bm{X}) = \frac{1}{n}\sum_{i=1}^n \log\big(1+\|\bm{x}_i\|\big)
\end{align*} 
which captures overall magnitude while reducing sensitivity to extreme observations. Such log-moment functionals are commonly used in heavy-tail analysis \citep{resnick2007heavy}.

\textbf{Financial data.} We use the joint extreme share diagnostic
\begin{align*}
d(\bm{X}) = \frac{1}{n}\sum_{i=1}^n \mathbf{1}\left(\min_{1\le j \le d} |\bm{x}_{ij}| > t , \, \mathrm{sign}(\bm{x}_{i1}) = \dots = \mathrm{sign}(\bm{x}_{id})\right),
\end{align*}
measuring the frequency of simultaneous extreme co-movements across dimensions. Such events are of interest in multivariate extreme value analysis \citep{embrechts2013modelling}.

\textbf{MNIST.} We use a centroid confidence diagnostic, motivated by prototype learning \citep{snell2017prototypical}. This diagnostic assesses whether generated samples resemble coherent digit prototypes. Let $\{c_k\}_{k=1}^{K}$ denote the class centroids where $c_k$ is the sample mean of images with label $k$ over the training data. The diagnostic is
\begin{align*}
d(\bm{X}) = \frac{1}{n}\sum_{i=1}^n \max_k \pi_{ik}, \quad \text{where  } \pi_{ik} = \frac{\exp(-\|\bm{x}_i - c_k\|^2)}{\sum_{\ell=1}^K \exp(-\|\bm{x}_i - c_\ell\|^2)}.
\end{align*}
For each image, $\max_k \pi_{ik}$ measures how strongly it resembles its closest class centroid.

\subsection{Results}

\parhead{1. Bivariate $t$ with 2 df.}
We first study heavy-tailed bivariate $t$-distributed data. We compare G-PS, G-EAP and G-DAP for the baseline VAE model to models with more complex priors: a VAE with Student's-$t$ prior (stVAE, \citet{24babff9-6288-42ab-a099-4272559768c3}), the $t^3$-VAE (\citet{kim2024t}); and a mixture-based prior (VampVAE, \citet{tomczak2018vamp}). We also fit a diffusion model \citep{ho2020denoising} directly to the data.

We see the G-EAP and G-DAP strategies can generate heavy-tailed data (\Cref{fig:bivt-2df}).  However, their \gls{APPC} $p$-value distributions (orange and green) are not quite uniformly distributed, with peaks at zero and one. These peaks indicate that either the sampling strategy does not replicate the data distribution or that finite-sample training, variational or aggregation errors remain non-negligible. Despite this, G-EAP and G-DAP often produce samples for which the diagnostic does not detect a discrepancy from the holdout data, and have similar $p$-value distributions to the more complicated models (VampPrior and stVAE). 
Further, G-EAP and G-DAP have similar predictive check $p$-values to a diffusion in data-space.

\begin{figure}[!htbp]
    \centering
    \includegraphics[width=\linewidth]{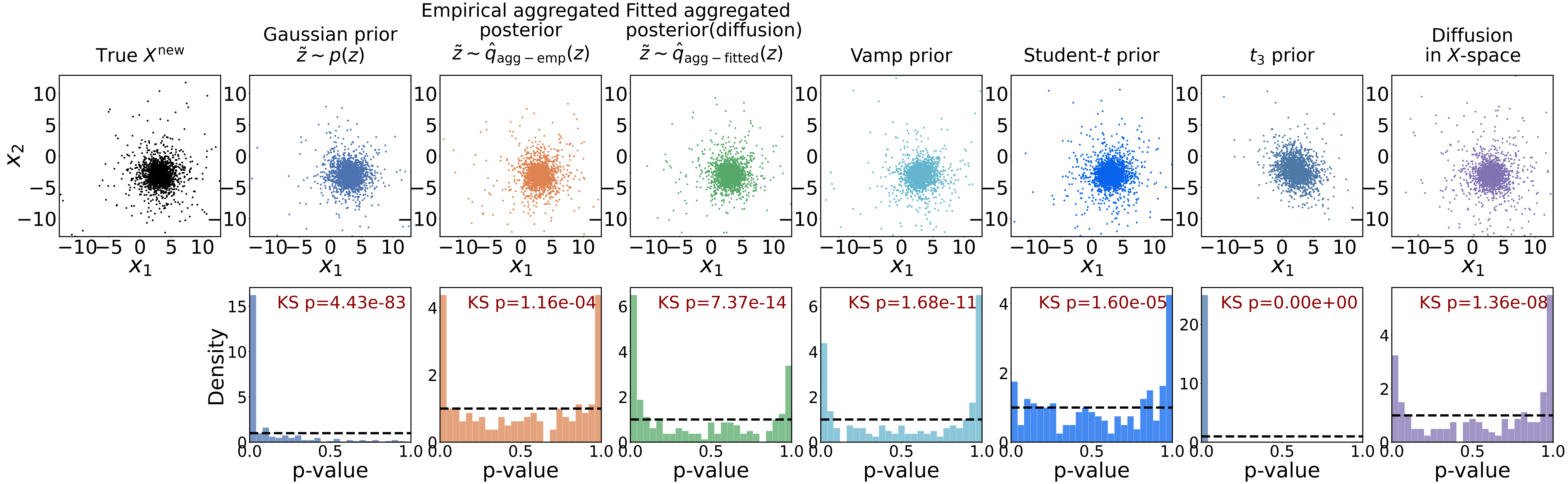}
    \caption{\textbf{Bivariate t.} Top: Draws of synthetic data $\bm{X}^\rep$. Bottom: $p$-values from model checks over $200$ trials. G-EAP and G-DAP $p$-values are not uniform but have similar performance as more complex generative models. Here, $N_\train = N_\agg = 10000, N_\new = 2000$.}
    \label{fig:bivt-2df}
\end{figure}

\parhead{2. Clustered $t$ with 2 df.} We next study  $5$-component bivariate $t$-distributed data (\Cref{fig:clustt-2df}). We compare G-PS, G-EAP and G-DAP to VAEs which have mixture priors (the VampPrior, and a mixture of $t$-distributed components). Our observations are similar to the previous experiment: G-EAP and G-DAP perform similarly to the VampPrior and the VAE with mixture of $t$-distributed components in terms of their predictive check $p$-values.

\begin{figure}[!htbp]
    \centering
    \includegraphics[width=\linewidth]{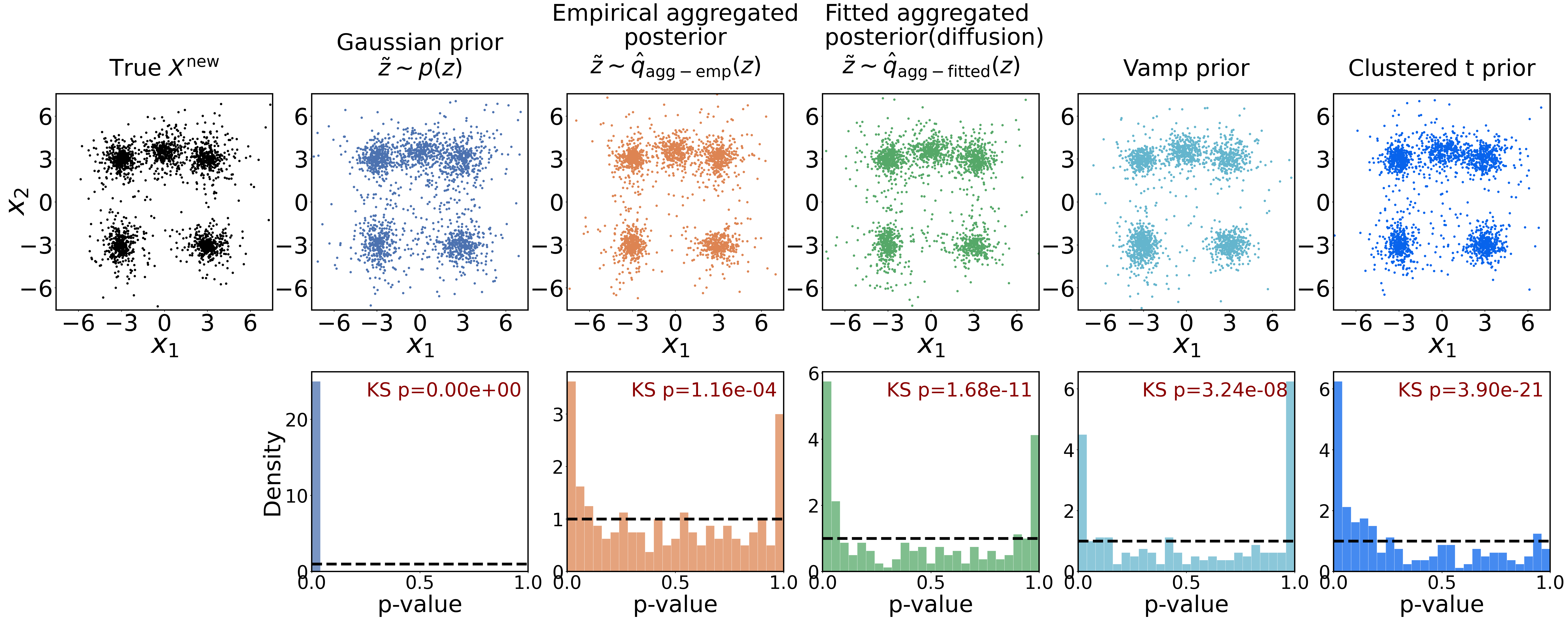}
    \caption{\textbf{Clustered t.} Top: Draws of synthetic data $\bm{X}^\rep$.  Bottom: $p$-values from model checks over $200$ trials. G-EAP and G-DAP $p$-values are not uniform but show similar behavior to more complex generative models. Here, $N_\train = N_\agg = 10000, N_\new = 2000$.}
    \label{fig:clustt-2df}
\end{figure}

\parhead{3. Finance data.} G-EAP can produce samples similar to the heavy-tailed observed data: the G-EAP $p$-values are consistent with a uniform distribution (\Cref{fig:financial}). Meanwhile, the G-DAP $p$-values are not uniform; the fitted diffusion may be oversmoothing. 

\begin{figure}[!htbp]
    \centering  \includegraphics[width=0.65\linewidth]{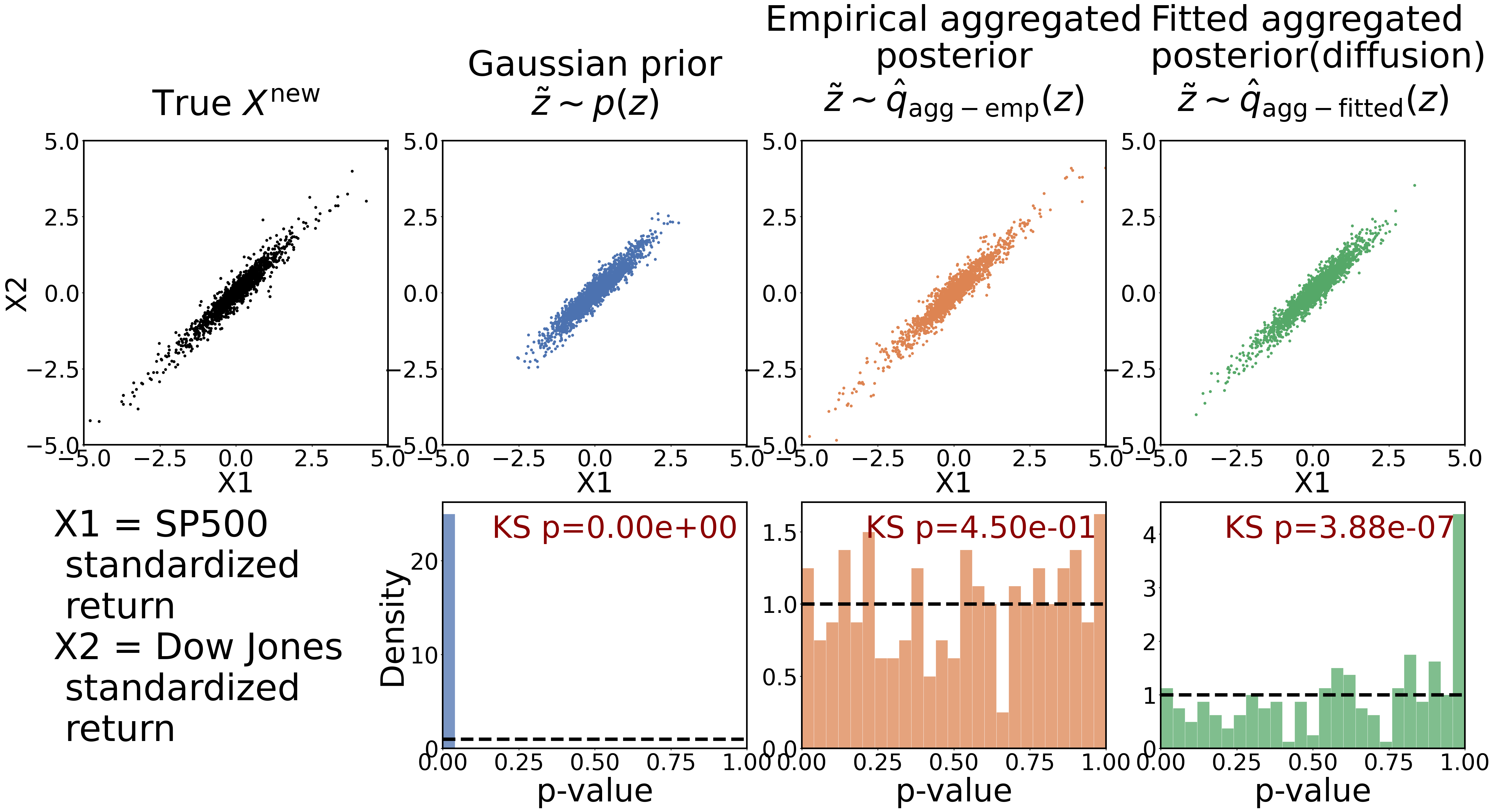}
    \caption{\textbf{Financial.} Top: Aggregated posterior with standard VAE (learnable variance, 1000 epochs) produces synthetic $\bm{X}^\rep$ which visually capture the heavy-tailed behavior of financial data. Bottom: G-EAP $p$-values are approximately uniform. Here,  $N_{\mathrm{train}} = N_{\mathrm{agg}} = 4000, N_{\mathrm{new}} = 2000$.}
    \label{fig:financial}
\end{figure}

\parhead{4. MNIST data.} G-EAP produces samples that are more visually aligned with a digit class than prior sampling and the G-DAP (\Cref{fig:mnist}). Further, the corresponding G-EAP $p$-values with the centroid confidence diagnostic are consistent with a uniform distribution.

\begin{figure}[!htbp]
    \centering
    \includegraphics[width=0.65\linewidth]{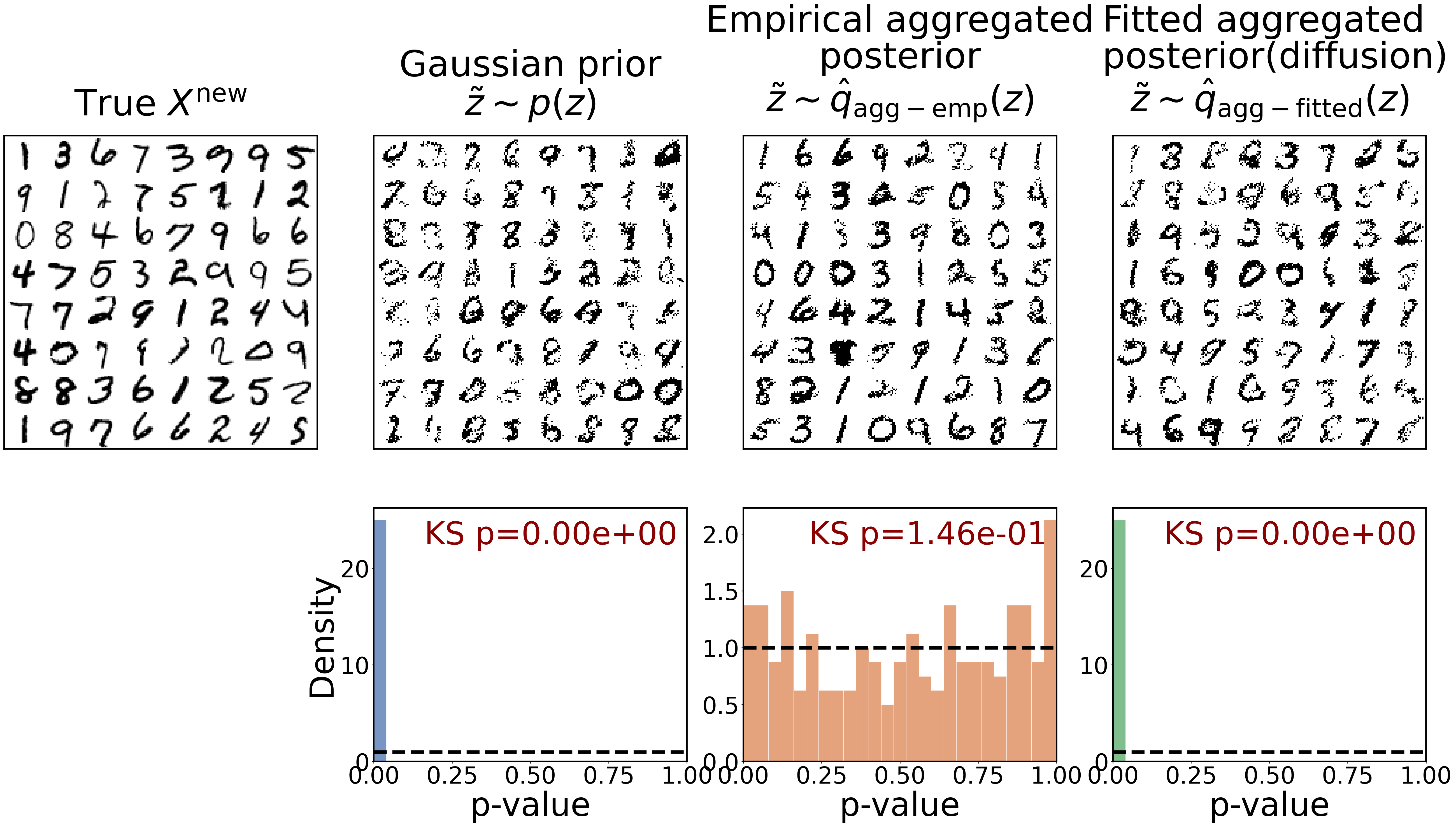}
    \caption{\textbf{MNIST.} Top: G-EAP produces recognizable digits. Bottom: G-EAP $p$-values are approximately uniform. Here,  $N_{\mathrm{train}} = N_{\mathrm{agg}} = 10000, N_{\mathrm{new}} = 2000$.}
    \label{fig:mnist}
\end{figure}

%% file: 5-discussion.tex
\glsresetall

\section{Discussion} \label{sec: Discussion}

In this paper, we introduce the \gls{APPC}, a holdout predictive check for modern generative modeling pipelines. In modern generative modeling, practitioners often use a two-stage strategy where first a latent variable generative model is fit, and second, a distribution over the encoded latent variables is estimated, i.e. the aggregated posterior.  To check such pipelines, we propose the \gls{APPC}. The \gls{APPC} asks: does the two-stage strategy produce synthetic data which ``looks like'' a holdout dataset?

We provide theoretical guarantees for the \gls{APPC}; in particular, we provide sufficient conditions under which the \gls{APPC} yields calibrated $p$-values, even when the prior is incorrect. We specialize these results to a classic latent variable model, PPCA. A limitation of our results is that the calibration conditions need to be verified on a model-by-model basis. 

In experiments, we show that for VAE models, the empirical aggregated posterior predictive can capture heavier tails than Gaussian prior sampling. Our empirical results suggest several interesting directions for future work. Specifically, the two-stage strategy often recovers the target distribution well, motivating a broader theoretical investigation of such procedures from an empirical Bayes lens. Another direction is to investigate whether iterative aggregated posterior refitting can further improve performance, where the aggregated posterior from the first stage is used as the prior in the next step and the model is refit.

%% file: appendix.tex

\glsresetall

\section{Proofs}

The proofs in this section are organized as follows. 
\begin{enumerate}
    \item We consider the PPCA example \Cref{subsec:ppca}.  We derive the predictive check $p$-values for the latent prior predictive and conditional-replicate predictive distributions in \Cref{app:ppca-predictive-check-pval}. We show that the conditional-replicate $p$-values are not calibrated.
    \item We prove the theorems for the \gls{APPC} from \Cref{sec:theory}.
    \item We derive a Bayesian analogue of the \gls{APPC}, treating $\theta$ as random.
    \item Finally, we prove the calibration results for the PPCA example under prior misspecification.
\end{enumerate}

\subsection{PPCA example: Predictive check $p$-values}
\label{app:ppca-predictive-check-pval}

For the PPCA model \Cref{eq:ppca-ws-model}, the different predictive distributions from \Cref{sec:local-latent-references} are:

\emph{Latent prior predictive distribution:}
    \begin{align}
    \label{eq:latent-prior-predictive-distribution-ppca}
\bm{x}^{\rep, \mathrm{prior}}\mid\bm{X}^\train \sim \mathcal{N}\left(0,\, \widehat{\bm{W}}\widehat{\bm{W}}^\top +  \widehat{\sigma}^2\bm{I}_D\right);
    \end{align}
\emph{Conditional-per-sample predictive distribution:}
    \begin{align}
    \label{eq:conditional-replicate-predictive-distribution-ppca}
\bm{x}_i^{\rep, \mathrm{cond}}\mid\bm{X}^\train, \bm{x}_i^\agg \sim \mathcal{N}\left(\widehat{\bm{B}}\bm{x}_i^\agg,\,\widehat{\sigma}^2(\bm{I}_D + \widehat{\bm{B}})\right), \quad i =1,\dots, N_{\new};
    \end{align}
where $\widehat{\bm{B}} =\widehat{\bm{W}}(\widehat{\bm{W}}^\top{\widehat{\bm{W}}} + \widehat{\sigma}^2I_K)^{-1}{\widehat{\bm{W}}}^\top$.
    
\emph{Empirical aggregated posterior predictive distribution:}
    \begin{align}
\bm{x}^{\rep,\agg} \mid \bm{X}^\train, \bm{X}^\agg \sim \frac{1}{N_\agg}\sum_{i=1}^{N_\agg}\mathcal{N}\left(\widehat{\bm{B}}\bm{x}_i^\agg,\, \widehat{\sigma}^2(\bm{I}_D + \widehat{\bm{B}})\right).
\label{eq:ppca-empirical-aggregate-replicates-generator}
\end{align}

We provide the $p$-values for the replicate distributions from \Cref{sec:local-latent-references} for the mean diagnostic:
\begin{align*}
d(\bm{X}) = \bm{a}^\top \overline{\bm{x}}, \quad \overline{\bm{x}}= (\overline{x}_1,\dots, \overline{x}_D), \quad \overline{x}_j = \frac{1}{N}\sum_{i=1}^N x_{ij}.
\end{align*}

\subsubsection{Latent prior predictive}
\label{app:ppca-PC-pval}
The distribution of a synthetic data point ${\bm{x}_i}^\rep, \, i = 1, \dots, N_\new$ is given by \Cref{eq:latent-prior-predictive-distribution-ppca}.
Since ${\bm{x}}_i^{\rep, \mathrm{prior}}$ conditional on $\bm{X}^\train$ are i.i.d. draws from \Cref{eq:latent-prior-predictive-distribution-ppca}, we have
\begin{align*}
    \sqrt{N_{\new}}\bm{a}^\top\overline{\bm{x}}^{\rep, \mathrm{prior}}\mid\bm{X}^\train  \stackrel{d}{\longrightarrow} \mathcal{N}\left(0,\bm{a}^\top( \widehat{\bm{W}}\widehat{\bm{W}}^\top + \widehat{\sigma}^2I_D )\bm{a}\right).
\end{align*}
Then the predictive check $p$-value from \Cref{eq:predictive-check-pval} for the latent prior predictive distribution is
\begin{align}
    p^{\mathrm{prior}}_{\mathrm{PC}}
    &= 1- \Phi\left(\frac{\bm{a}^\top\overline{\bm{x}}^{\new}}{\sqrt{\frac{1}{N_{\new}}\,\bm{a}^\top ( \widehat{\bm{W}}\widehat{\bm{W}}^\top + \widehat{\sigma}^2I_D) \bm{a}}}
    \right).
    \label{eq:ppca-latent-prior-p-val}
\end{align}

\subsubsection{Conditional-per-sample predictive}
\label{app:ppca-HPC-pval}
The conditional-per-sample predictive distribution is given by \Cref{eq:conditional-replicate-predictive-distribution-ppca}.
Conditional on $\bm{X}^{\train}$ and $\bm{X}^\agg$, we have
\begin{align*}
    \sqrt{N_{\new}}\,\bm{a}^\top(\overline{\bm{x}}^{\rep, \mathrm{cond}} -\widehat{\bm{B}}\overline{\bm{x}}^\agg_{1:N_{\new}})\mid \bm{X}^\agg, \bm{X}^\train \stackrel{d}{\longrightarrow} \mathcal{N}\left(0,\,\bm{a}^\top\widehat{\sigma}^2(I_D + \widehat{\bm{B}}) \bm{a}\right)
\end{align*}
Hence, the predictive check $p$-value from \Cref{eq:predictive-check-pval} becomes
\begin{align}
    \label{eq:ppca-cond-rep-prior-p-val}
    p^{\mathrm{cond}}_{\mathrm{PC}}
    &=  1- \Phi\left(\frac{\bm{a}^\top\overline{\bm{x}}^{\new} - \bm{a}^\top \widehat{\bm{B}}\,\overline{\bm{x}}^{\agg}_{1:N_\new}}{\sqrt{\frac{1}{N_\new}\,\bm{a}^\top  \widehat{\sigma}^2(I_D+\widehat{\bm{B}}) \bm{a}}}
    \right).
\end{align}

\subsubsection{Calibration of predictive check $p$-values for PPCA}
\label{app:ppca-predictive-check-pval-calibration}
\begin{proposition}
Suppose the observed data follows the well-specified PPCA data generating process \Cref{eq:ppca-ws-model}. Define $\bm{\Sigma}_0 = \bm{W}_0\bm{W}_0^\top + \sigma_0^2I_D$ and $\bm{R}_0 = \sigma_0^2(I_D+\bm{B}_0)$. Suppose $\widehat{\bm{W}}\widehat{\bm{W}}^\top \stackrel{P}{\to} \bm{W}_0\bm{W}_0^\top$ and $\widehat{\sigma}^2\stackrel{P}{\to}\sigma_0^2$. If $N_\train, N_\new \to \infty$ and $N_\agg \geq N_\new$, then 
\begin{align*}
p^{\mathrm{prior}}_{\mathrm{PC}} &\stackrel{d}{\longrightarrow} U(0,1),\\
p^{\mathrm{cond}}_{\mathrm{PC}} &= 1- \Phi(s_{\mathrm{cond}}Z) + o_P(1),\quad Z\sim N(0,1)
\end{align*}
where 
\begin{align}
s_{\mathrm{cond}} = \sqrt{\frac{\bm{a}^\top (\bm{\Sigma}_0 + \bm{B}_0\bm{\Sigma}_0\bm{B}_0)\bm{a}}{\bm{a}^\top\bm{R}_0\bm{a}}}.
\label{eq:cond-rep-ppca-std}
\end{align}
That is, $p^{\mathrm{prior}}_{\mathrm{PC}}$ is asymptotically uniform (i.e. calibrated) while $p^{\mathrm{cond}}_{\mathrm{PC}}$ is not calibrated except when $\bm{a}^\top\bm{B}_0=0$.
\end{proposition}
\begin{proof}
Under the null hypothesis that the model is correctly specified, we have as $N_{\new} \rightarrow \infty$,
\begin{align*}
  \sqrt{N_{\new}}\,\bm{a}^\top \overline{\bm{x}}^\new \stackrel{d}{\longrightarrow} \mathcal{N}\left(0,\bm{a}^\top (\bm{W}_0\bm{W}_0^\top + \sigma_0^2I_D) \bm{a}\right)
\end{align*}
Under the assumption that the estimators $(\widehat{\bm{W}}\widehat{\bm{W}}^\top,\widehat{\sigma}^2)$ are consistent, by continuous mapping and Slutsky's theorems, as $N_\train, N_\new \rightarrow \infty$,
\begin{align*}
    \frac{\bm{a}^\top\overline{\bm{x}}^{\new}}{\sqrt{\frac{1}{N_{\new}}\,\bm{a}^\top ( \widehat{\bm{W}}\widehat{\bm{W}}^\top + \widehat{\sigma}^2I_D) \bm{a}}} \stackrel{d}{\longrightarrow} \mathcal{N}\left(0,1\right).
\end{align*}
Hence, $p^{\mathrm{prior}}_{\mathrm{PC}}$ from \Cref{eq:ppca-latent-prior-p-val} is calibrated.

For the conditional replicate $p$-value in \Cref{eq:ppca-cond-rep-prior-p-val}, let 
\begin{align*}
Q = \frac{\sqrt{N_{\new}}(\bm{a}^\top\overline{\bm{x}}^{\new} - \bm{a}^\top \widehat{\bm{B}}\,\overline{\bm{x}}_{1:N_\new}^{\agg})}{\sqrt{\bm{a}^\top  \widehat{\bm R} \bm{a}}}
\end{align*}
where $\widehat{\bm{R}} = \widehat{\sigma}^2(I_D+\widehat{\bm{B}})$.
Note that as $N_\new\longrightarrow\infty$,
\begin{align*}
    Q\stackrel{d}{\longrightarrow} \mathcal{N}(0, s_{\mathrm{cond}}^2).
\end{align*}
Since $p^{\mathrm{cond}}_{\mathrm{PC}}=1-\Phi(Q)$, the continuous mapping theorem yields
$p^{\mathrm{cond}}_{\mathrm{PC}}
\stackrel{d}{\longrightarrow}1-\Phi(s_{\mathrm{cond}}Z)$ for
$Z\sim \mathcal{N}(0,1)$.

It remains to characterize $s_{\mathrm{cond}}$. For PPCA,
\begin{align*}
\bm\Sigma_0=\bm R_0+\bm B_0\bm\Sigma_0\bm B_0.
\end{align*}
Consequently,
$s_{\mathrm{cond}}^2\geq1$.

Moreover, $\bm\Sigma_0\succ0$, so
$\bm a^\top\bm B_0\bm\Sigma_0\bm B_0\bm a=0$ if and only if
$\bm B_0\bm a=0$, equivalently
$\bm a\perp\operatorname{col}(\bm W_0)$. Outside of this edge case, $s_{\mathrm{cond}}>1$ and
$1-\Phi(s_{\mathrm{cond}}Z)$ is then nonuniform.
\end{proof}

\subsection{Proofs for \Cref{sec:theory}}

\subsubsection{Proof of \Cref{prop:generic-appc-calibration}}
\begin{proof}
The conditional normal approximation in \Cref{ass:conditional-diagnostic-clt} implies
\begin{align*}
p_m(Q_m) = 1-\Phi(W_m)+o_P(1),
\qquad
W_m = \frac{\sqrt{N_{\new,m}}\left[d_{N_{\new,m}}(\bm{X}_m^\new)-\nu_m\right]}{\tau_m}.
\end{align*}
Decompose
\begin{align*}
W_m &=
\frac{\sqrt{N_{\new,m}}\left[d_{N_{\new,m}}\left(\bm{X}_m^\new\right)-\nu_0\right]}{\tau_0}\frac{\tau_0}{\tau_m} +
\frac{\sqrt{N_{\new,m}}\left(\nu_0-\nu_m\right)}{\tau_m} = \frac{Z_m}{s_m} + b_m.
\end{align*}
Then,
\begin{align*}
\frac{\sqrt{N_{\new,m}}\left[d_{N_{\new,m}}(\bm{X}_m^\new)-\nu_0\right]}{\tau_0} \stackrel{d}{\longrightarrow} \mathcal{N}\left(0, 1\right).
\end{align*}
by the Central Limit Theorem as $m\to\infty$.

If $b_m\stackrel{P}{\longrightarrow} b$ and $s_m\stackrel{P}{\longrightarrow}s > 0$, then Slutsky's theorem and the continuous mapping theorem give:
\begin{align*}
p_m(Q_m) \Rightarrow 1 - \Phi(s^{-1}Z + b), \quad Z\sim \mathcal{N}(0,1).
\end{align*}
\end{proof}

\subsection{Proof of \Cref{cor:oracle-ws-calibration}}
\begin{proof}
In the well-specified case, we have from \Cref{eq:pop-optimal-ap-kernel},
\begin{align*}
(P_{\theta^\star} K_{\theta^\star,\mathrm{post}})(A)
&= \int \int
P_{\theta^\star}(\bm{x}^\rep \in A \mid\bm z)
\,p_{\theta^\star}(d\bm z\mid\bm x)
\,P_{\theta^\star}(d\bm x)
\notag\\
&= \int P_{\theta^\star}(\bm{x}^\rep \in A \mid\bm z)\,p_Z(d\bm z)\\
&= P_{\theta^\star}(A).
\end{align*}
When $P_{0}=P_{\theta^\star}$, we have $\nu_0 = \nu_{\mathrm{post}}^\star$ and $\tau_0=\tau_{\mathrm{post}}^\star$. Then, 
\begin{align*}
b(P_0 K_{\theta^\star,\mathrm{post}}) = 0,\quad s(P_0 K_{\theta^\star,\mathrm{post}}) = 1,
\end{align*}
and by \Cref{cor:generic}, $p_m(P_0 K_{\theta^\star, \mathrm{post}})\stackrel{d}{\longrightarrow} U[0,1]$.
\end{proof}

\subsection{Proof of \Cref{prop:learned-appc-calibration}}
\begin{proof}
We need to show that the centering and scale
conditions in \Cref{cor:generic} are satisfied when the reference distribution is $Q_m = P_0 K_{\widehat\theta_m,\mathrm{post}}$, that is,
\begin{align*}
    b(P_0 K_{\widehat\theta_m,\mathrm{post}}) = \frac{\sqrt{N_{\new,m}} (\nu_0 - \nu_{\widehat{\theta}_m})}{\tau_{\widehat{\theta}_m}} \stackrel{P}{\longrightarrow}0, \qquad
s(P_0 K_{\widehat{\theta}_m, \mathrm{post}}) = \frac{\tau_{\widehat{\theta}_m}}{\tau_0} \stackrel{P}{\longrightarrow}1.
\end{align*}
Then by \Cref{cor:generic}, $p_m(P_0 K_{\widehat{\theta}_m,\mathrm{post}})\stackrel{d}{\longrightarrow} U[0,1]$.

Since $m_{1,\widehat\theta_m} = \nu_{\widehat\theta_m}$ and $m_{1,\mathrm{post}}^\star = \nu_{\mathrm{post}}^\star$, by \Cref{eq:moment-holder-to-learned-set},
\begin{align*}
\lvert\nu_{\widehat\theta_m}-\nu_{\mathrm{post}}^\star\rvert
= O_P\left(r_{\train,m}^\alpha\right),
\qquad
\lvert m_{2,\widehat\theta_m}-m_{2,\mathrm{post}}^\star\rvert
= O_P\left(r_{\train,m}^\alpha\right).
\end{align*}
Hence, by \Cref{eq:training-reference-rate-condition} if $\sqrt{N_{\new,m}}\,r_{\train,m}^\alpha
\longrightarrow0$,
\begin{align}
\sqrt{N_{\new,m}}\,
\lvert\nu_{\widehat\theta_m}-\nu_{\mathrm{post}}^\star\rvert
=o_P(1).
\label{eq:centering-error-learnedtheta-posterior}
\end{align}
Moreover, since $\lvert\nu_{\mathrm{post}}^{\star}\rvert = O_P(1)$, then $\lvert\nu_{\widehat\theta_m}+\nu_{\mathrm{post}}^{\star}\rvert = O_P(1)$ as $r_{\train,m}^\alpha
\longrightarrow0$. Now, by the triangle inequality,
\begin{align*}
\lvert\tau_{\widehat\theta_m}^2-(\tau_{\mathrm{post}}^\star)^2\rvert
&= \lvert m_{2,\widehat\theta_m} -m_{2,\mathrm{post}}^\star -\nu_{\widehat\theta_m}^2 + (\nu_{\mathrm{post}}^{\star})^2\rvert\\
&\le \lvert m_{2,\widehat\theta_m} -m_{2,\mathrm{post}}^\star\rvert + \lvert\nu_{\widehat\theta_m}-\nu_{\mathrm{post}}^{\star}\rvert\lvert\nu_{\widehat\theta_m}+\nu_{\mathrm{post}}^{\star}\rvert\\
&= O_P\left(r_{\train,m}^\alpha\right) = o_P(1).
\end{align*}
Since $\tau^\star_{\mathrm{post}}$ is bounded away from zero, we have 
\begin{align}
\frac{\tau_{\widehat{\theta}_m}}{\tau^\star_{\mathrm{post}}} \stackrel{P}{\longrightarrow} 1.
\label{eq:scale-learnedtheta-posterior}
\end{align} 
Combining with the second condition of \Cref{cor:oracle-appc-calibration}, the scale condition is satisfied, 
\begin{align*}
    \frac{\tau_{\widehat{\theta}_m}}{\tau_0} = \frac{\tau_{\widehat{\theta}_m}}{\tau^\star_{\mathrm{post}}}\frac{\tau^\star_{\mathrm{post}}}{\tau_0} \stackrel{P}{\longrightarrow}1.
\end{align*}
Consequently, $\tau_{\widehat{\theta}_m}$ is also bounded away from zero with probability tending to 1.

For the centering condition, we decompose into different source of error as in \Cref{eq:error-decomp-est-post},
\begin{align*}
      b(P_0 K_{\widehat\theta_m,\mathrm{post}}) &= \frac{\sqrt{N_{\new,m}} (\nu_0 - \nu_{\widehat{\theta}_m})}{\tau_{\widehat{\theta}_m}} \\
      &= \underbrace{\frac{\sqrt{N_{\new,m}} (\nu_0 - \nu^\star_{\mathrm{post}})}{\tau^\star_{\mathrm{post}}}\frac{\tau^\star_{\mathrm{post}}}{\tau_{\widehat{\theta}_m}}}_{\text{Term 1}} + \underbrace{\frac{\sqrt{N_{\new,m}} ( \nu^\star_{\mathrm{post}} - \nu_{\widehat{\theta}_m})}{\tau_{\widehat{\theta}_m}}}_{\text{Term 2}}.
\end{align*}
The first term is $o_P(1)$ using the first condition of \Cref{cor:oracle-appc-calibration} and \Cref{eq:scale-learnedtheta-posterior}. The second term becomes $o_P(1)$ from \Cref{eq:centering-error-learnedtheta-posterior} and since ${\tau_{\widehat{\theta}_m}}$ is bounded away from zero, which completes the proof.
\end{proof}

\subsection{APPC with variational parameters.}
\label{app:appc-vi}
In this section, we provide sufficient conditions for calibration of the \gls{APPC} with variational parameters from \Cref{subsec:appc-vi}.

The centering and scaling terms from \Cref{prop:generic-appc-calibration} becomes:
\begin{align*}
b(P_0 K_{\widehat\eta_m}) &= \frac{\sqrt{N_{\new,m}} (\nu_0 - \nu_{\widehat{\eta}_m})}{\tau_{\widehat{\eta}_m}}, \qquad
s(P_0 K_{\widehat{\eta}_m}) = \frac{\tau_{\widehat{\eta}_m}}{\tau_0},
\end{align*}
where $\nu_{\widehat{\eta}_m} = \E_{\bm{x}\sim P_0 K_{\widehat{\eta}_m}}[h(\bm{x})]$, $\tau_{\widehat{\eta}_m}^2 =  \mathrm{Var}_{\bm{x}\sim P_0 K_{\widehat{\eta}_m}}\{h(\bm{x})\}$.

Analogously to the previous section, we define a distance to the variational optimizer set: for $\eta\in \Theta \times \mathcal{E}$,
\begin{align*}
\operatorname{dist}(\eta,\mathcal{S}_{\mathrm{VI}}^\star)
= \inf_{\bar\eta\in\mathcal{S}_{\mathrm{VI}}^\star}
\lVert\eta-\bar\eta\rVert.
\end{align*}
We further define $(\tau_{ \mathrm{VI}}^\star)^2 = \mathrm{Var}_{\bm{x}\sim P_0 K_{\eta^\star}}\{h(\bm{x})\}$ and for $j\in\{1,2\}$, we write $m_{j,\eta} = \mathbb{E}_{\bm{x}\sim P_0K_{\eta}}[h^j(\bm{x})]$ and $m_{j,\mathrm{VI}}^\star = \mathbb{E}_{\bm{x}\sim P_0 K_{\eta^\star}}[h^j(\bm{x})]$. 

The next proposition states sufficient conditions for the calibration of the \gls{APPC} with estimated variational parameters.  

\vspace{1em}

\begin{proposition}
\label{prop:vi-appc-calibration}
Suppose \Cref{ass:optimizer-predictive-equivalence} and the conditional CLT in \Cref{ass:conditional-diagnostic-clt} hold for $Q_m = P_0 K_{\widehat\eta_m}$. Assume further that the conditions of \Cref{cor:oracle-appc-calibration} hold. Suppose that, for some
sequence $r_{\train,m}\to0$, $\operatorname{dist}
(\widehat\eta_m,\mathcal{S}_{\mathrm{VI}}^\star)
= O_P(r_{\train,m})$.
Assume that there exist constants $C<\infty$, $\epsilon>0$, and
$\alpha>0$ such that, for $j\in\{1,2\}$ and every $\eta$ satisfying
$\operatorname{dist}(\eta,\mathcal{S}_{\mathrm{VI}}^\star)<\epsilon$,
\begin{align}
\left|
    m_{j,\eta}-m_{j,\mathrm{VI}}^\star
\right|
\le
C\operatorname{dist}
(\eta,\mathcal{S}_{\mathrm{VI}}^\star)^\alpha,
\label{eq:moment-holder-to-vi-set}
\end{align}
where $m_{1,\mathrm{VI}}^\star=\nu_{\mathrm{VI}}^\star$. Suppose also that $\lvert \nu^\star_{\mathrm{VI}} \rvert = O(1)$ and $\tau_{\mathrm{VI}}^\star$ is bounded away from zero.  Finally, suppose that the variational error satisfies:
\begin{align}
\frac{
    \sqrt{N_{\new,m}}
    (\nu^\star_{\mathrm{post}}-\nu_{\mathrm{VI}}^\star)
}{
    \tau_{\mathrm{VI}}^\star
}
\stackrel{P}{\longrightarrow}0,  \qquad \frac{
    \tau_{\mathrm{VI}}^\star
}{
    \tau_{\mathrm{post}}^\star
}
\stackrel{P}{\longrightarrow}1
\label{eq:vi-post-convergence}
\end{align}
If
\begin{align}
\sqrt{N_{\new,m}}\,r_{\train,m}^\alpha
\longrightarrow0,
\tag{\ref{eq:training-reference-rate-condition}}
\end{align}
then the \gls{APPC} with estimated variational parameters is calibrated.
\end{proposition}

\begin{proof}
The proof is similar to that of \Cref{prop:learned-appc-calibration}. By \Cref{eq:moment-holder-to-vi-set},
\begin{align*}
\lvert\nu_{\widehat\eta_m}-\nu_{\mathrm{VI}}^\star\rvert
=
O_P\left(r_{\train,m}^\alpha\right),
\qquad
\lvert m_{2,\widehat\eta_m}-m_{2,\mathrm{VI}}^\star\rvert
=
O_P\left(r_{\train,m}^\alpha\right).
\end{align*}
As $\sqrt{N_{\new,m}}\,r_{\train,m}^\alpha
\longrightarrow0$,
\begin{align}
\sqrt{N_{\new,m}}\,
\lvert\nu_{\widehat\eta_m}-\nu_{\mathrm{VI}}^\star\rvert
=o_P(1).
\label{eq:centering-error-variational-posterior}
\end{align}
Since $\lvert\nu_{\mathrm{VI}}^{\star}\rvert = O_P(1)$ and $r_{\train,m}^\alpha
\longrightarrow0$, then $\lvert\nu_{\widehat\eta_m}+\nu_{\mathrm{VI}}^{\star}\rvert = O_P(1)$ and 
\begin{align*}
\lvert\tau_{\widehat\eta_m}^2-(\tau_{\mathrm{VI}}^\star)^2\rvert
&\le \lvert m_{2,\widehat\eta_m} -m_{2,\mathrm{VI}}^\star\rvert + \lvert\nu_{\widehat\eta_m}-\nu_{\mathrm{VI}}^{\star}\rvert\lvert\nu_{\widehat\eta_m}+\nu_{\mathrm{VI}}^{\star}\rvert\\
&= O_P\left(r_{\train,m}^\alpha\right) = o_P(1).
\end{align*} 
As $\tau^\star_{\mathrm{VI}}$ is bounded away from zero, we have 
\begin{align}
\frac{\tau_{\widehat{\eta}_m}}{\tau^\star_{\mathrm{VI}}} \stackrel{P}{\longrightarrow} 1.
\label{eq:scale-variational-posterior}
\end{align} 
Combining with the second condition of \Cref{cor:oracle-appc-calibration} and the second assumption \Cref{eq:vi-post-convergence} of \Cref{prop:vi-appc-calibration} satisfies the required scale condition, 
\begin{align*}
    \frac{\tau_{\widehat{\eta}_m}}{\tau_0} = \frac{\tau_{\widehat{\eta}_m}}{\tau^\star_{\mathrm{VI}}}\frac{\tau^\star_{\mathrm{VI}}}{\tau^\star_{\mathrm{post}}}\frac{\tau^\star_{\mathrm{post}}}{\tau_0} \stackrel{P}{\longrightarrow}1.
\end{align*}
Consequently, $\tau_{\widehat{\eta}_m}$ is also bounded away from zero with probability tending to 1.

The centering condition has the additional variational error and we can decompose as in \Cref{eq:error-decomp-variational},
\begin{align*}
      b(P_0 K_{\widehat\eta_m}) &= \frac{\sqrt{N_{\new,m}} (\nu_0 - \nu_{\widehat{\eta}_m})}{\tau_{\widehat{\eta}_m}} \\
      &= \underbrace{\frac{\sqrt{N_{\new,m}} (\nu_0 - \nu^\star_{\mathrm{post}})}{\tau^\star_{\mathrm{post}}}\frac{\tau^\star_{\mathrm{post}}}{\tau^\star_{\mathrm{VI}}}\frac{\tau^\star_{\mathrm{VI}}}{\tau_{\widehat{\eta}_m}}}_{\text{Term 1}} +
      \underbrace{\frac{\sqrt{N_{\new,m}} ( \nu^\star_{\mathrm{post}} - \nu^\star_{\mathrm{VI}})}{\tau^\star_{\mathrm{VI}}}\frac{\tau^\star_{\mathrm{VI}}}{\tau_{\widehat{\eta}_m}}}_{\text{Term 2}} + \underbrace{\frac{\sqrt{N_{\new,m}} ( \nu^\star_{\mathrm{VI}} - \nu_{\widehat{\eta}_m})}{\tau_{\widehat{\eta}_m}}}_{\text{Term 3}}.
\end{align*}
The first term is $o_P(1)$ using the first condition of \Cref{cor:oracle-appc-calibration}, the scaling condition \Cref{eq:vi-post-convergence} of \Cref{prop:vi-appc-calibration} and \Cref{eq:scale-variational-posterior}. The second term is $o_P(1)$ using the centering condition \Cref{eq:vi-post-convergence} of \Cref{prop:vi-appc-calibration} and \Cref{eq:scale-variational-posterior}.  The third term becomes $o_P(1)$ from \Cref{eq:centering-error-variational-posterior} and as ${\tau_{\widehat{\eta}_m}}$ is bounded away from zero.

We therefore verify the centering and scale conditions in \Cref{cor:generic} for \Cref{prop:vi-appc-calibration},
\begin{align*}
    b(P_0 K_{\widehat\eta_m}) = \frac{\sqrt{N_{\new,m}} (\nu_0 - \nu_{\widehat{\eta}_m})}{\tau_{\widehat{\eta}_m}} \stackrel{P}{\longrightarrow}0, \qquad
s(P_0 K_{\widehat{\eta}_m}) = \frac{\tau_{\widehat{\eta}_m}}{\tau_0} \stackrel{P}{\longrightarrow}1.
\end{align*}
By \Cref{cor:generic}, $p_m(P_0 K_{\widehat{\eta}_m})\stackrel{d}{\longrightarrow} U[0,1]$.
\end{proof}

\subsection{Proof of \Cref{prop:emp-appc-calibration}}

\begin{proof}
We first define: For $\bm{x}, \bm{x}' \in \mathcal{X}$, 
\begin{align*}
    g_{\widehat{\eta}_m}(\bm x) \defeq \int h(\bm{x}')K_{\widehat{\eta}_m}(\bm{x},d\bm{x}'), \quad \text{ and } \quad g^{(2)}_{\widehat{\eta}_m}(\bm x) \defeq  \int h^2(\bm{x}')K_{\widehat{\eta}_m}(\bm{x},d\bm{x}').
\end{align*}
We also use the notation:
$PK_\eta h \defeq \mathbb{E}_{\bm{x}\sim PK_\eta}[h(\bm{x})]$ and $K_{\eta}h(\bm x) \defeq \int h(\bm {x}')K_\eta(\bm{x},d\bm{x}')$. 

Hence we can write
\begin{align*}
\nu_{\agg,\widehat{\eta}_m} 
&= \int\int h(\bm{x}')K_{\widehat{\eta}_m}(\bm{x},d\bm{x}')\,\widehat{P}_{\agg,m}(d\bm{x})\\
&= \int g_{\widehat\eta_m}(\bm{x})\widehat{P}_{\agg,m}(d\bm{x})\\
&= \frac{1}{N_{\agg,m}}\sum_{i=1}^{N_{\agg,m}} g_{\widehat{\eta}_m}(\bm{x}_{m,i}^\agg)
\end{align*}
where the last equality follows using \Cref{eq:definition-empirical-distribution}. The corresponding population quantity is $\E_{P_0}[g_{\widehat{\eta}_m}(\bm{x})]$.

Because the aggregation sample is independent of the training sample,
conditional on $\widehat\eta_m$,
\begin{align*}
\E_{P_0}\left[
    (\nu_{\agg,\widehat\eta_m}-\nu_{\widehat\eta_m}
    )^2 \mid
    \widehat\eta_m
\right] = {
    \Var_{P_0}
    \{g_{\widehat{\eta}_m}(\bm{x})\mid \widehat{\eta}_m\}
}/{N_{\agg,m}}.
\end{align*}
By Jensen's inequality, $[g_{\widehat\eta_m}(\bm{x})]^2 \leq g^{(2)}_{\widehat\eta_m}(\bm{x})$.
Therefore,
\begin{align*}
\mathbb{E}_{P_0}
\left[ \{g_{\widehat{\eta}_m}(\bm{x})\}^2
\mid \widehat{\eta}_m \right]
&\leq \mathbb{E}_{P_0}
\left[g^{(2)}_{\widehat\eta_m}(\bm{x})
\mid \widehat{\eta}_m\right]\\
&=\mathbb{E}_{\bm{x}^{\rep}\sim P_0 K_{\widehat{\eta}_m}}
\left[h^2(\bm{x}^{\rep})
\mid \widehat{\eta}_m \right]
\end{align*}
which is $O_P(1)$ under the fourth-moment assumption \Cref{eq:h4-emp-agg}. Hence, $\Var_{P_0}\left\{g_{\widehat{\eta}_m}(\bm{x})
\mid \widehat{\eta}_m \right\} = O_P(1).$ Consequently, we have by Chebyshev's inequality, $\nu_{\agg,\widehat\eta_m}-\nu_{\widehat\eta_m}
=
O_P(N_{\agg,m}^{-1/2}).$
Thus, the aggregation component of the centering condition is negligible whenever ${N_{\new,m}}/{N_{\agg,m}}
\longrightarrow0$, satisfied by the assumption of \Cref{prop:emp-appc-calibration}. That is,
\begin{align}
\sqrt{N_{\new,m}}\,
\lvert\nu_{\agg,\widehat\eta_m}-\nu_{\widehat\eta_m}\rvert
&=o_P(1).
\label{eq:centering-error-empirical-posterior}
\end{align}
For the scale component, we can write
\begin{align*}
    m_{2,\widehat{\eta}_m} &=  \mathbb{E}_{P_0 K_{\widehat{\eta}_m}}\left[h^2(\bm{x})\right] = \E_{P_0}\left[g^{(2)}_{\widehat{\eta}_m}(\bm{x})\right], \\ m_{2,\agg,\widehat{\eta}_m} &= \mathbb{E}_{\widehat{P}_{\agg,m}K_{\widehat{\eta}_m}}\left[h^2(\bm{x})\right] = \frac{1}{N_{\agg,m}}\sum_{i=1}^{N_{\agg,m}} g^{(2)}_{\widehat{\eta}_m}(\bm{x}_{m,i}^\agg).
\end{align*}
By the triangle inequality,
\begin{align*}
\lvert \tau^2_{\agg,\widehat{\eta}_m} - \tau^2_{\widehat{\eta}_m}\rvert
\leq \lvert m_{2,\agg,\widehat{\eta}_m} - m_{2,\widehat{\eta}_m}\rvert + \lvert
\nu_{\agg,\widehat{\eta}_m} - \nu_{\widehat{\eta}_m}\rvert
\lvert\nu_{\agg,\widehat{\eta}_m} +\nu_{\widehat{\eta}_m}\rvert.
\end{align*}
We now need to bound the first term in the RHS. We proceed with the same argument as before.

Conditional on $\widehat{\eta}_m$, $\mathbb{E}_{P_0}\left[
(m_{2,\agg,\widehat{\eta}_m}
- m_{2,\widehat{\eta}_m}
)^2 \mid \widehat{\eta}_m \right]
= {\Var_{P_0}\{g^{(2)}_{\widehat{\eta}_m}(\bm{x}) \mid\widehat{\eta}_m\}}/{N_{\agg,m}}.$

Using Jensen's inequality, 
\begin{align*}
\mathbb{E}_{P_0}
\left[ \{g^{(2)}_{\widehat{\eta}_m}(\bm{x})\}^2
\mid \widehat{\eta}_m \right]
&\leq \mathbb{E}_{P_0} \left[K_{\widehat{\eta}_m}h^4(\bm{x})
\mid \widehat{\eta}_m \right] \\
&=\mathbb{E}_{\bm{x}^{\rep}\sim P_0 K_{\widehat{\eta}_m}}
\left[h^4(\bm{x}^{\rep}) \mid \widehat{\eta}_m \right].
\end{align*}
Under the fourth-moment condition, this quantity is
$O_P(1)$. Therefore, $\Var_{P_0}\left\{g^{(2)}_{\widehat{\eta}_m}(\bm{x}) \mid
\widehat{\eta}_m\right\} = O_P(1),$
and Chebyshev's inequality gives $
    m_{2,\agg,\widehat{\eta}_m} - m_{2,\widehat{\eta}_m} = O_P(N_{\agg,m}^{-1/2}).$

Now, $\lvert\nu_{\widehat\eta_m}\rvert =O_P(1)$. As $N_{\agg,m}
\longrightarrow \infty$, $\lvert\nu_{\agg,\widehat\eta_m}+\nu_{\widehat\eta_m}\rvert = O_P(1)$ and
\begin{align}
\lvert\tau_{\agg, \widehat\eta_m}^2-\tau^2_{\widehat\eta_m}\rvert
= o_P(1).
\label{eq:scaling-error-empirical-posterior}
\end{align} 
Since $\tau_{\widehat\eta_m}$ is bounded away from zero, we have 
\begin{align}
\frac{\tau_{\agg, \widehat{\eta}_m}}{\tau_{\widehat\eta_m}} \stackrel{P}{\longrightarrow} 1.
\label{eq:scale-empirical-posterior}
\end{align} 
Combining with the second condition of \Cref{cor:oracle-appc-calibration}, \Cref{eq:scale-variational-posterior} and the scaling assumption of \Cref{prop:vi-appc-calibration},
\begin{align*}
    \frac{\tau_{\agg, \widehat{\eta}_m}}{\tau_0} = \frac{\tau_{\agg, \widehat{\eta}_m}}{\tau_{\widehat\eta_m}} \frac{\tau_{\widehat{\eta}_m}}{\tau^\star_{\mathrm{VI}}}\frac{\tau^\star_{\mathrm{VI}}}{\tau^\star_{\mathrm{post}}}\frac{\tau^\star_{\mathrm{post}}}{\tau_0} \stackrel{P}{\longrightarrow}1.
\end{align*}
Therefore, $\tau_{\agg, \widehat{\eta}_m}$ is also bounded away from zero with probability tending to 1.

The centering condition has the additional empirical approximation error, as in \Cref{eq:empirical-centering-decomposition}
\begin{align*}
      b(\widehat{P}_{\agg,m} K_{\widehat\eta_m}) &= \frac{\sqrt{N_{\new,m}} (\nu_0 - \nu_{\agg, \widehat{\eta}_m})}{\tau_{\agg, \widehat{\eta}_m}} \\
      &=  
      \begin{aligned}[t]
      \underbrace{\frac{\sqrt{N_{\new,m}} (\nu_0 - \nu^\star_{\mathrm{post}})}{\tau^\star_{\mathrm{post}}}\frac{\tau^\star_{\mathrm{post}}}{\tau^\star_{\mathrm{VI}}}\frac{\tau^\star_{\mathrm{VI}}}{\tau_{\widehat{\eta}_m}}\frac{\tau_{\widehat\eta_m}}{\tau_{\agg, \widehat{\eta}_m}}}_{\text{Term 1}} 
      &+ \underbrace{\frac{\sqrt{N_{\new,m}} ( \nu^\star_{\mathrm{post}} - \nu^\star_{\mathrm{VI}})}{\tau^\star_{\mathrm{VI}}}\frac{\tau^\star_{\mathrm{VI}}}{\tau_{\widehat{\eta}_m}}\frac{\tau_{\widehat\eta_m}}{\tau_{\agg, \widehat{\eta}_m}}}_{\text{Term 2}}
      \\
      + \underbrace{\frac{\sqrt{N_{\new,m}} (\nu^\star_{\mathrm{VI}} - \nu_{\widehat{\eta}_m})}{\tau_{\widehat{\eta}_m}}\frac{\tau_{\widehat\eta_m}}{\tau_{\agg, \widehat{\eta}_m}}}_{\text{Term 3}} 
      &+\underbrace{\frac{\sqrt{N_{\new,m}} ( \nu_{\widehat{\eta}_m} - \nu_{\agg, \widehat{\eta}_m})}{\tau_{\agg, \widehat{\eta}_m}}}_{\text{Term 4}}
      \end{aligned}
\end{align*}
The fourth term becomes $o_P(1)$ from \Cref{eq:centering-error-empirical-posterior} as ${\tau_{\agg, \widehat{\eta}_m}}$ is bounded away from zero. The rest of the quantities are also $o_P(1)$ as in \Cref{prop:vi-appc-calibration} and by \Cref{eq:scale-empirical-posterior}. 

We hence have the centering and scale conditions in \Cref{cor:generic} verified for \Cref{prop:emp-appc-calibration},
\begin{align*}
    b(\widehat{P}_{\agg,m} K_{\widehat\eta_m}) = \frac{\sqrt{N_{\new,m}} (\nu_0 - \nu_{\agg, \widehat{\eta}_m})}{\tau_{\agg, \widehat{\eta}_m}} \stackrel{P}{\longrightarrow}0, \qquad
s(\widehat{P}_{\agg,m}K_{\widehat{\eta}_m}) = \frac{\tau_{\agg,\widehat{\eta}_m}}{\tau_0} \stackrel{P}{\longrightarrow}1.
\end{align*}
By \Cref{cor:generic}, $p_m(\widehat{P}_{\agg,m} K_{\widehat{\eta}_m})\stackrel{d}{\longrightarrow} U[0,1]$.
\end{proof}

\subsection{Proof of \Cref{prop:ppca-three-reference-calibration}}

\textbf{Notation.} For a symmetric matrix $\bm A \in \mathbb{R}^{D \times D}$, the operator norm is defined as $\lVert\bm A\rVert_{\mathrm{op}} \defeq \max_{1\leq i\leq D}\lvert\lambda_i(\bm A)\rvert$ where $\lambda_i(\bm A)$ are the eigenvalues of $\bm A$. Furthermore, the smallest and largest eigenvalues of $\bm A$ are denoted by $\lambda_{\min}(\bm A)$ and $\lambda_{\max}(\bm A)$, respectively, whereas, $\lambda^{+}_{\min}(\bm A)$ denotes the smallest strictly positive eigenvalue of $\bm A$. The Frobenius norm of symmetric $\bm A$ is $\lVert \bm A \rVert_F = \sqrt{\operatorname{tr}(\bm A^\top \bm A)} = \sqrt{\sum_{i=1}^D\lambda_i^2(\bm A)}$.

For positive deterministic sequences $a_m$ and $b_m$, we write
$a_m\asymp b_m$ if there exist constants $0<c<C<\infty$ such that
\begin{align*}
c \, b_m\leq a_m\leq C \, b_m
\end{align*}
for all sufficiently large $m$. For nonnegative random sequences $A_m$, we write $A_m\asymp_P b_m$ if
\begin{align*}
    A_m/b_m=O_P(1)
\qquad\text{and}\qquad
b_m/A_m=O_P(1).
\end{align*}

Recall that we consider the true data generating distribution for the PPCA model with misspecified prior from \Cref{eq:ppca-ms-data-theory}. The true mean and variance of the latent factor distribution are $\bm m_Z
\defeq \E_G[\bm z]$ and $\bm\Lambda_m \defeq \Var_G(\bm z)$. Then the true population mean and variance are 
\begin{align*}
    \bm\mu_0\defeq\E_{P_0}[\bm x] = \bm{W}_0 \bm{m}_Z \qquad \text{and} \qquad \bm\Sigma_0\defeq\Var_{P_0}(\bm x) = \bm{W}_0\bm{\Lambda}_m\bm{W}_0^T + \sigma_0^2I_{D_m}.
\end{align*}
Before we proceed, we establish the following three results for this model. The proofs are available at \Cref{sec:additional-ppca-proofs}.
\begin{result}
\label{result:ppca-population-optimal}
Under \Cref{ass:ppca-regularity,ass:ppca-pervasive}, for any population-optimal $\theta^\star = (\bm{W}^\star, \sigma^{\star 2}) \in \mathcal{S}^\star_{\mathrm{post}}$,
\begin{align*}
    \sigma^{\star 2} = \sigma_0^2, \qquad \bm{W}^\star\bm{W}^{\star\top} = \bm{W}_0\left(\bm{\Lambda}_m+\bm{m}_Z\bm{m}_Z^\top\right)\bm{W}_0^\top, 
    \qquad \mathrm{col}\left(\bm{W}^\star\right) = \mathrm{col}\left(\bm{W}_0\right)
\end{align*}
where $\bm{W}_0$, $\sigma_0^2$ are the true parameters and $\bm m_Z = \E_G[\bm z]$, $\bm\Lambda_m = \Var_G(\bm z)$.
\end{result}

\vspace{1em}

\begin{result}
\label{result:order-of-eigenvalues}
    From the pervasiveness condition (\Cref{ass:ppca-pervasive}) and \Cref{result:ppca-population-optimal}, we have
    \begin{align*}
    \lambda_{\min}^{+}(\bm{W}^\star\bm{W}^{\star\top})\asymp D_m,\qquad\lambda_{\max}(\bm{W}^\star\bm{W}^{\star\top})\asymp D_m,
    \end{align*}
    where $\lambda^+_{\mathrm{min}}(\bm{W}^\star\bm{W}^{\star\top})$ is the smallest strictly positive eigenvalue and $\lambda_{\mathrm{max}}(\bm{W}^\star\bm{W}^{\star\top})$ is the largest eigenvalue of $\bm{W}^\star\bm{W}^{\star\top}$.
\end{result}

\vspace{1em}

\begin{result}
\label{result:ppca-training-eigenspace-rate}
Let
\begin{align*}
\widehat{\bm S}_{\train,m}
=\frac{1}{N_{\train,m}}
\sum_{i=1}^{N_{\train,m}}
\bm x_{m,i}^\train\bm x_{m,i}^{\train\top}.
\end{align*}
Let $\widehat{\bm P}_m$ be the orthogonal projector onto its leading $K$ eigenvectors and $\bm P_{\bm W^\star}$ denote the orthogonal projector
onto $\operatorname{col}(\bm W^\star)$. Suppose \Cref{result:ppca-population-optimal,result:order-of-eigenvalues} hold. Under
\Cref{ass:ppca-regularity,ass:ppca-pervasive}, as
$D_m,N_{\train,m}\to\infty$,
\begin{align}
\lVert\widehat{\bm P}_m-\bm P_{\bm W^\star}\rVert_{\mathrm{op}}
&=O_P\left(N_{\train,m}^{-1/2}\right),
\label{eq:ppca-davis-kahan-rate}\\
\widehat\sigma_m^2-\sigma_0^2&=o_P(1),
\qquad
\lambda_{\min}^+
\left(\widehat{\bm W}_m\widehat{\bm W}_m^\top\right)
\asymp_P D_m.
\label{eq:ppca-fitted-spectrum-rate}
\end{align}
Consequently,
\begin{align}
\lVert\widehat{\bm B}_m-\bm B^\star\rVert_{\mathrm{op}}
=O_P\left(D_m^{-1}+N_{\train,m}^{-1/2}\right).
\label{eq:ppca-estimation-error-B-bound}
\end{align}
\end{result}

Results 1 and 2 provide the bounds needed to control misspecification error in the APPC. Specifically, even under latent-prior misspecification, population-optimal PPCA recovers the true signal subspace, and pervasiveness makes $\bm B^\star$ approach the projector onto this
subspace. Result 3 provides the bounds needed to control parameter-estimation
error in the APPC: it establishes noise-variance consistency and
translates standard PCA/PPCA estimation bounds into control of
$\lVert\widehat{\bm B}_m-\bm B^\star\rVert_{\mathrm{op}}$.

For readability, we write $\bm a_m = \bm a$ and $D_m = D$.

\begin{proof}
    \subsubsection*{Proof of (i): PPCA with oracle APPC}
Throughout the proof, we use $\sigma^{\star2} = \sigma^2_0$ as established in \Cref{result:ppca-population-optimal}.

\textbf{Mean shrinkage condition.} We first prove the oracle centering condition, \Cref{eq:ppca-oracle-centering-condition}.

Let $\bm P_{\bm{W}_0}$ denote the orthogonal projection onto $\mathrm{col}(\bm{W}_0)$. As $\bm{\mu}_0=\bm{W}_0\bm m_Z\in\mathrm{col}(\bm{W}_0)$, and $\col(\bm{B}^\star) \subseteq \col(\bm{W}^\star) = \col(\bm{W}_0)$ from \Cref{result:ppca-population-optimal}, then $\bm P_{\bm W_0}\bm{\mu}_0 = \bm{\mu}_0$ and $\bm{B}^\star=\bm P_{\bm{W}_0}\bm{B}^\star = \bm{B}^\star \bm P_{\bm{W}_0}$. Hence, $(I_D-\bm{B}^\star)\bm{\mu}_0 = (\bm P_{\bm W_0}-\bm{B}^\star)\bm{\mu}_0$. 

Denote $\Delta \defeq \bm P_{\bm W_0}-\bm{B}^\star$.
Now,
\begin{align*}
    \left| \bm{a}^\top\left(I_D-\bm{B}^\star\right)\bm\mu_0 \right| = \lvert \bm{a}^\top\Delta\bm\mu_0 \rvert \underset{(i)}{\leq} \lVert \bm{a} \rVert \, \lVert \Delta \rVert_{\mathrm{op}} \, \lVert \bm\mu_0 \rVert \underset{(ii)}{=} O(D^{-1})
\end{align*}
where $(i)$ is by Cauchy-Schwarz inequality followed by standard operator norm inequality.

Now for $(ii)$, the assumption $D\lVert \bm{a}\rVert_2^2 = O(1)$ implies $\lVert \bm{a}\rVert = O(D^{-1/2})$. And, since by \Cref{ass:ppca-regularity}, $\lVert\bm m_Z\rVert=O(1)$, using the pervasive condition \Cref{ass:ppca-pervasive},
\begin{align}
\label{eq:mu-norm-val}
\lVert\bm{\mu}_0\rVert
\le
\lVert\bm{W}_0\rVert_{\mathrm{op}} \lVert\bm m_Z\rVert
\le
\sqrt{\lambda_{\max}(\bm{W}_0^\top \bm{W}_0)}\lVert\bm m_Z\rVert
=
O\left(D^{1/2}\right).
\end{align}
Finally, on signal subspace of $\bm W^\star$, the eigenvalues of $\bm{B}^\star$ are
${\lambda^\star_k}/{\left(\lambda^\star_k+\sigma_0^{2}\right)}$, for $1 \leq k \leq \mathrm{rank}(\bm{W}^\star)$,
where $\lambda^\star_k$ are the positive eigenvalues of $\bm{W}^\star\bm{W}^{\star\top}$. Hence we have
\begin{align}
\label{eq:delta-norm-val} 
\lVert\Delta\rVert_{\mathrm{op}} = \lVert P_{\bm{W}_0}-\bm{B}^\star\rVert_{\mathrm{op}}
= \max_k
\frac{\sigma_0^{2}}{\lambda^\star_k+\sigma_0^{2}}
\le
\frac{1}{1+\rho^\star}
\end{align}
where $\rho^\star = {\lambda_{\min}^+(\bm {W}^\star\bm {W}^{\star\top})}/{\sigma_0^{ 2}}$. By \Cref{ass:ppca-regularity}, $\sigma_0^2 = O(1)$ and from \Cref{result:order-of-eigenvalues} we obtain $\lVert\Delta\rVert_{\mathrm{op}} = O(D^{-1})$.

Hence, under $\sqrt{N_{\new,m}}/D\to 0$,
\begin{align*}
    \left| \sqrt{N_{\new,m}} \,\bm{a}^\top\left(I_D-\bm{B}^\star\right)\bm{\mu}_0 \right| = O\left(\frac{\sqrt{N_{\new,m}}}{D}\right) = o(1).
\end{align*}
Since $\bm{a}^\top \bm{V}_m^{\mathrm{oracle}}\bm{a}$ is bounded away from $0$, we have
\begin{align*}
\frac{\sqrt{N_{\new,m}}\,
\bm{a}^\top\left(I_D-\bm{B}^\star\right)\bm{\mu}_0
}{\sqrt{\bm{a}^\top \bm{V}_m^{\mathrm{oracle}}\bm{a}}
}\stackrel{P}{\longrightarrow} 0.
\end{align*}

\textbf{Variance matching condition.} We next prove the oracle scaling condition, \Cref{eq:ppca-oracle-scale-condition}.
    
Let $\bm{E}_0 \defeq \bm{W}_0\bm{\Lambda}_m\bm{W}_0^\top - \bm{W}^\star\bm{W}^{\star\top} \in \mathbb{R}^{D\times D}.$ Then,
$\bm\Sigma_0
= \bm{W}^\star\bm{W}^{\star\top}+\sigma_0^2I_D+\bm{E}_0$. 

By definition from \Cref{tab:ppca-reference-distributions},
\begin{align}
\bm{V}_m^{\mathrm{oracle}}
&=  \bm R^\star+\bm B^\star\bm\Sigma_0\bm B^\star \notag\\
&= \sigma_0^{2}\left(I_D+\bm{B}^\star\right)+\bm{B}^\star\left(\bm{W}^\star\bm{W}^{\star\top}+\sigma_0^{2}I_D+\bm{E}_0\right)\bm{B}^\star \notag\\
&\underset{(i)}{=} \bm{W}^\star\bm{W}^{\star\top} + \sigma_0^2I_D+\bm{B}^\star\bm{E}_0\bm{B}^\star \notag\\
& = \bm\Sigma_0 + \left[\bm{B}^\star\bm{E}_0\bm{B}^\star-\bm{E}_0\right]
\label{eq:ppca-oracle-var-rewrite}
\end{align}
where $(i)$ follows from the PPCA identity $\sigma_0^2\bm{B}^\star + \bm{B}^\star\left(\bm{W}^\star\bm{W}^{\star\top}\right)\bm{B}^\star + \sigma_0^2\bm{B}^{\star2} = \bm{W}^\star\bm{W}^{\star\top}$.

We shall bound the second term in the RHS of \Cref{eq:ppca-oracle-var-rewrite}. Note that $\bm{E}_0$ is supported on $\mathrm{col}(\bm{W}_0)$, hence $\bm{E}_0=\bm P_{\bm{W}_0}\bm{E}_0 \bm P_{\bm{W}_0}$. We can re-write $\bm{B}^\star\bm{E}_0\bm{B}^\star-\bm{E}_0
= \left(\bm P_{\bm{W}_0}-\Delta\right)\bm{E}_0\left(\bm P_{\bm{W}_0}-\Delta\right)-\bm{E}_0=\Delta\bm{E}_0\Delta
-\Delta\bm{E}_0-\bm{E}_0\Delta.$
where $\Delta = \bm P_{\bm W_0}-\bm{B}^\star$.
Thus,
\begin{align*}
\lVert\bm{B}^\star\bm{E}_0\bm{B}^\star-\bm{E}_0\rVert_{\mathrm{op}}
\underset{(i)}{\le}
\left(
2\lVert\Delta\rVert_{\mathrm{op}}+\lVert\Delta\rVert_{\mathrm{op}}^2
\right)\lVert\bm{E}_0\rVert_{\mathrm{op}}
\underset{(ii)}{=} O(1).
\end{align*}
where $(i)$ follows from triangle inequality and submultiplicity of operator norm. To obtain $(ii)$, from \Cref{eq:delta-norm-val} and \Cref{result:order-of-eigenvalues},
\begin{align*}
2\lVert\Delta\rVert_{\mathrm{op}}+\lVert\Delta\rVert_{\mathrm{op}}^2 \le \frac{2}{1+\rho^\star} + \frac{1}{(1+\rho^\star)^2} = O\left(D^{-1}\right).
\end{align*}
And, by definition of $\bm{E}_0$ and \Cref{result:ppca-population-optimal}, $\bm{E}_0 = -\bm{W}_0\bm{m}_Z\bm{m}_Z^\top\bm{W}_0^\top = -\bm{\mu}_0\bm{\mu}_0^\top$, hence from \Cref{eq:mu-norm-val}, $\lVert\bm{E}_0\rVert_{\mathrm{op}} = O(D)$.

Finally, using the assumption $\lVert \bm{a}\rVert = O(D^{-1/2})$, 
\begin{align*}
\left|
\bm{a}^\top\left(\bm{B}^\star\bm{E}_0\bm{B}^\star-\bm{E}_0\right)\bm{a}
\right|
\underset{(i)}{\le} \lVert\bm{B}^\star\bm{E}_0\bm{B}^\star-\bm{E}_0\rVert_{\mathrm{op}}\,\lVert\bm{a}\rVert^2
= O\left(D^{-1}\right)
\end{align*}
where $(i)$ follows from the standard quadratic-form bound for the symmetric matrix $\bm{B}^\star\bm{E}_0\bm{B}^\star-\bm{E}_0$.

Since $D \longrightarrow \infty$ and $\bm{a}^\top \bm\Sigma_0 \bm{a}$ is bounded away from $0$,
\begin{align*}
\frac{
\bm{a}^\top\left(\bm{B}^\star\bm{E}_0\bm{B}^\star-\bm{E}_0\right)\bm{a}
}{
\bm{a}^\top \bm\Sigma_0 \bm{a}
}
\stackrel{P}{\rightarrow}0
\end{align*}
completes the proof.

That is, we proved $b(P_0 K_{\theta^\star,\mathrm{post}}) \stackrel{P}{\longrightarrow} 0,\, s(P_0 K_{\theta^\star,\mathrm{post}}) \stackrel{P}{\longrightarrow} 1$ by substituting the first row of
\Cref{tab:ppca-reference-distributions} into \Cref{eq:ppca-common-b-s}. Hence, by \Cref{cor:generic}, the oracle \gls{APPC} for PPCA is calibrated.

\subsubsection*{Proof of (ii): PPCA with estimated APPC} We next prove calibration of \gls{APPC} with estimated parameters for PPCA. We need to show
\begin{align*}
\frac{\sqrt{N_{\new,m}}\,\bm{a}^\top
(\widehat{\bm B}_m-\bm B^\star)\bm\mu_0}{
\sqrt{\bm{a}^\top\bm{V}_m^{\mathrm{oracle}}\bm{a}}
}\stackrel{P}{\longrightarrow}0,
\qquad
\frac{\bm{a}^\top\bm{V}_m^{\mathrm{learn}}\bm{a}
}{\bm{a}^\top\bm{V}_m^{\mathrm{oracle}}\bm{a}}
&\stackrel{P}{\longrightarrow}1.
\end{align*}
By \Cref{result:ppca-training-eigenspace-rate},
$\lVert\widehat{\bm B}_m-\bm B^\star\rVert_{\mathrm{op}}
=O_P(r_{B,m})$, where
$r_{B,m}=D^{-1}+N_{\train,m}^{-1/2}$. Since
$\lVert\bm a\rVert=O(D^{-1/2})$ and
$\lVert\bm\mu_0\rVert=O(D^{1/2})$ from \Cref{eq:mu-norm-val},
\begin{align*}
\left| \sqrt{N_{\new,m}}\,\bm{a}^\top
(\widehat{\bm B}_m-\bm B^\star)\bm\mu_0 \right|
&\leq \sqrt{N_{\new,m}}\,\lVert \bm{a} \rVert
\lVert \widehat{\bm{B}}_m - \bm{B}^\star \rVert_{\mathrm{op}}
\lVert \bm\mu_0 \rVert\\
&=O_P\left(\frac{\sqrt{N_{\new,m}}}{D}
+\sqrt{\frac{N_{\new,m}}{N_{\train,m}}}\right).
\end{align*}
Under the assumptions $\sqrt{N_{\new,m}}/D\longrightarrow 0$ and ${N_{\new,m}}/{N_{\train,m}} \longrightarrow 0$, this quantity is $o_P(1)$. And, as $\bm{a}^\top \bm{V}_m^{\mathrm{oracle}}\bm{a}$ is bounded away from $0$, we have
\begin{align}
\frac{\sqrt{N_{\new,m}}\,
\bm{a}^\top\left(\widehat{\bm B}_m-\bm{B}^\star\right)\bm{\mu}_0
}{\sqrt{\bm{a}^\top \bm{V}_m^{\mathrm{oracle}}\bm{a}}
}\stackrel{P}{\longrightarrow} 0.
\label{eq:ppca-estimated-centering-condition}
\end{align}

For the scale condition, \Cref{tab:ppca-reference-distributions} gives
\begin{align*}
\bm V_m^{\mathrm{learn}}-\bm V_m^{\mathrm{oracle}}
&=\widehat{\bm R}_m-\bm R^\star
+(\widehat{\bm B}_m-\bm B^\star)\bm\Sigma_0\widehat{\bm B}_m
+\bm B^\star\bm\Sigma_0(\widehat{\bm B}_m-\bm B^\star).
\end{align*}
Moreover,
\begin{align*}
\widehat {\bm R}_m - \bm R^\star =
(\widehat\sigma_m^2 - \sigma_0^2)(I_D + \widehat {\bm B}_m)
+ \sigma_0^2(\widehat {\bm B}_m - {\bm B}^\star).
\end{align*}
Now, $\lVert\widehat{\bm B}_m\rVert_{\mathrm{op}},
\lVert\bm B^\star\rVert_{\mathrm{op}}\leq1$,
$\lVert\bm\Sigma_0\rVert_{\mathrm{op}}=O(D)$, and
$\lVert\bm a\rVert^2=O(D^{-1})$. Then by \Cref{result:ppca-training-eigenspace-rate},
$\lVert\widehat{\bm R}_m-\bm R^\star\rVert_{\mathrm{op}}=o_P(1)$. Finally,
\begin{align*}
\left|\bm a^\top
(\bm V_m^{\mathrm{learn}}-\bm V_m^{\mathrm{oracle}})\bm a\right|
&\leq \lVert\bm a\rVert^2
\left\{\lVert\widehat{\bm R}_m-\bm R^\star\rVert_{\mathrm{op}}
+\lVert\bm\Sigma_0\rVert_{\mathrm{op}}
\lVert\widehat{\bm B}_m-\bm B^\star\rVert_{\mathrm{op}}
(\lVert\widehat{\bm B}_m\rVert_{\mathrm{op}}+\lVert\bm B^\star\rVert_{\mathrm{op}})
\right\}\\
&=o_P(1)+O_P(r_{B,m})=o_P(1).
\end{align*}
Since $\bm{a}^\top\bm{V}_m^{\mathrm{oracle}}\bm{a}$ is bounded away from zero,
\begin{align*}
\frac{
\bm{a}^\top\left(\bm{V}_m^{\mathrm{learn}} - \bm{V}_m^{\mathrm{oracle}}\right)\bm{a}
}{
\bm{a}^\top\bm{V}_m^{\mathrm{oracle}}\bm{a}
}
&\stackrel{P}{\longrightarrow}0
\end{align*}
which gives us the scale condition
\begin{align}
    \frac{\bm{a}^\top\bm{V}_m^{\mathrm{learn}}\bm{a}
}{\bm{a}^\top\bm{V}_m^{\mathrm{oracle}}\bm{a}}
\stackrel{P}{\longrightarrow}1.
\label{eq:ppca-estimated-scale-condition}
\end{align}
Consequently, $\bm a^\top\bm V_m^\mathrm{learn}\bm a$ is bounded away from $0$.

Combining part (i), \Cref{eq:ppca-learned-centering}, and
\Cref{eq:ppca-estimated-centering-condition,eq:ppca-estimated-scale-condition}
gives
\begin{align*}
b(P_0K_{\widehat\theta_m,\mathrm{post}})\stackrel{P}{\longrightarrow}0,
\qquad
s(P_0K_{\widehat\theta_m,\mathrm{post}})\stackrel{P}{\longrightarrow}1.
\end{align*}
Thus the estimated-parameter
\gls{APPC} is calibrated by \Cref{cor:generic}.

\subsubsection*{Proof of (iii): PPCA with empirical APPC} We need to show:
\begin{align}
\frac{
\sqrt{N_{\new,m}}\,
\bm{a}^\top\widehat{\bm B}_m
(\overline{\bm x}_m^\agg-\bm\mu_0)
}{
\sqrt{\bm{a}^\top\bm{V}_m^{\mathrm{learn}}\bm{a}}
} \stackrel{P}{\longrightarrow} 0 \quad\text{and}\quad \frac{\bm{a}^\top\bm{V}_m^{\mathrm{learn}}\bm{a}}{\bm{a}^\top\bm{V}_m^{\mathrm{emp}}\bm{a}}\stackrel{P}{\longrightarrow}1.\label{eq:ppca-emp-appc-to-show}
\end{align}

The aggregation sample is independent of the training sample. Conditional on
$\widehat\theta_m$, the matrices $\widehat{\bm B}_m$ and
$\widehat{\bm R}_m$ are fixed, while
$\mathbb{E}(\overline{\bm x}_m^{\agg}-\bm\mu_0\mid\widehat\theta_m)=0$
and
$\Var(\overline{\bm x}_m^{\agg}\mid\widehat\theta_m)
=\bm\Sigma_0/N_{\agg,m}$.
Therefore, we have
\begin{align}
&\E\left[
\left\{
\frac{
\sqrt{N_{\new,m}}\,
\bm{a}^\top\widehat{\bm B}_m
(\overline{\bm x}_m^\agg-\bm\mu_0)
}{
\sqrt{\bm{a}^\top\bm{V}_m^{\mathrm{learn}}\bm{a}}
}
\right\}^2
\rvert
\widehat\theta_m
\right]
=
\frac{N_{\new,m}}{N_{\agg,m}}
\frac{
\bm{a}^\top\widehat{\bm B}_m
\bm\Sigma_0
\widehat{\bm B}_m\bm{a}
}{
\bm{a}^\top\bm{V}_m^{\mathrm{learn}}\bm{a}
}
\leq
\frac{N_{\new,m}}{N_{\agg,m}},
\label{eq:ppca-aggregation-rate-bound}
\end{align}
because
$\widehat{\bm{R}}_m \succeq 0$ and hence $\widehat{\bm B}_m\bm\Sigma_0\widehat{\bm B}_m
\preceq\bm{V}_m^{\mathrm{learn}}$.
Hence, by Chebyshev's inequality, the centering term is
$O_P\{(N_{\new,m}/N_{\agg,m})^{1/2}\}$. Then, \Cref{eq:ppca-aggregation-rate-bound} along with the condition ${N_{\new,m}}/{N_{\agg,m}}\longrightarrow0$ implies that the centering term is $o_P(1)$.

To prove the second term of \Cref{eq:ppca-emp-appc-to-show}, conditional on
$\widehat\theta_m$, define the i.i.d. scalar variables and their conditional
moments by
\begin{align*}
Y_{m,i}&=\bm a^\top\widehat{\bm B}_m\bm x_{m,i}^\agg,
&
\widehat{\mu}_{Y,m}&=\E[Y_{m,i}\mid\widehat\theta_m],
&
\widehat v_{Y,m}&=\Var(Y_{m,i}\mid\widehat\theta_m).
\end{align*}
Here
$\widehat v_{Y,m}=\bm a^\top\widehat{\bm B}_m\bm\Sigma_0
\widehat{\bm B}_m\bm a$, while
\begin{align*}
\bm a^\top\widehat{\bm B}_m\widehat{\bm\Sigma}_{\agg,m}
\widehat{\bm B}_m\bm a
=\frac1{N_{\agg,m}}\sum_{i=1}^{N_{\agg,m}}(Y_{m,i}-\overline Y_m)^2.
\end{align*}
Hence, $\bm{a}^\top
\left(
\bm{V}_m^{\mathrm{emp}}-\bm{V}_m^{\mathrm{learn}}
\right)
\bm{a}= \bm{a}^\top\widehat{\bm B}_m\left(\widehat{\bm\Sigma}_{\agg,m} -\bm\Sigma_0\right)\widehat{\bm B}_m\bm{a}$ is the difference between empirical and population variance of $Y_{m,i}$.
\begin{align*}
    \bm{a}^\top
\left(
\bm{V}_m^{\mathrm{emp}}-\bm{V}_m^{\mathrm{learn}}
\right)
\bm{a}&= \bm{a}^\top\widehat{\bm B}_m\left(\widehat{\bm\Sigma}_{\agg,m} -\bm\Sigma_0\right)\widehat{\bm B}_m\bm{a} \\
&= \underbrace{\left[\frac{1}{N_{\agg,m}}\sum_{i=1}^{N_{\agg,m}}
(Y_{m,i}-\widehat{\mu}_{Y,m})^2-\widehat v_{Y,m}\right]}_{\mathrm{Term \, 1}}
-\underbrace{(\overline Y_m-\widehat{\mu}_{Y,m})^2}_{\mathrm{Term \, 2}}.
\end{align*}
The fourth-moment condition in
\Cref{eq:ppca-oracle-fourth-moment}, together with conditional Jensen's inequality, implies
$\E_{P_0}[\lvert Y_{m,i}\rvert^4\mid\widehat\theta_m]=O_P(1)$. Consequently,
$\E_{P_0}\{(Y_{m,i}-\widehat{\mu}_{Y,m})^4\mid\widehat\theta_m\}
=O_P(1)$ and
$\Var\{(Y_{m,i}-\widehat{\mu}_{Y,m})^2\mid\widehat\theta_m\}=O_P(1)$.
Chebyshev's inequality therefore gives
$\mathrm{Term~1}=O_P(N_{\agg,m}^{-1/2})$.

The same fourth-moment bound gives
$\widehat v_{Y,m}=O_P(1)$, so
$\Var(\overline Y_m\mid\widehat\theta_m)=O_P(N_{\agg,m}^{-1})$
and $\mathrm{Term~2}=O_P(N_{\agg,m}^{-1})$.

Hence,
\begin{align*}
\bm{a}^\top
\left(
\bm{V}_m^{\mathrm{emp}}-\bm{V}_m^{\mathrm{learn}}
\right)
\bm{a}
=
O_P\left(N_{\agg,m}^{-1/2}\right)
=
o_P(1),
\end{align*}
Since $\bm a^\top\bm V_m^\mathrm{learn}\bm a$ is bounded away from zero,
the second condition in \Cref{eq:ppca-emp-appc-to-show} follows, and
$\bm a^\top\bm V_m^\mathrm{emp}\bm a$ is also bounded away from zero. Combining these two incremental bounds with part (ii) and
\Cref{eq:ppca-empirical-centering} gives
\begin{align*}
b(\widehat P_{\agg,m}K_{\widehat\theta_m,\mathrm{post}})\stackrel{P}{\longrightarrow}0,
\qquad
s(\widehat P_{\agg,m}K_{\widehat\theta_m,\mathrm{post}})\stackrel{P}{\longrightarrow}1.
\end{align*}
Calibration now follows from \Cref{cor:generic}.
\end{proof}

\subsection{Bayesian APPC}
\label{app:bayesian-appc}
In this subsection, we consider a regular, identifiable, finite-dimensional
model for the global parameter, with an exact conditional posterior for the
local latent variables. The result is restricted to the well-specified case, as in \citet{moran2024a}. 

\vspace{1em}

\begin{definition}[Bayesian aggregated posterior predictive check]
Let $\mathcal F_m=\sigma(\bm X_m^\train,\bm X_m^\agg)$. When the
global parameter is treated as random, the Bayesian \gls{APPC} $p$-value is
given by
\begin{align*}
p_m^{\mathrm{bayes}}
&= \int
\Pr\left\{
d_{N_{\new,m}}(\bm X_m^\rep)
\ge d_{N_{\new,m}}(\bm X_m^\new)
\mid \mathcal F_m,\bm X_m^\new,\theta
\right\}
\Pi(d\theta\mid\bm X_m^\train).
\end{align*}
Here the empirical aggregated posterior is
\begin{align*}
\widehat q_{\theta,\agg}(d\bm z)
=\int p_\theta(d\bm z\mid\bm x)\,
\widehat P_{\agg,m}(d\bm x).
\end{align*}
Equivalently, for each replicated dataset, draw one shared global parameter
$\widetilde\theta\sim\Pi(\cdot\mid\bm X_m^\train)$ and then draw, independently
over $i$, conditional on $(\mathcal F_m,\widetilde\theta)$,
\begin{align*}
\bm z_{m,i}^\rep\sim\widehat q_{\widetilde\theta,\agg},
\qquad
\bm x_{m,i}^\rep\sim
p_{\widetilde\theta}(\cdot\mid\bm z_{m,i}^\rep),
\qquad i=1,\ldots,N_{\new,m}.
\end{align*}
\end{definition}

Conditional on $(\mathcal F_m,\theta)$, let $Q_{m,\theta}$ denote the
distribution of one replicate. For the empirical aggregated posterior,
$Q_{m,\theta}=\widehat P_{\agg,m}K_{\theta,\mathrm{post}}$. Define its
conditional diagnostic moments by
\begin{align*}
\nu_m(\theta)=\int h(\bm x)\,Q_{m,\theta}(d\bm x),
\qquad
\tau_m^2(\theta)=\int
\{h(\bm x)-\nu_m(\theta)\}^2Q_{m,\theta}(d\bm x).
\end{align*}

We also define the population reference quantities:
\begin{align*}
Q_\theta=P_0K_{\theta,\mathrm{post}},
\qquad
\nu(\theta)=\int h(\bm x)\,Q_\theta(d\bm x),
\qquad
\tau^2(\theta)=\Var_{Q_\theta}\{h(\bm x)\}.
\end{align*}

\begin{proposition}[Calibration of the Bayesian APPC]
\label{prop:bayesian-appc-calibration}  
Suppose $P_0=P_{\theta^\star}$ for an interior, locally identifiable
$\theta^\star\in\mathbb R^p$, where $p$ is fixed, and suppose the holdout CLT
in \Cref{ass:diag-clt} holds. Let $I(\theta^\star)$ denote the per-observation
Fisher information, assumed positive definite. Suppose also that:
\begin{enumerate}
    \item For some fixed neighborhood $U$ of $\theta^\star$, the conditional
    reference CLT holds uniformly in $\theta\in U$:
    \begin{align*}
    \sup_{\theta\in U}\sup_{t\in\mathbb R}
    \left|
    \Pr\left\{
    \frac{\sqrt{N_{\new,m}}
    \{d_{N_{\new,m}}(\bm X_m^\rep)-\nu_m(\theta)\}}
    {\tau_m(\theta)}
    \le t
    \;\middle|\;\mathcal F_m,\theta
    \right\}-\Phi(t)
    \right|
    \stackrel{P}{\longrightarrow}0.
    \end{align*}
    \item The functions $\nu_m$ are continuously differentiable on $U$.
    Because $\nu_m$ and $\tau_m^2$ themselves vary with $m$, smoothness for
    each fixed $m$ does not by itself control the expansion as $m\to\infty$.
    We therefore assume that, for every fixed $M<\infty$, the following
    uniform smoothness conditions hold over the posterior concentration
    neighborhood:
    \begin{align*}
    \sup_{\substack{\theta\in U\\
    \sqrt{N_{\train,m}}\lVert\theta-\theta^\star\rVert\le M}}
    \lVert\dot\nu_m(\theta)-\dot\nu_m(\theta^\star)\rVert
    &\stackrel{P}{\longrightarrow}0,
    \\
    \sup_{\substack{\theta\in U\\
    \sqrt{N_{\train,m}}\lVert\theta-\theta^\star\rVert\le M}}
    \lvert\tau_m^2(\theta)-\tau_m^2(\theta^\star)\rvert
    &\stackrel{P}{\longrightarrow}0.
    \end{align*}
    Moreover, $\tau_m(\theta)>0$ on $U$ with probability tending to one,
    and
    \begin{align*}
    \dot\nu_m(\theta^\star)
    \stackrel{P}{\longrightarrow}\dot\nu(\theta^\star).
    \end{align*}
    \item For some center $\widehat\theta_m$, the posterior satisfies the
    Bernstein--von Mises approximation
    \begin{align*}
    \left\|
    \Pi\left(
    \sqrt{N_{\train,m}}(\theta-\widehat\theta_m)\in\,\cdot
    \mid\bm X_m^\train
    \right)
    -\mathcal N_p\{0,\,I(\theta^\star)^{-1}\}
    \right\|_{\mathrm{TV}}
    \stackrel{P}{\longrightarrow}0,
    \end{align*}
    and its center is asymptotically efficient:
    \begin{align*}
    \sqrt{N_{\train,m}}(\widehat\theta_m-\theta^\star)
    \stackrel{d}{\longrightarrow}
    \mathcal N_p\{0,I(\theta^\star)^{-1}\}.
    \end{align*}
\end{enumerate}
Assume that the split sizes satisfy
\begin{align*}
\kappa_m\defeq
\frac{N_{\new,m}}{N_{\train,m}}
\longrightarrow\kappa\in[0,\infty),
\qquad
N_{\agg,m}\longrightarrow\infty,
\qquad
\frac{N_{\new,m}}{N_{\agg,m}}\longrightarrow0.
\end{align*}
Then
\begin{align}
p_m^{\mathrm{bayes}}
&=1-\Phi(T_m)+o_P(1),
\notag\\
V_m^{\mathrm{bayes}}
&=\tau_m^2(\theta^\star)
+\kappa_m\dot\nu_m(\theta^\star)^T
I(\theta^\star)^{-1}\dot\nu_m(\theta^\star),
\notag\\
T_m
&=\frac{1}{\sqrt{V_m^{\mathrm{bayes}}}}
\begin{aligned}[t]
\big[&\sqrt{N_{\new,m}}
\{d_{N_{\new,m}}(\bm X_m^\new)-\nu_0\}\\
&+\sqrt{N_{\new,m}}\{\nu_0-\nu_m(\theta^\star)\}\\
&-\sqrt{N_{\new,m}}\,
\dot\nu_m(\theta^\star)^T
(\widehat\theta_m-\theta^\star)\big],
\end{aligned}
\label{eq:bayesian-appc-standardized-statistic}
\end{align}
and $T_m\stackrel{d}{\longrightarrow}\mathcal{N}(0,1)$. Consequently,
$p_m^{\mathrm{bayes}}\stackrel{d}{\longrightarrow}U(0,1)$.
\end{proposition}

\begin{proof}
Under the joint construction in the definition, let
\begin{align*}
W_m=\sqrt{N_{\new,m}}
\{d_{N_{\new,m}}(\bm X_m^\rep)
-d_{N_{\new,m}}(\bm X_m^\new)\}.
\end{align*}
Then $p_m^{\mathrm{bayes}}
=\Pr(W_m\ge0\mid\mathcal F_m,\bm X_m^\new)$ exactly.

Adding and subtracting $\nu_m(\theta)$, $\nu_m(\theta^\star)$, and $\nu_0$
gives the exact decomposition
\begin{align}
W_m
={}&\sqrt{N_{\new,m}}
\{d_{N_{\new,m}}(\bm X_m^\rep)-\nu_m(\theta)\}
\notag\\
&+\sqrt{N_{\new,m}}
\{\nu_m(\theta)-\nu_m(\theta^\star)\}
\notag\\
&+\sqrt{N_{\new,m}}
\{\nu_m(\theta^\star)-\nu_0\}
\notag\\
&-\sqrt{N_{\new,m}}
\{d_{N_{\new,m}}(\bm X_m^\new)-\nu_0\}.
\label{eq:bayesian-w-decomposition}
\end{align}
The four terms isolate, respectively, replicate sampling variation, variation
due to the posterior draw of $\theta$, aggregation error, and holdout sampling
variation. To see the third interpretation explicitly, decompose
\begin{align}
\nu_m(\theta^\star)-\nu_0
={}&\underbrace{\nu_m(\theta^\star)-\nu(\theta^\star)}_
{\text{aggregation error}}
+\underbrace{\nu(\theta^\star)-\nu_0}_
{\text{population reference error}}.
\label{eq:bayesian-aggregation-decomposition}
\end{align}
The second term is zero by the well-specified assumption. The first satisfies
\begin{align*}
\sqrt{N_{\new,m}}
\{\nu_m(\theta^\star)-\nu(\theta^\star)\}
={}&\left(\frac{N_{\new,m}}{N_{\agg,m}}\right)^{1/2}
\sqrt{N_{\agg,m}}
\{\nu_m(\theta^\star)-\nu(\theta^\star)\}.
\end{align*}
Now, 
\begin{align*}
\sqrt{N_{\agg,m}}
\{\nu_m(\theta^\star)-\nu(\theta^\star)\} = O_P(1)
\end{align*}
using the same Jensen plus Chebyshev argument of \Cref{prop:emp-appc-calibration}.  So, the first term of \Cref{eq:bayesian-aggregation-decomposition} is $o_P(1)$ because
$N_{\new,m}/N_{\agg,m}\to0$. We also have
$\tau_m^2(\theta^\star) - \tau^2(\theta^\star)=o_P(1)$ by the law of large numbers. 

The Bernstein--von Mises approximation and the limit distribution of its
center imply
$\theta-\theta^\star=O_P(N_{\train,m}^{-1/2})$.
For any $\epsilon>0$, we have a fixed $M$ such that
$\sqrt{N_{\train,m}}\|\theta-\theta^\star\|\le M$ with probability at least
$1-\epsilon$ asymptotically. On this event, the mean-value theorem and uniform
smoothness over the corresponding neighborhood give the posterior delta-method
expansion
\begin{align*}
\nu_m(\theta)-\nu_m(\theta^\star)
=\dot\nu_m(\theta^\star)^T(\theta-\theta^\star)
+o_P\left(N_{\train,m}^{-1/2}\right).
\end{align*}
Because $\kappa_m=O(1)$, the last term remains $o_P(1)$ after
multiplication by $\sqrt{N_{\new,m}}$. Notice that the linear term is
generally $O_P(1)$ when $\kappa>0$ and is not discarded. Similarly,
uniform smoothness of the variance over this neighborhood gives
$\tau_m^2(\theta)-\tau_m^2(\theta^\star)=o_P(1)$.

The uniform conditional CLT and the total-variation Bernstein--von Mises
approximation now imply that, conditional on
$(\mathcal F_m,\bm X_m^\new)$, $W_m$ is asymptotically Gaussian with mean
\begin{align*}
\mu_{W,m}
={}&\sqrt{N_{\new,m}}\{\nu_m(\theta^\star)-\nu_0\}
+\sqrt{N_{\new,m}}\,
\dot\nu_m(\theta^\star)^T(\widehat\theta_m-\theta^\star)\\
&-\sqrt{N_{\new,m}}
\{d_{N_{\new,m}}(\bm X_m^\new)-\nu_0\},
\end{align*}
and variance
\begin{align*}
V_{W,m}
=\tau_m^2(\theta^\star)
+\kappa_m\dot\nu_m(\theta^\star)^T
I(\theta^\star)^{-1}\dot\nu_m(\theta^\star).
\end{align*}
It follows that
\begin{align*}
p_m^{\mathrm{bayes}}
=1-\Phi\left(-\frac{\mu_{W,m}}{\sqrt{V_{W,m}}}\right)+o_P(1)
=1-\Phi(T_m)+o_P(1),
\end{align*}
with $T_m$ as in
\Cref{eq:bayesian-appc-standardized-statistic}.

It remains to obtain the unconditional limit of $T_m$. By the holdout CLT,
\begin{align*}
\sqrt{N_{\new,m}}
\{d_{N_{\new,m}}(\bm X_m^\new)-\nu_0\}
\stackrel{d}{\longrightarrow} \mathcal{N}(0,\,\tau_0^2).
\end{align*}
Also, by asymptotic efficiency of the posterior center,
\begin{align*}
\sqrt{N_{\new,m}}\,
\dot\nu_m(\theta^\star)^T(\widehat\theta_m-\theta^\star)
\stackrel{d}{\longrightarrow} \mathcal{N}\left(0,\,\kappa\dot\nu(\theta^\star)^T
I(\theta^\star)^{-1}\dot\nu(\theta^\star)\right).
\end{align*}
The holdout sample is independent of $\mathcal F_m$, so these two limits are
independent. The aggregation calculation following
\Cref{eq:bayesian-aggregation-decomposition} makes the remaining numerator
term $o_P(1)$ and gives
\begin{align*}
V_{W,m}\stackrel{P}{\longrightarrow}
\tau_0^2+\kappa\dot\nu(\theta^\star)^T
I(\theta^\star)^{-1}\dot\nu(\theta^\star).
\end{align*}
Consequently, $T_m\stackrel{d}{\longrightarrow} \mathcal{N}(0,1)$ by Slutsky's
theorem, and the probability integral transform gives
$p_m^{\mathrm{bayes}}\stackrel{d}{\longrightarrow}U(0,1)$. In particular,
calibration allows $N_{\new,m}/N_{\train,m}\to\kappa>0$ because posterior
uncertainty supplies the corresponding variance term. In contrast,
$\widehat q_{\theta,\agg}$ is a plug-in empirical distribution, so its
uncertainty is negligible here because
$N_{\new,m}/N_{\agg,m}\to0$.
\end{proof}

\subsection{Proof of additional PPCA results}\label{sec:additional-ppca-proofs}

In the proof of the following results, for readability, we write $\bm a_m = \bm a$ and $D_m = D$.

\subsubsection{Proof of \Cref{result:ppca-population-optimal}}

\begin{proof}
    Recall the fitted PPCA model:
    \begin{align*}
    \bm z_i &\sim \mathcal{N}\left(0,\,I_K\right),
    \qquad
    \bm x_i\mid\bm z_i \sim \mathcal{N}\left(\bm W\bm z_i,\,\sigma^2I_D\right)
    \\
    \bm x_i &\sim \mathcal{N}\left(0, \,\bm C_\theta\right),
    \qquad \bm C_\theta \defeq \bm W\bm{W}^\top + \sigma^2I_D.
    \end{align*}
The population log-likelihood is
\begin{align*}
    \mathcal{L}_{\mathrm{post}}(\theta) = -\frac{1}{2}\left[\log \det \bm C_{\theta} + \operatorname{tr}\left(\bm C_{\theta}^{-1} \bm S_0\right)\right] + \text{constant}, 
    \quad
    \text{where} \quad
    \bm S_0 \defeq \mathbb{E}_{P_0}\left[\bm x_i \bm{x}_i^\top\right].
\end{align*}
Over all covariance matrices $\bm C \succ 0$, $\mathcal{L}_{\mathrm{post}}(\theta)$ is uniquely maximized at $\bm C = \bm S_0$.

Under the true data-generating model, \Cref{eq:ppca-ms-data-theory}, 
\begin{align*}
    \bm S_0 = \mathbb{E}_{P_0}\left[\bm x_i \bm{x}_i^\top\right] = \bm{W}_0\left(\bm{\Lambda}_m+\bm{m}_Z\bm{m}_Z^\top\right)\bm{W}_0^\top + \sigma_0^2I_D.
\end{align*}

Let $\bm{H}_m \defeq \E_{G}[\bm z_i\bm z_i^\top] = \bm{\Lambda}_m + \bm{m}_Z\bm{m}_Z^\top$.
Since $\bm\Lambda_m \succ 0$, $\bm{H}_m \succ 0$, we can rewrite 
\begin{align*}
    \bm S_0 = \left(\bm{W}_0 \bm{H}_m^{1/2}\right) \left(\bm{W}_0 \bm{H}_m^{1/2}\right)^\top + \sigma_0^2I_D,
\end{align*} so $\bm S_0$ belongs to the PPCA covariance class. Hence for any $\theta^\star \in \mathcal{S}^\star_{\mathrm{post}}$, $\bm C_{\theta^\star} = \bm S_0$.  Thus,
\begin{align*}
\bm{W}^\star\bm{W}^{\star\top} + \sigma^{\star 2}I_D = \bm{W}_0\bm{H}_m\bm{W}_0^\top + \sigma_0^2 I_D.
\end{align*}
Furthermore, $K < D$, $\bm H_m \succ 0$, $\mathrm{rank}(\bm W_0) = K$ and $\mathrm{rank}(\bm W^\star \bm W^{\star \top}) \leq K$. Hence, $\sigma^{\star 2} = \lambda_{\min}(\bm C_{\theta^\star}) = \lambda_{\min}(\bm S_0) = \sigma_0^2$
which completes the proof of the first part.

For the next part, note that from the regularity assumptions, \Cref{ass:ppca-regularity}, the eigenvalues of $\bm{\Lambda}_m$ are uniformly bounded and $\|\bm m_Z\|=O(1)$. Also, $\bm \Lambda_m \succ 0$ and $\bm{m}_Z\bm{m}_Z^\top \succeq 0$ yields
\begin{align}
\label{eq:bound-eigenvalues-H}
    0 < c_{\bm \Lambda} \leq \lambda_{\min}(\bm \Lambda_m) &\leq \lambda_{\min}(\bm H_m) \notag\\
    &\leq \lambda_{\max}(\bm H_m) \leq \lambda_{\max}(\bm \Lambda_m)  + \lVert \bm{m}_Z \rVert ^2 \leq C_{\bm \Lambda} + \lVert \bm{m}_Z \rVert ^2 < \infty.
\end{align}
The matrix $\bm{H}_m$ is symmetric positive definite by \Cref{ass:ppca-regularity} and since $\lambda_{\min}(\bm \Lambda_m) > 0$, $\bm{H}_m$ is invertible. Now,
\begin{align*}
    \col\left(\bm{W}^\star\right) = \col\left(\bm{W}_0 \bm{H}_m \bm{W}_0^\top\right)
    = \col\left(\bm{W}_0 \bm{H}_m^{1/2} (\bm{W}_0 \bm{H}_m^{1/2})^\top\right) = \col\left(\bm{W}_0\right).
\end{align*}
\end{proof}

\subsubsection{Proof of \Cref{result:order-of-eigenvalues}}

\begin{proof}
From \Cref{ass:ppca-pervasive} since $K$ is fixed and $\bm{Q}_W$ is positive definite, all eigenvalues of $\bm{Q}_W$ are bounded away from $0$ and $\infty$. By continuity of extreme eigenvalues,
\begin{align*}
    \lambda_{\min}\left(\frac{1}{D}\bm{W}_0^\top\bm{W}_0\right) \longrightarrow \lambda_{\min}(\bm{Q}_W), \qquad
    \lambda_{\max}\left(\frac{1}{D}\bm{W}_0^\top\bm{W}_0\right) \longrightarrow \lambda_{\max}(\bm{Q}_W).
\end{align*}
For sufficiently large $D$, there exist constants $0 < c < C < \infty$ such that 
\begin{align*}
    cD \le \lambda_{\min}\left(\bm{W}_0^\top\bm{W}_0\right) \le \lambda_{\max}\left(\bm{W}_0^\top\bm{W}_0\right) \le CD.
\end{align*} 
Note that, $\bm{W}_0 \in \mathbb{R}^{D \times K}$ and $\bm{W}_0\bm{W}_0^{\top} \in \mathbb{R}^{D \times D}$ have rank $K < D$ for sufficiently large $D$, under the pervasiveness assumption. Furthermore, $\bm{W}_0^\top\bm{W}_0$ and $\bm{W}_0\bm{W}_0^\top$ have the same $K$ nonzero eigenvalues.
Therefore, 
\begin{align*}
    \lambda_{\min}^{+}(\bm{W}_0\bm{W}_0^{\top}) = \lambda_{\min}(\bm{W}_0^{\top}\bm{W}_0)\asymp D,\quad \lambda_{\max}(\bm{W}_0\bm{W}_0^{\top}) = \lambda_{\max}(\bm{W}_0^{\top}\bm{W}_0)\asymp D.
\end{align*}
The nonzero eigenvalues of $\bm{W}^\star\bm{W}^{\star\top} = \bm{W}_0 \bm{H}_m \bm{W}_0^\top$ (from \Cref{result:ppca-population-optimal}) are the eigenvalues of $\bm{H}_m^{1/2}\bm{W}_0^\top \bm{W}_0 \bm{H}_m^{1/2}$.
From \Cref{eq:bound-eigenvalues-H}  for some finite constant $C_{\bm{H}}$, $c_{\bm \Lambda} \leq \lambda_{\min}(\bm{H}_m) \leq \lambda_{\max}(\bm{H}_m) \leq C_{\bm{H}}$. Therefore,
\begin{align*}
    c_{\bm \Lambda}\lambda_{\min}\left(\bm{W}_0^\top \bm{W}_0\right) &\leq \lambda_{\min}\left(\bm{H}_m^{1/2}\bm{W}_0^\top \bm{W}_0 \bm{H}_m^{1/2}\right) \\
    &= \lambda_{\min}^+(\bm{W}^\star \bm{W}^{\star\top}) \leq
    \lambda_{\max}(\bm{W}^\star \bm{W}^{\star\top}) \le C_{\bm{H}} \lambda_{\max}(\bm{W}_0^\top \bm{W}_0).
\end{align*}
Hence, $\lambda_{\min}^{+}(\bm{W}^\star\bm{W}^{\star \top})\asymp D,\, \lambda_{\max}(\bm{W}^\star\bm{W}^{\star\top})\asymp D.$
\end{proof}

\subsubsection{Proof of
\Cref{result:ppca-training-eigenspace-rate}}

\begin{proof}
For readability, write $n=N_{\train,m}$ and $D=D_m$. 

Let $\bm H_m \defeq \E_G[\bm z_i\bm z_i^\top]
=\bm\Lambda_m+\bm m_Z\bm m_Z^\top$. The population second-moment matrix under the true data generating model, \Cref{eq:ppca-ms-data-theory} is
\begin{align*}
\bm S_0=\E_{P_0}\left[\bm x_i\bm x_i^\top\right]
=\bm W_0\bm H_m\bm W_0^\top+\sigma_0^2I_D,
\end{align*}
and its observable empirical counterpart satisfies
\begin{align*}
\widehat{\bm S}_{\train,m}-\bm S_0
={}&\frac1n\sum_{i=1}^n
\left\{\bm x_i\bm x_i^\top-\bm S_0\right\}.
\end{align*}
Pervasiveness, the bounded fourth moment of $\bm z_i$, and the Gaussian noise assumption imply $\E\lVert\bm x_i\rVert^4=O(D^2)$. Independence across observations therefore gives
\begin{align*}
\E\lVert\widehat{\bm S}_{\train,m}-\bm S_0\rVert_F^2
&=\frac1n
\E\lVert\bm x_i\bm x_i^\top-\bm S_0\rVert_F^2\\
&\leq\frac1n\E\lVert\bm x_i\bm x_i^\top\rVert_F^2
=\frac1n\E\lVert\bm x_i\rVert^4
=O\left(D^2/n\right).
\end{align*}
Since the operator norm is bounded by the Frobenius norm, Markov's inequality
then yields
\begin{align}
\lVert\widehat{\bm S}_{\train,m}-\bm S_0\rVert_{\mathrm{op}}
=O_P\left(D/\sqrt n\right).
\label{eq:ppca-training-covariance-rate}
\end{align}

Let $\bm U_m$ and $\widehat{\bm U}_m$ contain the leading $K$
orthonormal eigenvectors of $\bm S_0$ and
$\widehat{\bm S}_{\train,m}$, respectively. The population eigengap is
\begin{align*}
\delta_m
&=\lambda_K(\bm S_0)-\lambda_{K+1}(\bm S_0)\\
&=\lambda_{\min}^+(\bm W_0\bm H_m\bm W_0^\top)
=\lambda_{\min}^+(\bm W^\star\bm W^{\star\top})
\asymp D,
\end{align*}
where the last equality follows from
\Cref{result:ppca-population-optimal,result:order-of-eigenvalues}.
The Davis--Kahan $\sin\Theta$ theorem
\citep{davis1970rotation}, in the population-gap form of
\citet[Theorem~2]{yu2015useful}, gives
\begin{align*}
\lVert\widehat{\bm P}_m-\bm P_{\bm W^\star}\rVert_{\mathrm{op}}
&=\lVert\sin\Theta(\widehat{\bm U}_m,\bm U_m)\rVert_{\mathrm{op}}\\
&\leq
\lVert\sin\Theta(\widehat{\bm U}_m,\bm U_m)\rVert_F\\
&\leq
\frac{2\sqrt K\,
\lVert\widehat{\bm S}_{\train,m}-\bm S_0\rVert_{\mathrm{op}}}
{\delta_m}
=O_P(n^{-1/2}).
\end{align*}
The last equality uses \Cref{eq:ppca-training-covariance-rate} and the facts that $D,n\to\infty$ and $K$ is fixed. This proves
\Cref{eq:ppca-davis-kahan-rate}. Notice that this version of the theorem requires only the population eigengap; no separate sample-eigengap assumption is needed.

It remains to prove consistency of $\widehat\sigma_m^2$ and establish the
order of the positive eigenvalues of
$\widehat{\bm W}_m\widehat{\bm W}_m^\top$. Let
$\widehat\ell_1\geq\cdots\geq\widehat\ell_D$ be the eigenvalues of
$\widehat{\bm S}_{\train,m}$. The PPCA maximum-likelihood estimator satisfies
\begin{align*}
\widehat\sigma_m^2
=\frac{1}{D-K}\sum_{j=K+1}^D\widehat\ell_j.
\end{align*}
The matrix $\bm H_m$ is positive definite by
\Cref{ass:ppca-regularity}, and pervasiveness implies that $\bm W_0$ has
column rank $K$. Hence $\bm W_0\bm H_m\bm W_0^\top$ is positive semidefinite with rank $K$ and therefore has $D-K$ zero eigenvalues. Because
$\bm S_0=\bm W_0\bm H_m\bm W_0^\top+\sigma_0^2I_D$, adding
$\sigma_0^2I_D$ shifts every eigenvalue by $\sigma_0^2$. Thus,
\begin{align*}
\lambda_j(\bm S_0)=\sigma_0^2,
\qquad j=K+1,\ldots,D.
\end{align*}
The fourth-moment assumptions imply
\begin{align*}
\frac{\operatorname{tr}(\widehat{\bm S}_{\train,m}-\bm S_0)}D
=O_P(n^{-1/2}),
\end{align*}
because $\E\|\bm x_i\|^4=O(D^2)$. Weyl's inequality and
\Cref{eq:ppca-training-covariance-rate} also give
\begin{align*}
\frac1D\sum_{j=1}^K
|\widehat\ell_j-\lambda_j(\bm S_0)|
&\leq\frac KD
\|\widehat{\bm S}_{\train,m}-\bm S_0\|_{\mathrm{op}}\\
&=O_P(n^{-1/2})=o_P(1),
\end{align*}
where we use that $K$ is fixed. Using the trace identities for the sample and population eigenvalues now gives
\begin{align*}
\widehat\sigma_m^2-\sigma_0^2
={}&\frac{\operatorname{tr}
(\widehat{\bm S}_{\train,m}-\bm S_0)}{D-K}-\frac{1}{D-K}\sum_{j=1}^K
\left\{\widehat\ell_j-\lambda_j(\bm S_0)\right\}.
\end{align*}
Both terms are $O_P(n^{-1/2})$, since $D/(D-K)\to1$. Therefore,
$\widehat\sigma_m^2-\sigma_0^2=O_P(n^{-1/2})=o_P(1)$.
Weyl's inequality and \Cref{result:order-of-eigenvalues} then show that $\widehat\ell_K-\widehat\sigma_m^2\asymp_P D$. The positive eigenvalues of
$\widehat{\bm W}_m\widehat{\bm W}_m^\top$ are
$\widehat\ell_j-\widehat\sigma_m^2$, $j=1,\ldots,K$; hence its smallest
positive eigenvalue is of order $D$, proving
\Cref{eq:ppca-fitted-spectrum-rate}.

Finally, insert the two signal-space projectors:
\begin{align*}
\|\widehat{\bm B}_m-\bm B^\star\|_{\mathrm{op}}
\leq{}&
\|\widehat{\bm B}_m-\widehat{\bm P}_m\|_{\mathrm{op}}
+\|\widehat{\bm P}_m-\bm P_{\bm W^\star}\|_{\mathrm{op}}
+\|\bm P_{\bm W^\star}-\bm B^\star\|_{\mathrm{op}}.
\end{align*}
The first and third terms are $O_P(D^{-1})$ and $O(D^{-1})$,
respectively, because the corresponding positive eigenvalues of the loading covariance matrices are of order $D$ while the noise variances are of order one. The middle term is $O_P(n^{-1/2})$ by
\Cref{eq:ppca-davis-kahan-rate}. This proves
\Cref{eq:ppca-estimation-error-B-bound}.
\end{proof}

\section{Additional PPCA Experiments}\label{app:ppca-extra-experiment}

Here we provide an additional experiment for the PPCA example in \Cref{subsec:ppca}. Specifically, we re-run the experiments in \Cref{subsec:ppca} with an alternative diagnostic function, the centered fourth cumulant:
\begin{align*}
d(\bm{X}) = \frac{1}{n}\sum_{i=1}^n \left[ \bm{a}^\top\bm{x}_i - \bm{a}^\top\overline{\bm{x}}\right]^4 - 3 \left(\frac{1}{n}\sum_{i=1}^n \left[ \bm{a}^\top\bm{x}_i - \bm{a}^\top\overline{\bm{x}}\right]^2\right)^2. 
\end{align*}

For these experiments, we run centered PPCA.

\begin{figure}[!htbp]
\begin{subfigure}{0.48\textwidth}
    \centering
    \includegraphics[width=\linewidth]{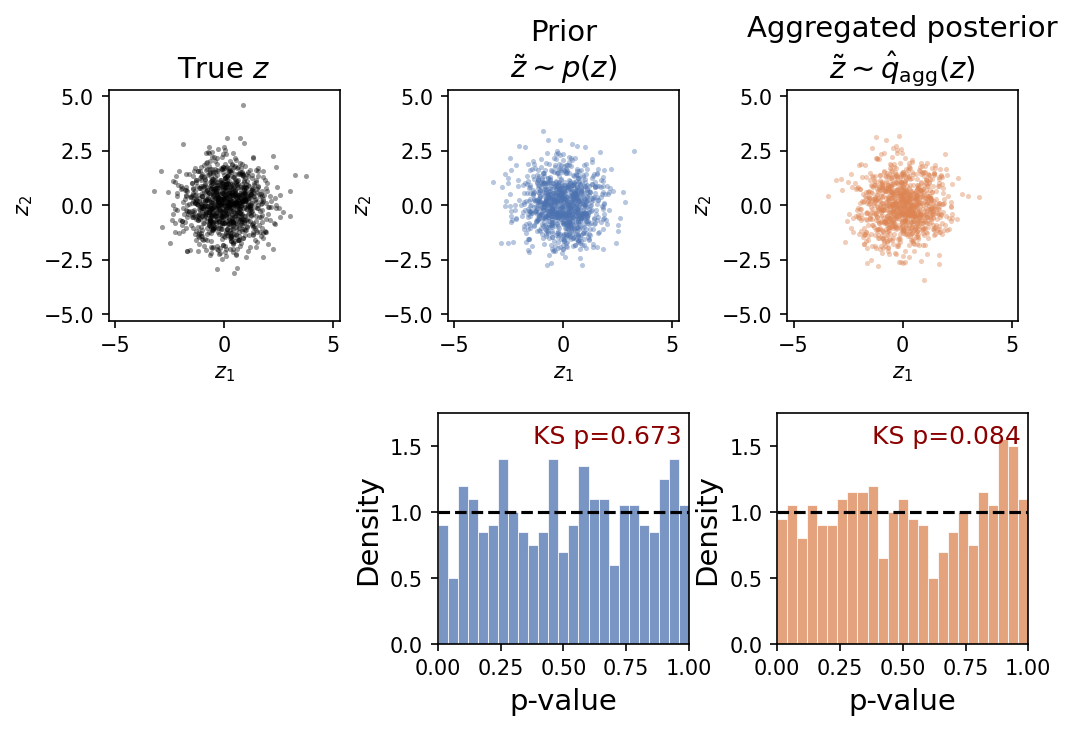}
     \caption{Well-specified case.}
    \label{fig:ppca_summary_ws_fourth}
\end{subfigure}
\begin{subfigure}{0.48\textwidth}
    \centering
    \includegraphics[width=\linewidth]{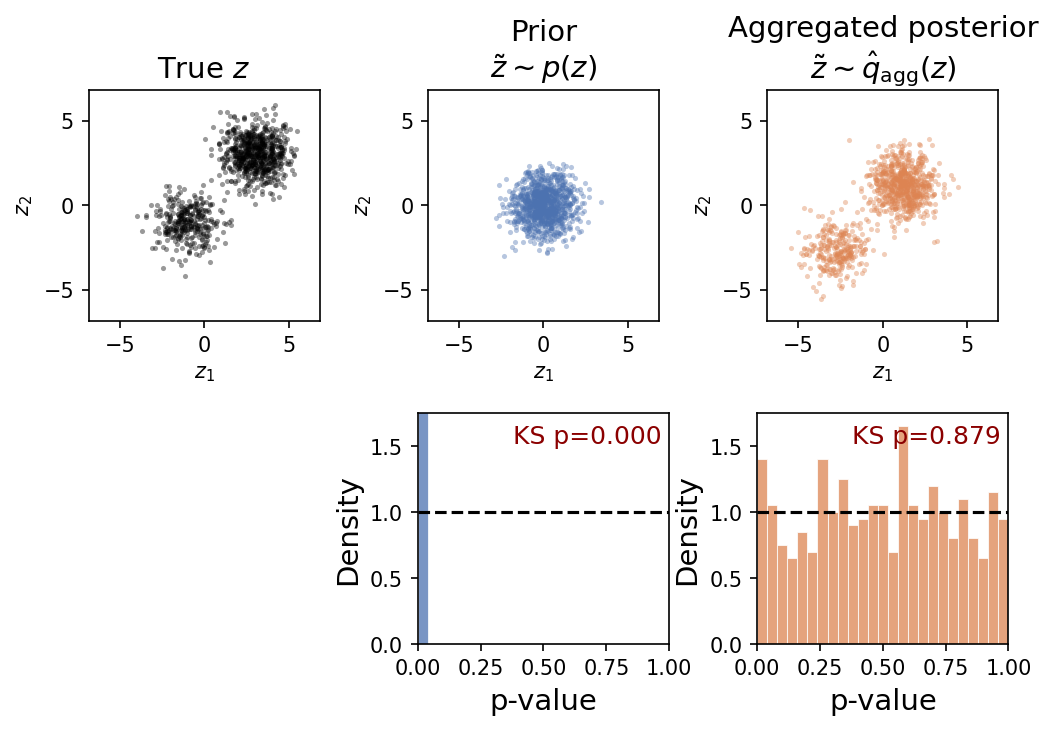}
     \caption{Misspecified case.}
    \label{fig:ppca_summary_ms_fourth}
\end{subfigure}
\caption{\textbf{PPCA example (fourth cumulant diagnostic).} The \gls{APPC} is calibrated; the empirical aggregated posterior sampling strategy generates realistic samples in both well-specified and misspecified cases. Top row: example draws of latent variables $\bm{z}$. Bottom row: $p$-values from predictive checks over 500 trials.}
\end{figure}

\section{Experiment Settings}

\subsection{VAE Details}\label{app:vae-details}

Variational autoencoders (VAEs) \citep{kingma2013auto,rezende2014stochastic} are latent variable generative models designed to learn nonlinear, low-dimensional representations of high-dimensional data. In a VAE, a latent prior and an observation model is specified,
\begin{align*}
\bm z_i\sim\mathcal N_K(\bm 0,I_K), \qquad
\bm x_i\mid\bm z_i\sim p_{\theta}(\bm x_i\mid\bm z_i),
\end{align*}
where the conditional distribution $p_{\theta}(\bm x_i\mid\bm z_i)$, referred to as the decoder, is typically parameterized by a neural network, such that for continuous data,
\begin{align*}
\bm x_i=f_{\theta}(\bm z_i)+\bm\varepsilon_i,
\qquad \bm\varepsilon_i\sim\mathcal N_D\left (0,\,I_D\right).
\end{align*}
Since the posterior distribution $p_\theta(\bm z_i \mid \bm x_i)$ is generally intractable, a VAE approximates it with a variational distribution
\begin{align*}
q_{\phi}(\bm z_i\mid\bm x_i)
= \mathcal N_K\left(
\mu_{\phi}(\bm x_i),\,
\operatorname{diag}\{\sigma_{\phi}^2(\bm x_i)\}
\right),
\end{align*}
where $\mu_{\phi}$ and $\sigma_{\phi}^2$ are typically parameterized by neural networks, referred as the encoder. The parameters $\{\theta,\phi\}$ are learned jointly by maximizing the evidence lower bound (ELBO),
\begin{align*}
\sum_{i=1}^N \left[ \mathbb E_{q_{\phi}(\bm z_i\mid\bm x_i)}
\{\log p_{\theta}(\bm x_i\mid\bm z_i)\}
- D_{\mathrm{KL}} \{q_{\phi}(\bm z_i\mid\bm x_i)\,\|\,p(\bm z_i)\}
\right].
\end{align*}
The first term encourages the latent representation to retain information needed to reconstruct the observations, while the second regularizes the approximate posterior toward the latent prior. Once fitted, the VAE provides both an approximate posterior representation of an observation through $q_{\widehat{\phi}}(\bm z\mid\bm x)$ and a generative mechanism for producing new observations through $p_{\widehat{\theta}}(\bm x\mid\bm z)$.

During training stage, for each observed data point $\bm x_i$, a latent variable is sampled using reparametrization trick:
\begin{align*}
    \bm z_i &\sim q_{\phi}(\bm z_i \mid \bm x_i),\\
    \bm z_i &= \mu_\phi(\bm x_i) + \sigma_\phi(\bm x_i) \odot \varepsilon_i, \qquad \varepsilon_i \sim \mathcal{N}_K\left(0,\,I_K\right)
\end{align*}

\subsection{VAE Datasets}\label{app:syn-datasets}

\textbf{Synthetic datasets:}
\begin{enumerate}
    \item Bivariate-$t_{2,2}$: We generate samples from a bivariate Student's-$t$ distributions with $2$ degrees of freedom, centered at $(3,-3)$. 
    \item Clustered-$t_{2,2,5}$: We generate samples from a $5$-component mixture of bivariate Student's-$t$ distributions with $2$ degrees of freedom. The component means are 
    \begin{align*}
        (-3,-3),\,(3,-3),\,(0,3.5),\,(-3,3),\,(3,3)
    \end{align*}
    and the component scale matrices are diagonal with respective parameters 
    \begin{align*}
        (0.16,0.36),\, (0.25,0.16),\, (0.36,0.25),\, (0.16,0.16),\, (0.25,0.25).
    \end{align*}
    This dataset provides a multimodal heavy-tailed setting with clustered structure.
    \item Latent Gaussian Mixture: We use a Gaussian mixture latent-variable model with $5$ bivariate Gaussian components. The latent mixture means are 
    \begin{align*}
        (-3.8,-3.6),\, (3.8,-3.4),\, (0,4.6),\, (-2.8,3.2),\, (3.2,3.4),
    \end{align*}
    with equal mixing probabilities. The covariance matrices are
    \begin{align*}
    \begin{pmatrix}
    0.55 & 0 \\
    0 & 0.55
    \end{pmatrix},
    \quad
    \begin{pmatrix}
    0.55 & 0 \\
    0 & 0.55
    \end{pmatrix},
    \quad
    \begin{pmatrix}
    0.85 & 0.18 \\
    0.18 & 0.28
    \end{pmatrix},
    \quad
    \begin{pmatrix}
    0.28 & 0 \\
    0 & 0.75
    \end{pmatrix},
    \quad
    \begin{pmatrix}
    0.60 & 0 \\
    0 & 0.30
    \end{pmatrix}.
    \end{align*} 
    Given a latent draw $\bm{z}$, the observations $\bm{x}$ are generated according to:
    \begin{align*}
    \bm{x} \mid \bm{z} \sim \mathcal{N}\left(\bm{A}\bm{z}+\bm{b},\,\Sigma_{\obs}\right):
    \quad
    \bm{A} =
    \begin{pmatrix}
    1.35 & 0.20\\
    0.08 & 1.15
    \end{pmatrix},
    \quad
    \bm{b}=\mathbf 0,
    \quad
    \Sigma_{\mathrm{obs}} =
    \begin{pmatrix}
    0.18 & 0\\
    0 & 0.18
    \end{pmatrix}
    \end{align*}
    This dataset yields a multimodal distribution with clusters exhibiting varying shapes and orientations. However, this dataset isolates multimodality without heavy-tailed within-cluster variation.

\end{enumerate}

\subsection{Synthetic data experiments}\label{app:syn-data-details}

\begin{itemize}
    \item {\bf Clustered $t$.} For numerical stability, we employed few model adjustments, such as, clipping log-variance parameters to $[-10,10]$; while for clustered-t VAE, bounded prior logits before the softmax transformation, and ensured the degrees-of-freedom parameter $\nu > 2.001$.
\end{itemize}

\subsubsection{Additional results: the latent GMM example}
\begin{figure}[!htbp]
    \centering
    \includegraphics[width=0.96\linewidth]{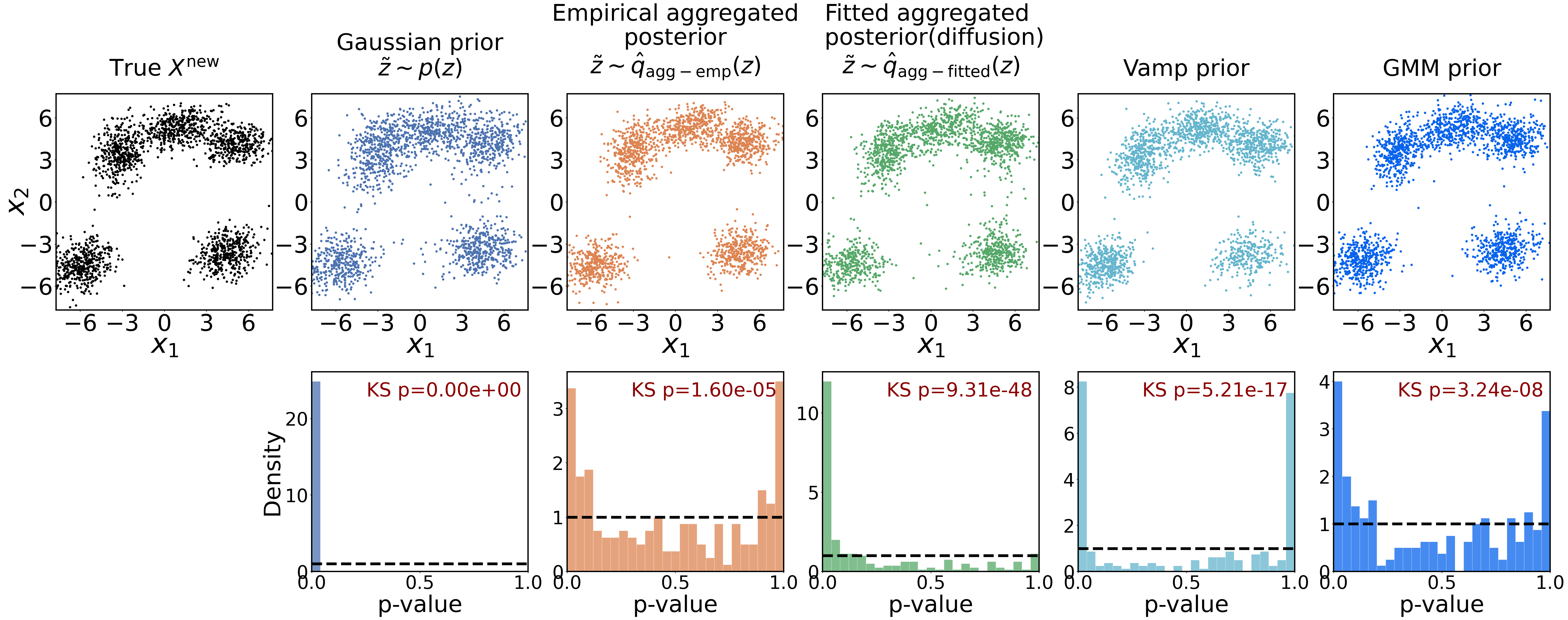}
    \caption{\textbf{GMM.} Top: Draws of data $\bm{x}$. Aggregated posterior with standard VAE (learnable variance, 700 epochs) performs better compared to VAEs with true or mixture-based priors. Bottom: $p$-values from model checks over $200$ trials.  Here, $N_\train = N_\agg = 10000, N_{\new} = 2000$.}
    \label{fig:gmm}
\end{figure}

\subsection{Real data experiments}\label{app:real-data-adjustments}

 We implemented the following modifications for the real datasets to improve the model in producing realistic samples.
\begin{enumerate}
\item For real datasets, as the latent dimensions used were very large, we use a weaker KL regularization by averaging the KL term across latent dimensions \citep{higgins2017betavae}. This mitigates over-regularization and posterior collapse, thereby allowing more flexible latent structure \citep{bowman2016generating}. Hence, the KL regularization term is 
\begin{align*}
\frac{1}{K}\sum_{j=1}^{K}D_{KL}\left(q_{\phi}(\bm{z}_j \mid \bm{x})\,\|\,p(\bm{z}_j)\right),
\end{align*}
where $K$ denotes the latent dimensions and $D_{KL}(.\|.)$ denotes the Kullback-Leibler Divergence between the two distributions $q_{\phi}(\bm{z}_j \mid \bm{x})$ and $p(\bm{z}_j)$.

\item We employ a tempered aggregated posterior for the $i^{th}$ sample by sampling
\begin{align*}
\bm{z}_i = \mu_{\bm{z}_i} + \tau \sigma_{\bm{z}_i} \odot \epsilon_i,
\qquad \epsilon_i \sim \mathcal{N}\left(0,I\right), 
\qquad 0 < \tau < 1
\end{align*}
where $\odot$ denotes elementwise multiplication and $\tau$ scales the posterior variance. Reduced temperature ($\tau$) controls excessive latent variance in high-dimensional latent spaces and reduce overly diffuse latent samples \citep{kingma2018glow, vahdat2020nvae}. For our model, we use $\tau = 0.5$. Consequently, the aggregated posterior is constructed using these tempered latent draws.

\item For both real and synthetic data experiments, to better approximate the aggregated latent distribution when fitting a latent density, we draw multiple posterior samples per data point instead of drawing a single latent sample \citep{mnih2016variational}. For $n$ data points with $M$ posterior samples, the latent sampling is 
\begin{align*}
    \bm{z}_{i}^{(m)} &\sim q_{\phi}(\bm{z} \mid \bm{x}_i), \qquad i=1,\ldots,n,\quad m=1,\ldots,M \\
    \bm{z}_{i}^{(m)} & = \mu_{\bm{z}_i} + \sigma_{\bm{z}_i} \odot \epsilon_i^{(m)}, \qquad \epsilon_i^{(m)} \sim \mathcal{N}\left(0,I\right).
\end{align*}
The latent training pool for latent density estimation is then
\begin{align*}
    \bm{z}^{\mathrm{train}}
= \left\{\bm{z}_i^{(m)} : i=1,\dots,N,\; m=1,\dots,M\right\}.
\end{align*}
This provides a richer approximation of the aggregated posterior distribution $q_{\mathrm{agg}}$ by incorporating posterior uncertainty, rather than relying only on the posterior mean. As a result, the model for density estimation (e.g. latent diffusion) is trained on samples that reflect the full posterior distribution $q_\phi(\bm{z}\mid\bm{x})$ and hence improves coverage of the latent space.

\end{enumerate}

\subsection{Parameters for VAE modeling}

We report the hyperparameters used to run the experiments from \Cref{sec:Experiments}.

The following table has the parameters of the VAE models for each data. We ran experiments with different values of epochs.  
\begin{table}[H]
\centering
\begin{tabularx}{\linewidth}{l X X X X X}
    \toprule
    \textbf{Hyperparameter} & 
     \textbf{Bivariate $t_2$} & 
     \textbf{Bivariate clustered $t_2$} &
     \textbf{Bivariate Gaussian mixture} &
     \textbf{Financial}  &
     \textbf{MNIST} \\
     \midrule
     Input Dimension   
     & 2 & 2 & 2 & 2 & 784 \\
     Latent Dimension  
     & 2 & 2 & 2 & 64 & 64 \\
     Hidden Dimension  
     & 128 & 128 & 128 & 128 & 400 \\
     \textbf{Epochs}
     & 700 & 700 & 700 & 1000 & 350 \\
     Batch Size        
     & 256 & 256 & 256 & 64 & 256 \\
     Learning Rate     
     & 0.001 & 0.001 & 0.001 & 0.0001 & 0.0001 \\
     \bottomrule
\end{tabularx}
\caption{Hyperparameter settings of VAE for the datasets}
\end{table}
The next table has parameters used for latent diffusion modeling (LDM) when estimating the aggregated posterior distribution $\widehat{q}_{\agg}(\bm{z}).$
 \begin{table}[H]
\centering
\begin{tabularx}{\linewidth}{l X X X X X}
    \toprule
    \textbf{Hyperparameter} & 
     \textbf{Bivariate $t_2$} & 
     \textbf{Bivariate clustered $t_2$} &
     \textbf{Bivariate Gaussian mixture} &
     \textbf{Financial}  &
     \textbf{MNIST} \\
     \midrule
     T & 100 & 100 & 100 & 100 & 100 \\
     Hidden Dimension  
     & 128 & 128 & 128 & 128 & 256 \\
     \textbf{Epochs}
     & 700 & 500 & 700 & 300 & 200 \\
     Batch Size        
     & 256 & 256 & 256 & 128 & 256 \\
     Learning Rate     
     & 0.001 & 0.001 & 0.001 & 0.0001 & 0.0001 \\
     \bottomrule
\end{tabularx}
\caption{Hyperparameter settings of LDM for the datasets}
\end{table}

 The hyperparameter values remain the same when using alternative VAE variants - stVAE \citep{24babff9-6288-42ab-a099-4272559768c3}, t3VAE \citep{kim2024t}, clustered t as the prior for VAE, and Vamp VAE \citep{tomczak2018vamp}. We report the additional parameter values that are model specific. A `x' indicates that the corresponding setting is not applicable for the given model or dataset.

 \begin{table}[H]
\centering
\begin{tabularx}{\linewidth}{l X X X X X}
    \toprule
    \textbf{Hyperparameter} & 
     \textbf{stVAE} & 
     \textbf{t3VAE} &
     \textbf{Clustered tVAE} &
     \textbf{VAE with GMM prior} & 
     \textbf{VampVAE} \\
     \midrule
     Degrees of freedom for $t$ & 3 & 3 & x & x & x \\
     MC KL samples & 5 & x & x & x & x \\
     Number of components & x & x & 5 & 5 & x \\
     Number of Pseduo Inputs & x & x & x & x & 64 \\
     \bottomrule
\end{tabularx}
\caption{Hyperparameter settings of models.}
\end{table}

\Cref{tab:p-val-settings} displays the experimental settings and model hyperparameters used to obtain the $p$-values for each data:
\begin{table}[H]
\centering
\begin{tabularx}{\linewidth}{l X X X X X}
    \toprule
    \textbf{Experimental Setting} & 
     \textbf{Bivariate $t_2$} & 
     \textbf{Bivariate clustered $t_2$} &
     \textbf{Bivariate Gaussian mixture} &
     \textbf{Financial}  &
     \textbf{MNIST} \\
     \midrule
     Training samples ($N_\train$)  
     & 10000  & 10000 & 10000 & 4000 & 10000 \\
     Aggregated samples ($N_\agg$)  
     & 10000 & 10000 & 10000 & 4000 & 10000 \\
     Replicate samples ($N_{\new}$) 
     & 2000 & 2000 & 2000 & 2000 & 2000 \\
     Number of trials ($N_{\mathrm{trial}}$)
     & 200  & 200 & 200 & 200 & 200  \\
     Number of repetitions ($B$)        
     & 100 & 100 & 100 & 100 & 100  \\[2mm]
     MC LDM samples    
     & 20 & 20 & 20 & 20 & 20  \\
     Temperature ($\tau$)    
     & x & x & x & 0.5 & 0.5 \\
     Diagnostic tail threshold ($t$)   
     & x & x & x & 1.5 & x \\
     \bottomrule
\end{tabularx}
\caption{Experimental settings for the datasets}
\label{tab:p-val-settings}
\end{table}

%% file: bib/references.bib
@inproceedings{24babff9-6288-42ab-a099-4272559768c3,
	title        = {{Variational auto-encoders with Student’s t-prior}},
	author       = {{Abiri, Najmeh and Ohlsson, Mattias}},
	year         = {{2019}},
	booktitle    = {{European Symposium on Artificial Neural Networks}},
	isbn         = {{978-287-587-065-0}},
	url          = {{https://www.elen.ucl.ac.be/Proceedings/esann/esannpdf/es2019-42.pdf}},
	language     = {{eng}}
}

@article{bai2003inferential,
	title        = {{Inferential theory for factor models of large dimensions}},
	author       = {Bai, Jushan},
	year         = 2003,
	journal      = {{Econometrica}},
	publisher    = {Wiley Online Library},
	volume       = 71,
	number       = 1,
	pages        = {135--171}
}

@article{bayarri2000p,
	title        = {{P values for composite null models}},
	author       = {Bayarri, MJ and Berger, James O},
	year         = 2000,
	journal      = {{Journal of the American Statistical Association}},
	publisher    = {Taylor \& Francis},
	volume       = 95,
	number       = 452,
	pages        = {1127--1142}
}

@inproceedings{bowman2016generating,
	title        = {{Generating sentences from a continuous space}},
	author       = {Bowman, Samuel and Vilnis, Luke and Vinyals, Oriol and Dai, Andrew and Jozefowicz, Rafal and Bengio, Samy},
	year         = 2016,
	booktitle    = {{Proceedings of the 20th SIGNLL conference on Computational Natural Language Learning}}
}

@inproceedings{dai2019diagnosing,
	title        = {{Diagnosing and Enhancing {VAE} Models}},
	author       = {Bin Dai and David Wipf},
	year         = 2019,
	booktitle    = {{International Conference on Learning Representations}},
	url          = {https://openreview.net/forum?id=B1e0X3C9tQ}
}

@article{efron2019bayes,
	title        = {{Bayes, oracle Bayes and empirical Bayes}},
	author       = {Efron, Bradley},
	year         = 2019,
	journal      = {{Statistical Science}},
	volume       = 34,
	number       = 2,
	pages        = {177--201}
}

@book{embrechts2013modelling,
	title        = {{Modelling extremal events: for insurance and finance}},
	author       = {Embrechts, Paul and Kl{\"u}ppelberg, Claudia and Mikosch, Thomas},
	year         = 2013,
	publisher    = {Springer Science \& Business Media},
	volume       = 33
}

@article{falck2021multi,
	title        = {{Multi-facet clustering variational autoencoders}},
	author       = {Falck, Fabian and Zhang, Haoting and Willetts, Matthew and Nicholson, George and Yau, Christopher and Holmes, Chris C},
	year         = 2021,
	journal      = {{Neural Information Processing Systems}}
}

@article{lecun2010mnist,
	title        = {{MNIST handwritten digit database}},
	author       = {LeCun, Yann and Cortes, Corinna and Burges, CJ},
	year         = 2010,
	journal      = {{ATT Labs [Online]. Available: http://yann.lecun.com/exdb/mnist}}
}

@article{shen2026simulation,
	title        = {{Simulation-Based Empirical Bayes}},
	author       = {Shen, Xinwei and Cai, Diana and Zhang, Cheng and Blei, David M},
	year         = 2026,
	journal      = {{arXiv preprint arXiv:2607.21843}}
}

@article{fan2013large,
	title        = {{Large covariance estimation by thresholding principal orthogonal complements}},
	author       = {Fan, Jianqing and Liao, Yuan and Mincheva, Martina},
	year         = 2013,
	journal      = {{Journal of the Royal Statistical Society Series B: Statistical Methodology}},
	publisher    = {Oxford University Press},
	volume       = 75,
	number       = 4,
	pages        = {603--680}
}

@article{yu2015useful,
	title        = {{A Useful Variant of the {Davis--Kahan} Theorem for Statisticians}},
	author       = {Yu, Yi and Wang, Tengyao and Samworth, Richard J.},
	year         = 2015,
	journal      = {{Biometrika}},
	volume       = 102,
	number       = 2,
	pages        = {315--323},
	doi          = {10.1093/biomet/asv008}
}

@article{fuest2026diffusion,
	title        = {{Diffusion models and representation learning: A survey}},
	author       = {Fuest, Michael and Ma, Pingchuan and Gui, Ming and Schusterbauer, Johannes and Hu, Vincent Tao and Ommer, Bj{\"o}rn},
	year         = 2026,
	journal      = {{IEEE Transactions on Pattern Analysis and Machine Intelligence}},
	publisher    = {IEEE}
}

@article{gelman1996posterior,
	title        = {{Posterior predictive assessment of model fitness via realized discrepancies}},
	author       = {Gelman, Andrew and Meng, Xiao-Li and Stern, Hal},
	year         = 1996,
	journal      = {{Statistica Sinica}},
	pages        = {733--760}
}

@article{heusel2017gans,
	title        = {{GANs trained by a two time-scale update rule converge to a local Nash equilibrium}},
	author       = {Heusel, Martin and Ramsauer, Hubert and Unterthiner, Thomas and Nessler, Bernhard and Hochreiter, Sepp},
	year         = 2017,
	journal      = {{Neural Information Processing Systems}}
}

@inproceedings{higgins2017betavae,
	title        = {{beta-{VAE}: Learning Basic Visual Concepts with a Constrained Variational Framework}},
	author       = {Irina Higgins and Loic Matthey and Arka Pal and Christopher Burgess and Xavier Glorot and Matthew Botvinick and Shakir Mohamed and Alexander Lerchner},
	year         = 2017,
	booktitle    = {{International Conference on Learning Representations}},
	url          = {https://openreview.net/forum?id=Sy2fzU9gl}
}

@article{ho2020denoising,
	title        = {{Denoising diffusion probabilistic models}},
	author       = {Ho, Jonathan and Jain, Ajay and Abbeel, Pieter},
	year         = 2020,
	journal      = {{Neural Information Processing Systems}}
}

@inproceedings{kim2024t,
	title        = {{$t3$-Variational Autoencoder: Learning Heavy-tailed Data with Student's t and Power Divergence}},
	author       = {Kim, Juno and Kwon, Jaehyuk and Cho, Mincheol and Lee, Hyunjong and Won, Joong-Ho},
	year         = 2024,
	booktitle    = {{International Conference on Learning Representations}}
}

@inproceedings{kingma2013auto,
	title        = {{Auto-encoding variational {B}ayes}},
	author       = {Kingma, Diederik P and Welling, Max},
	year         = 2014,
	booktitle    = {{International Conference on Learning Representations}}
}

@article{davis1970rotation,
	title        = {{The Rotation of Eigenvectors by a Perturbation. III}},
	author       = {Davis, Chandler and Kahan, W. M.},
	year         = 1970,
	journal      = {{SIAM Journal on Numerical Analysis}},
	volume       = 7,
	number       = 1,
	pages        = {1--46},
	doi          = {10.1137/0707001}
}

@article{kingma2018glow,
	title        = {{Glow: Generative flow with invertible 1x1 convolutions}},
	author       = {Kingma, Durk P and Dhariwal, Prafulla},
	year         = 2018,
	journal      = {{Neural Information Processing Systems}}
}

@inproceedings{li2021low,
	title        = {{On the low-density latent regions of VAE-based language models}},
	author       = {Li, Ruizhe and Peng, Xutan and Lin, Chenghua and Rong, Wenge and Chen, Zhigang},
	year         = 2021,
	booktitle    = {{NeurIPS 2020 Workshop on Pre-registration in Machine Learning}}
}

@article{li2026calibrated,
	title        = {{Calibrated model criticism using split predictive checks}},
	author       = {Li, Jiawei and Huggins, Jonathan H},
	year         = 2026,
	journal      = {{Journal of the American Statistical Association}},
	publisher    = {Taylor \& Francis},
	number       = {just-accepted},
	pages        = {1--21}
}

@inproceedings{locatello2020weakly,
	title        = {{Weakly-supervised disentanglement without compromises}},
	author       = {Locatello, Francesco and Poole, Ben and R{\"a}tsch, Gunnar and Sch{\"o}lkopf, Bernhard and Bachem, Olivier and Tschannen, Michael},
	year         = 2020,
	booktitle    = {{International Conference on Machine Learning}}
}

@inproceedings{mittal2023diffusion,
	title        = {{Diffusion based representation learning}},
	author       = {Mittal, Sarthak and Abstreiter, Korbinian and Bauer, Stefan and Sch{\"o}lkopf, Bernhard and Mehrjou, Arash},
	year         = 2023,
	booktitle    = {{International Conference on Machine Learning}}
}

@inproceedings{mnih2016variational,
	title        = {{Variational inference for Monte Carlo objectives}},
	author       = {Mnih, Andriy and Rezende, Danilo},
	year         = 2016,
	booktitle    = {{International Conference on Machine Learning}}
}

@article{moran2024a,
	title        = {{Holdout predictive checks for Bayesian model criticism}},
	author       = {Moran, Gemma E and Blei, David M and Ranganath, Rajesh},
	year         = 2024,
	journal      = {{Journal of the Royal Statistical Society Series B: Statistical Methodology}},
	volume       = 86,
	pages        = {194--214}
}

@inproceedings{ramesh2021zero,
	title        = {{Zero-shot text-to-image generation}},
	author       = {Ramesh, Aditya and Pavlov, Mikhail and Goh, Gabriel and Gray, Scott and Voss, Chelsea and Radford, Alec and Chen, Mark and Sutskever, Ilya},
	year         = 2021,
	booktitle    = {{International Conference on Machine Learning}}
}

@book{resnick2007heavy,
	title        = {{Heavy-tail phenomena: probabilistic and statistical modeling}},
	author       = {Resnick, Sidney I},
	year         = 2007,
	publisher    = {Springer}
}

@inproceedings{rezende2014stochastic,
	title        = {{Stochastic backpropagation and approximate inference in deep generative models}},
	author       = {Rezende, Danilo Jimenez and Mohamed, Shakir and Wierstra, Daan},
	year         = 2014,
	booktitle    = {{International Conference on Machine Learning}}
}

@inproceedings{rezende2015variational,
	title        = {{Variational inference with normalizing flows}},
	author       = {Rezende, Danilo and Mohamed, Shakir},
	year         = 2015,
	booktitle    = {{International Conference on Machine Learning}}
}

@article{rezende2018taming,
	title        = {{Taming VAEs}},
	author       = {Rezende, Danilo Jimenez and Viola, Fabio},
	year         = 2018,
	journal      = {{arXiv preprint arXiv:1810.00597}}
}

@article{robins2000asymptotic,
	title        = {{Asymptotic distribution of p values in composite null models}},
	author       = {Robins, James M and van der Vaart, Aad and Ventura, Val{\'e}rie},
	year         = 2000,
	journal      = {{Journal of the American Statistical Association}},
	publisher    = {Taylor \& Francis},
	volume       = 95,
	number       = 452,
	pages        = {1143--1156}
}

@inproceedings{rombach2022high,
	title        = {{High-resolution image synthesis with latent diffusion models}},
	author       = {Rombach, Robin and Blattmann, Andreas and Lorenz, Dominik and Esser, Patrick and Ommer, Bj{\"o}rn},
	year         = 2022,
	booktitle    = {{Computer Vision and Pattern Recognition}}
}

@article{Rubin:1984,
	title        = {{Bayesianly Justifiable and Relevant Frequency Calculations for the Applied Statistician}},
	author       = {Rubin, D.},
	year         = 1984,
	journal      = {{The Annals of Statistics}},
	volume       = 12,
	number       = 4,
	pages        = {1151--1172}
}

@article{salimans2016improved,
	title        = {{Improved techniques for training GANs}},
	author       = {Salimans, Tim and Goodfellow, Ian and Zaremba, Wojciech and Cheung, Vicki and Radford, Alec and Chen, Xi},
	year         = 2016,
	journal      = {{Neural Information Processing Systems}}
}

@article{seth2019model,
	title        = {{Model Criticism in Latent Space}},
	author       = {Seth, Sohan and Murray, Iain and Williams, Christopher KI},
	year         = 2019,
	journal      = {{Bayesian Analysis}},
	volume       = 14,
	number       = 3,
	pages        = {703--725}
}

@article{snell2017prototypical,
	title        = {{Prototypical networks for few-shot learning}},
	author       = {Snell, Jake and Swersky, Kevin and Zemel, Richard},
	year         = 2017,
	journal      = {{Neural Information Processing Systems}}
}

@inproceedings{sohl2015deep,
	title        = {{Deep unsupervised learning using nonequilibrium thermodynamics}},
	author       = {Sohl-Dickstein, Jascha and Weiss, Eric and Maheswaranathan, Niru and Ganguli, Surya},
	year         = 2015,
	booktitle    = {{International Conference on Machine Learning}}
}

@article{tam2025statistical,
	title        = {{On the statistical capacity of deep generative models}},
	author       = {Tam, Edric and Dunson, David B},
	year         = 2025,
	journal      = {{arXiv preprint arXiv:2501.07763}}
}

@inproceedings{Theis2016a,
	title        = {{A note on the evaluation of generative models}},
	author       = {L. Theis and A. van den Oord and M. Bethge},
	year         = 2016,
	month        = {Apr},
	booktitle    = {{International Conference on Learning Representations}},
	url          = {{http://arxiv.org/abs/1511.01844}}
}

@article{tipping1999probabilistic,
	title        = {{Probabilistic principal component analysis}},
	author       = {Tipping, Michael E and Bishop, Christopher M},
	year         = 1999,
	journal      = {{Journal of the Royal Statistical Society Series B: Statistical Methodology}},
	publisher    = {Oxford University Press},
	volume       = 61,
	number       = 3,
	pages        = {611--622}
}

@inproceedings{tomczak2018vamp,
	title        = {{VAE with a VampPrior}},
	author       = {Tomczak, Jakub and Welling, Max},
	year         = 2018,
	booktitle    = {{Artificial Intelligence and Statistics}},
	url          = {https://proceedings.mlr.press/v84/tomczak18a.html}
}

@article{vahdat2020nvae,
	title        = {{NVAE: A deep hierarchical variational autoencoder}},
	author       = {Vahdat, Arash and Kautz, Jan},
	year         = 2020,
	journal      = {{Neural Information Processing Systems}}
}

@article{van2017neural,
	title        = {{Neural Discrete Representation Learning}},
	author       = {van den Oord, Aaron and Vinyals, Oriol and Kavukcuoglu, Koray},
	year         = 2017,
	journal      = {{Neural Information Processing Systems}}
}
